\documentclass[11pt,a4paper]{article}
\usepackage{bbding}
\fussy
\usepackage{framed, xcolor}
\usepackage{tocvsec2}
\usepackage{datetime}
\usepackage{pifont}
\usepackage{pdflscape}
\usepackage{subfigure}          % for subfigures in ACM
\usepackage{colortbl}
\usepackage{booktabs}
\usepackage{pdfsync}
\usepackage[font=scriptsize,bf]{caption}
\usepackage{tikz,subfigure}
\usepackage[active]{srcltx}
\usepackage[margin=1in]{geometry}
\usepackage{algorithm}
\usepackage{algorithmic}
\usepackage{epsfig,amssymb,amsfonts,amsmath,amsthm}
\usepackage{multirow}
\usepackage[numbers,sort&compress,sectionbib]{natbib}
\usepackage{amsfonts}
\usepackage{xspace}
\usepackage{tabularx}
\usepackage{pstricks}
\usepackage{setspace}
\usepackage{xcolor}

\usepackage{rotating}
\usepackage{minitoc}

\usepackage{tikz}
\usepackage{enumerate}

\usepackage{hyperref}
\usetikzlibrary{chains,fit,shapes,arrows}
\usetikzlibrary{shapes.arrows}

\newcommand{\remove}[1]{}

\newcommand{\ce}{\mathrm{e}}

\newcommand{\gen}{\mathcal{G}}

\newcommand{\Oh}{O}

\renewcommand{\deg}{\mathrm{deg}}

\newcommand{\eps}{\epsilon}
\newcommand{\vol}{\operatorname{vol}}

\newcommand{\PRG}{\textsf{PRGs}}

\newcommand{\reg}{\mathrm{Reg}}

\renewcommand{\leq}{\leqslant}
\renewcommand{\geq}{\geqslant}
\renewcommand{\le}{\leqslant}
\renewcommand{\ge}{\geqslant}

\newcommand{\thmref}[1]{Theorem~\ref{thm:#1}}

\newcommand{\lemref}[1]{Lemma~\ref{lem:#1}}

\newcommand{\defref}[1]{Definition~\ref{def:#1}}

\newcommand{\secref}[1]{Section~\ref{sec:#1}}

\newcommand{\eq}[1]{\eqref{eq:#1}}

\newcommand{\Pro}[1]{\mathbf{Pr} \left[\,#1\,\right]}
\newcommand{\Prob}[2]{\mathbf{Pr}_{#1} \left[\,#2\,\right]}

\newcommand{\Ex}[1]{\mathbf{E} \left[\,#1\,\right]}

\newcommand{\EXX}[2]{\mathbf{E}_{#1} \left[\,#2\,\right]}

\newcommand{\Cov}[1]{\mathbf{Cov} \left[\,#1\,\right]}

\renewcommand{\tilde}{\widetilde}
\renewcommand{\epsilon}{\varepsilon}

\usepackage{setspace}

\newcommand{\mylemma}[2]{\begin{lem}\label{lem:#1}#2\end{lem}}

\newcommand{\againlemma}[2]{\noindent\textbf{Lemma~\ref{lem:#1}}
    (from page \pageref{lem:#1})\textbf{.}\emph{#2}}

\newtheorem{thm}{Theorem}  %[chapter]
\newtheorem{fact}[thm]{Fact}
\newtheorem{lem}[thm]{Lemma}

\newtheorem{cor}[thm]{Corollary}

\newtheorem{rem}[thm]{Remark}
\newtheorem{pro}[thm]{Proposition}

\newtheorem{protocol}{Protocol}
\newtheorem{defi}[thm]{Definition}

\renewcommand{\tilde}{\widetilde}

\numberwithin{thm}{section}
\numberwithin{equation}{section}

\renewcommand{\vec}[1]{\boldsymbol{\mathbf{#1}}}
\newcommand{\mat}[1]{\boldsymbol{\mathbf{#1}}}

\numberwithin{thm}{section}

\title{{\bf Gossip vs. Markov Chains, and Randomness-Efficient Rumor Spreading}}

\author{Zeyu Guo\footnote{This work is supported by
NSF grant CCF-1116111.  Part of this work was done while visiting Max Planck Institute for Informatics.}\\
California Institute of Technology\\ Pasadena, USA\\
\texttt{zguo@caltech.edu}
\and
He Sun\footnote{This work has partially been funded by the Cluster of Excellence ``Multimodal Computing and Interaction" within the Excellence Initiative of the German Federal Government. Part of this work was done while visiting California Institute of Technology. }\\
Max Planck Institute for Informatics\\
Saarbr\"{u}cken, Germany\\
\texttt{hsun@mpi-inf.mpg.de}
}

\date{}

\begin{document}

\thispagestyle{empty}

\maketitle

\begin{abstract}

We study gossip algorithms for the rumor spreading problem which asks one node to deliver a rumor to all nodes in an unknown network. We present the first protocol for \emph{any} expander graph $G$ with $n$ nodes such that, the protocol informs every node in $O(\log n)$ rounds with high probability, and uses $\tilde{O}(\log n)$ random bits in total. The runtime of our protocol is tight, and the randomness requirement of $\tilde{O}(\log n)$ random bits almost matches the lower bound of $\Omega(\log n)$ random bits for dense graphs. We further show that, for many graph families, polylogarithmic number of random bits in total suffice to spread the rumor in $O(\mathrm{poly}\log n)$ rounds.
These results together give us an almost complete understanding of the randomness requirement of this fundamental gossip process.

Our analysis relies on unexpectedly tight connections among gossip processes, Markov chains, and branching programs. First, we establish a connection between rumor spreading processes and Markov chains, which is used to approximate the rumor spreading time by the mixing time of  Markov chains. Second, we show a reduction from rumor spreading processes to branching programs, and this reduction provides a general framework to derandomize  gossip processes.
In addition to designing rumor spreading protocols, these novel techniques may have applications in studying parallel and multiple random walks, and  randomness complexity of distributed algorithms.

\vspace{1em}

\end{abstract}

\thispagestyle{empty}

\setcounter{page}{0}

\newpage

\section{Introduction\label{sec:IntroAbs}}

Gossip algorithms is one of the most important communication primitives in large $\mbox{networks}$, and  has been studied under different names such as rumor spreading, information
dissemination, or broadcasting. Efficient gossip algorithms for information spreading  have wide applications in  failure detection~\cite{Renesse1998}, resource discovery~\cite{Harchol-Balter1999}, replicated database systems~\cite{DGH+87,FPRU90}, and modeling the spread
of computer viruses~\cite{conf/soda/BergerBCS05}. Besides computer science,
the dynamics of such processes in social networks also constitutes a research topic in economics and sociology.

The simplest and widely studied form of gossip algorithms is the so-called \emph{push model} of rumor spreading.  Initially, a message, called
 \emph{a rumor}, is placed on an arbitrary node of an unknown network  with $n$ nodes. In subsequent synchronous rounds, every node that knows the rumor picks a neighbor uniformly at random and sends the rumor to the chosen neighbor. This process continues until every node gets the rumor.
 It was shown that this  simple protocol is very efficient on several network topologies~\cite{G11,FPRU90,ES09,ES07}.
 In particular, its \emph{runtime}, the number of rounds required until every node gets the rumor with high probability, is logarithmic in the number of nodes in the graph. Graphs satisfying this property range from complete graphs, hypercubes, Erd{\H o}s-R\'enyi random graphs, and ``quasi-regular" expanders~(i.e.,~expander graphs for which the ratio between the maximum and minimum degree is constant).  In addition to its efficiency, the protocol is local~(i.e.,~no knowledge of global graph structure is needed), simple,  and can tolerate  link failures. More recently, several variations of information spreading protocols have been proposed to allow information to spread efficiently on networks with weak expansion properties~\cite{conf/soda/Censor-HillelS11},
 arbitrary networks~\cite{conf/stoc/Censor-HillelHKM12},
  and dynamic networks~\cite{conf/soda/DuttaPRSV13}.

Most of these algorithms are inherently randomized in both their design and analysis in that they crucially rely on  choosing neighbors \emph{independently and uniformly at random} in each round, i.e., we assume that  every node of the graph has access to a random source of unbiased and independent coins.
 However, it is not known how to physically realize this abstraction in the real world and, from a theoretical point of view,
  it is not clear if this randomization is essential for efficiently disseminating the rumor.
 Hence the randomness requirement, the number of random bits used in total in order to spread the rumor efficiently, becomes a key measurement to evaluate rumor spreading protocols. One of the most studied questions concerns the randomness requirement: how many random bits are sufficient to efficiently spread a rumor to all nodes in a graph? While for any graph with $n$ nodes, the above-mentioned \emph{fully-random} push protocol requires $O(T\cdot n\log n)$ random bits for spreading a rumor within $T$ rounds, it is not difficult to show that for any graph $G$ of $n$ nodes, there is a protocol which uses $3\log n$ random bits in total, and whose runtime is as fast as the standard fully-random protocol~(cf. Corollary~\ref{thm:existence}).
However, the explicit construction of such protocols is more complicated, and a long line of research
 has been devoted to finding randomness-efficient protocols, see \cite{DFS09,GW11,conf/stacs/GSSW12} for instance.

\subsection{Our Results\label{sec:results}}

In this paper we establish a novel reduction from the problem of designing rumor spreading protocols of low randomness complexity to the problem of constructing pseudo-random generators~(\PRG) for branching programs.  To the best of our knowledge, this reduction gives the first application of the model of branching programs in the area of distributed computing and also provides a powerful tool for designing gossip algorithms.

At a high level, the connection between gossip processes and branching programs  is natural because (1) random walks over branching programs resemble the rumor spreading process where nodes send messages to random neighbors, and (2) in a rumor spreading protocol, each node has access to only its own
list of neighbors, and is oblivious to the structure of the network. This is an analogue of {\em oblivious derandomization} achieved by \PRG. However,  rumor spreading appears much more complicated
than small-space computation due to the following facts:  (1) In the rumor spreading process, rumors are ``duplicated'' every round, although every ``existing'' rumor viewed individually performs a random walk. Hence, instead of considering every single random walk performed by any fixed rumor, we need to study the  dynamics of the whole rumor spreading process.
 (2)
 The state of the process at some time essentially depends on the past behavior of all nodes and is by no means computable in small space. Indeed, even knowing if a single node $u$ gets the rumor at some round requires knowing the set of its neighbors having the rumor in the previous rounds, and may require $\deg(u)=\Theta(n)$ bits for dense graphs. For these reasons, this
 connection to small-space computation is delicate and not obvious.

Surprisingly, we show that such a reduction from designing rumor spreading protocols to constructing \PRG\ for branching programs exists. Hence the question of designing randomness-efficient rumor spreading protocols  is now exposed to the numerous techniques used in \textsf{PRG} constructions for small-space computation.  In particular,
\PRG\ with optimal parameters yield protocols whose randomness complexity matches the lower bound or the best known upper bound of existential results  from the probabilistic method (cf. Theorem \ref{cor_reduction}).
Our result is as follows:

%The present paper presents several randomness-efficient constructions of rumor spreading protocols. Our first result is as follows.

\begin{thm}[\textbf{Main Result}]\label{thm:mainresult1}
Let $G$ be a graph with $n$ nodes, spectral gap $\alpha\in(0,1)$ and irregularity $\beta\triangleq\Delta/\delta$. Then there is an explicit protocol using $O((\log (1/\alpha)+\log\beta)\cdot\log n)+\tilde{O}(\log n)$ random bits such that with high probability all nodes get the rumor in $T=O(C \log n)$ rounds,
where $C=(1/\alpha)\cdot\beta^2\max\{1, 1/(\alpha\cdot\Delta^{0.499})\}$.
\end{thm}

\thmref{mainresult1} implies that, for \emph{any} expander graph $G$ with $n$ nodes,  $\alpha=\Theta(1)$ and irregularity $\beta=O(1)$, the protocol finishes in $O(\log n)$ rounds and
 uses $\tilde{O}(\log n)$ random bits in total. Note that any protocol needs at least $\Omega(\log n)$ rounds to spread the rumor to all nodes, hence our runtime for expander graphs is tight. For the randomness complexity, our result improves the previous best bound of $O(\log^2 n)$ random bits~\cite{conf/stacs/GSSW12}. Since for any expander graph with minimum degree $\delta=n^{\Theta(1)}$, any protocol that finishes in $O(\log n)$ rounds with high probability needs at least $\Omega(\log n)$ random bits~(cf.~\thmref{lb_push}), our bound  is almost tight.

We further study the so-called \emph{averaging process}, which
is a generalization of rumor spreading process and can be considered as the random matching model of load balancing with a certain initial load vector~(cf. \cite{conf/focs/SauerwaldS12,conf/stoc/FriedrichS09}).
We show that this general averaging process can be modeled by branching programs as well, which leads to an explicit averaging protocol. This approach implies the following result~(\thmref{ResultAssumption}) for the rumor spreading problem, and has independent interest in studying other distributed algorithms, e.g. quasi-random load balancing~\cite{journals/siamcomp/FriedrichGS12}. Due to page limitation, we defer the formal discussion about the averaging process to \secref{SimplifiedProtocol}.

\begin{thm}\label{thm:ResultAssumption}
Let $G$ be a graph, $\mathsf{List}(u)$ be the adjacency list of node $u$, and $N(u)$ be the set of neighbors of $u$. We assume that each node $u$ knows the ID of its neighbors $v\in N(u)$, and its index in $\mathsf{List}(v)$ for any neighbor $v\in N(u)$.\footnote{We remark that similar assumptions are also made in other references, e.g.~\cite{conf/soda/Haeulper}, and one can deterministically
use $O(\Delta)$ preprocessing time to guarantee this assumption.}
Then there is an  explicit rumor spreading protocol using $O((\log(1/\alpha)+\log\beta+\log\log n)\cdot\log n)$ random bits, such that with high probability
all nodes get the rumor  in $T=O((1/\alpha)\cdot\beta^2\log n)$ rounds.
\end{thm}

Our third result is for general graph with conductance $\phi$. In contrast to \thmref{mainresult1} and
\thmref{ResultAssumption} that are based on branching programs, this result relies on the observation that the rumor spreading process enjoys
nice locality when the maximum degree is small.
%Our result is as follows:

\begin{thm}\label{thm:gen_graph_result}
Let $G$ be a graph with $n$ nodes, conductance $\phi$ and irregularity $\beta$. Then there is an explicit protocol using
 $O\big((1/\phi)\cdot \beta\cdot \log n \cdot(\log \log n +\log \Delta)\big)$ random bits in total, such that with high probability
all nodes get the rumor in $O((1/\phi)\cdot\beta\cdot \log n)$ rounds.
\end{thm}

The runtime in \thmref{gen_graph_result} matches the upper bound known in the truly random protocol, and is tight, in the sense that there are graphs with diameter $\Omega((1/\phi)\log n)$~\cite{CLP10b}. For the randomness requirement, our result improves the previous best one in \cite{conf/stacs/GSSW12}, which needs $O((1/\phi)\log^2 n)$ random bits in total and only holds for graphs with $\beta=O(1)$.

Our protocol takes advantage of the locality by using a ``two-level hashing'' construction: We use a family of objects called unbalanced expanders to hash the node IDs into a smaller space, and then apply the classical pairwise independent generators.
 This construction yields much smaller seed length than using pairwise independent generators alone. The protocol has the advantage of being very simple. Furthermore, a variant of this protocol using \PRG\ for combinatorial rectangles achieves the best possible runtime for strong expanders:

\begin{thm}\label{thm:expander_result}
Let $G$ be a graph such that $\Delta/\delta=1+o(1)$ and $\alpha=1-o(1)$. Then there is a protocol using  $O(\log n
\cdot(\log\log n+\log\Delta))$ random bits in total, such that with high probability all nodes get the rumor in
$\log n+\ln n+o(\log n)$ rounds.
\end{thm}

The runtime in \thmref{expander_result} matches the precise runtime  for the truly random protocol~\cite{FPRU90,ES09,ES07}, and is known to be tight~\cite{ES09}. Moreover, our protocol uses $O(\log n\cdot(\log\log n+\log\Delta))$ random bits in total, in contrast to $\Omega(\log^3 n)$ random bits used for all previous protocols, e.g.~\cite{FH09,conf/stacs/GSSW12}.
 These four results~(\thmref{mainresult1}--\thmref{expander_result}), together with the existential proof~(Corollary~\ref{thm:existence}) and the lower
bound analysis~(\thmref{lb_push}), give us an almost complete understanding of the randomness
complexity of this fundamental gossip problem.

\begin{rem}
One common feature of our protocols is that all randomness is picked by the initial node having the rumor, and the whole rumor spreading process becomes deterministic once the random seed is picked. We remark that, through our protocol, the whole rumor spreading dynamics is encoded in this short random seed, and any node can recover the rumor spreading process once it receives the random seed. This feature may have applications in studying algebraic gossip algorithms, and other settings.
\end{rem}

\subsection{Techniques\label{sec:techniques}}

To derive the results above, we develop several new techniques for studying gossip processes. We highlight some of them in this subsection.

\paragraph{Approximation via Random Walks.}

The usual analyses for fast rumor spreading proceed by showing some measure (e.g. the volume of the set of informed/uninformed nodes) increases or decreases over time.
Our approach is fundamentally different from previous work. Roughly speaking, we approximate the rumor spreading process by a collection of random walks and then use the rapid mixing of the random walks to prove the property of fast rumor spreading. It turns out that the pieces of local information provided by these random walks give a surprisingly good control of the global behavior of rumor spreading, despite that the walks are complicated and highly correlated.

Formally, we approximate the rumor spreading process by various random walks, distinguished by whether the walks are lazy or non-lazy in each round.
Each walk is associated with a positive number called its \emph{weight}.
A node $u$ is informed if the total weights of random walks reaching $u$ is positive. By the Cauchy-Schwarz inequality, we lower bound the probability of this event in terms of the expectation of the total weights reaching $u$ as well as its second moment.

%{\blue This technique in analyzing rumor spreading processes is new, and ...}

\paragraph{Analysis of Markov Chains.} With the weights chosen intelligently, the expectation and the second moment of total weights reaching a node are computed by certain Markov chains.
The expected total weights are computed by the chain
 $\mat{M}$ representing a lazy random walk in the graph.
 It follows from the rapid mixing of $\mat{M}$ that it can be well estimated using the stationary distribution of $\mat{M}$. The case for the second moments is more complicated as they correspond to a {\em non-reversible} chain $\mat{M}'$.
A key result we manage to show is that $\mat{M}'$ and $\mat{M}\otimes\mat{M}$ have very close stationary distributions and comparable mixing time. We remark that this result is interesting on its own since $\mat{M}'$ is a very natural Markov chain, closely related to the Doeblin coupling \cite{Lin02}.

\paragraph{Simulating Pull by Push.}

While a randomness-efficient protocols using a global seed can be easily implemented in the push model, the ``dual'' protocol in the pull model is not physically realizable,  as it is impossible for a node to perform random pulls before getting the seed.
Using the technique called {\em simulating pull by push}, we are able to employ the analysis for the pull model while actually using the push model. This is crucial in our analysis, since when most nodes already have the rumor, the random walks defined via push operations become too congested and correlated, whereas the ``reversed'' random walks using pull operations work well.

%
%\paragraph{Pseudorandom generators.}\NOTE{I suggest to drop it.}
%The approach above has the advantage of only depending on local information, namely the forward and the reversed random walks. This feature allows us to apply various pseudorandom generators to approximate the same local information using few random bits but still getting the same performance. In particular, we use \PRG\ for combinatorial rectangles to generate the bits used by different nodes in one round from a single seed. The memoryless feature of random walks and Markov chains further allows us to use \PRG\ for branching programs to generate the seeds used in different rounds.

%
%\paragraph{\PRG\ using independent seeds}
%We also investigate the approach of using \PRG\ with independent seeds in different rounds, as in \cite{conf/stacs/GSSW12}.
%We present a protocol that uses an independent seed of length $O(\log\log n+\log\Delta)$ in each round, where $\Delta$ is the maximum degree of the graph. The key observation is that the rumor spreading process enjoys nice locality when the maximal degree is small. Our protocol takes advantage of this feature by composing a pairwise independent generator with a family of pseudorandom objects called {\em unbalanced expanders}.
%The protocol has the advantage of being very simple. Furthermore, a simple variant of this protocol using \PRG\ for combinatorial rectangles achieves the tight spreading time $\log n+\ln n+o(\log n)$ for strong expanders.

\subsection{Related Work}

There is a large amount of literature devoted to various aspects of rumor spreading.
The majority of research studies the rumor spreading time in terms of the graph properties, e.g. conductance~\cite{G11,CLP10b}, mixing time~\cite{BGPS06}, diameter~\cite{FPRU90} and degree~\cite{FPRU90}. For instance, the first explicit connection between randomized rumor spreading and graph expansion was established by \citet{MS08}, who proved that on any regular graph with conductance $\phi$, the protocol finishes in $O((1/\phi)\cdot\log n)$ rounds. More recent work
includes the study of rumor spreading in social networks~\cite{conf/soda/FountoulakisPS12,conf/stoc/DoerrFF11} and
dynamic graphs~\cite{conf/soda/DuttaPRSV13,conf/esa/ClementiCDFIPPS13}, and algebraic gossip algorithms~\cite{conf/stoc/Haeupler11}.

The study of determining and reducing the amount of randomness required for rumor spreading has been studied extensively in the past years. \citet{DFS08} proposed a \emph{quasi-random} version of the rumor spreading push protocol. In contrast to $O(n\log^2n)$ random bits that used in the standard push model, the quasi-random rumor spreading model uses $\Theta(n\log n)$ random bits, and has been shown to be efficient on several graph topologies~\cite{FH09,DFS09}. Further progress along this line include
\cite{GW11,conf/stacs/GSSW12}.
 Besides this, researchers also studied the question of designing randomness-efficient or deterministic protocols for similar problems.  For instance, \citet{conf/soda/Haeulper} presented one deterministic gossip algorithm for the $k$-local broadcast
and the global broadcast problem. However, the algorithms in \cite{conf/soda/Haeulper} require that  all nodes in the graph have unique identifiers~(UID), and every node knows
its own and the neighbors' UIDs. Hence the techniques developed there cannot be applied to our setting.

In addition to rumor spreading, the technique of  pseudorandomness was also studied in other settings of online algorithms, e.g., in the context of Local Computation Algorithms~(\textsf{LCA})~\cite{AlonRVX12}, and complexity analysis of information spreading in dynamic networks~\cite{conf/soda/DuttaPRSV13}.

\subsection{Notations}

Let $G=(V,E)$ be a connected, undirected, and simple graph with $n$ nodes. For any node $u$, the degree of $u$ is represented by $\deg(u)$. Let $\Delta,\delta$ and $d$ be the maximum, minimum and average degree of $G$, respectively, and call $\beta\triangleq\Delta/\delta$ the \emph{irregularity} of $G$.
We use $\mat{A}_G$ to express the adjacency matrix of $G$, and $\mat{N}_G\triangleq \mat{D}^{-1/2}\mat{A}_G\mat{D}^{-1/2}$, where $\mat{D}$ is the  $n\times n$ diagonal matrix defined by $\mat{D}_{uu}=\deg(u)$ for $u\in V[G]$.
Define the $n$ real eigenvalues of $\mat{N}_G$ by $1=\lambda_1\geq\cdots\geq \lambda_n\geq -1$, and let $\lambda_{\max}\triangleq\max\{\lambda_2,|\lambda_i|\}$.
The spectral gap $\alpha$ is defined by $\alpha\triangleq 1-\lambda_2$, whereas the absolute spectral gap is defined as $1-\lambda_{\max}$. For simplicity, we also use $\alpha$ to express the spectral expansion of a reversible Markov chain if the chain is clear from the context.

 By $\log x$ we denote the binary logarithm of $x$.
For any integer $m$, define $[m]\triangleq\{0,\ldots,m-1\}$.
With high probability stands for with probability $1-n^{-\Theta(1)}$.

\section{Gossip vs. Markov Chains\label{sec_expander}}

Let $G=(V,E)$ be an undirected and simple graph with $V[G]=[n]$.
We consider only $T'$-round protocols for $G$, in which nodes send rumors only for the first $T'$ rounds, and assume that $T'=O(n^c)$ for a constant $c>0$. Through this section, we  assume that each node has a unique identifier~(ID), and each node initially solely knows its own ID, which is from $0$ to $n^c$ for a constant $c$. Let $s$ be  the initial node having the rumor.
For simplicity, we assume the adjacency list of each node $u$ has length $\Delta$, and the last $\Delta-\deg(u)$ neighbors are $u$ itself, i.e. we add $\Delta-\deg(u)$ self-loops for ever node $u$. However,  we use $\deg(u)$ and $N(u)$ to represent the degree and the set of neighbors of $u$ respectively in the underlying simple graph.

\subsection{Preliminaries}

Given $d\in\mathbb{N}$ and a finite set $S=\prod_{i\in [d]}S_i$, define
$\mathsf{CR}_S\triangleq\left\{\prod_{i\in [d]} A_i: A_i\subseteq S_i\right\}$.
The members of $\mathsf{CR}_S$ are called {\em combinatorial rectangles in $S$} and $d$ is their {\em dimension}.
For $\epsilon>0$, $d\in\mathbb{N}$, and a finite set $S=\prod_{i\in [d]}S_i$, we call $\gen:\{0,1\}^{\ell}\to S$ an $\epsilon$-\textsf{PRG} for $\mathsf{CR}_S$ with seed length $\ell$ if  $
\left|
\Prob{x\in \{0,1\}^{\ell}}{\gen(x)\in A}
-|A|/|S|
\right|
\leq \epsilon
$ for any $A\in\mathsf{CR}_S$.

The second family of \PRG\ that we will use is \PRG\  for Branching Programs\footnote{See \defref{BranchingProgram} for the formal definition of branching programs.}. Let $\mathcal{B}$ be a branching program of length $L$, width $W$ and degree $D$. For $x=(x_1,\dots,x_L)\in [D]^L$ and a node $(s,0)$ on the first layer, define $\mathcal{B}(s,x)\in [W]$ such that the random walk that starts from $(s,0)$ and takes the edge with label $x_i$ at the $i$th step for $1\leq i\leq L$ finally arrives at $(\mathcal{B}(s,x), L)$. We call a function $\gen:\{0,1\}^{\ell}\to [D]^L$ an $\epsilon$-\textsf{PRG} for $(L,W,D)$-branching programs if for any $(L,W,D)$-branching program, and any node $(s,0)$ on the first layer, it holds that
$$
\sum_{u\in [W]}\left|\Prob{x\in\{0,1\}^\ell}{\mathcal{B}(s,\gen(x))=u}-\Prob{x\in [D]^L}{\mathcal{B}(s,x)=u}\right|\leq \epsilon.
$$

%\begin{thm}[\cite{NZ96}]\label{thm:nz}
%There exists an explicit $\epsilon$-\textsf{PRG} for $(L,W,D)$-branching programs with seed length {\red to be written}.
%\end{thm}

\subsection{Analysis of the Prototype Protocol}

In this subsection we
relate rumor spreading processes to Markov chains, and show how the mixing time of certain Markov chains relates to the rumor spreading time. We first analyze the following prototype of  rumor spreading  protocols, which includes the standard push protocol as a special case.

\begin{protocol}[Prototype of Rumor Spreading Protocols]\label{pro_prototype}
Let $\mathcal{D}$ be a distribution over the set of functions $f:[T]\times V[G]\to [\Delta]$.
Sample $f$ according to $\mathcal{D}$.
In the $i$th round, an informed node $u$ sends the message to its $f(i,u)$th neighbor in its adjacency list.
\end{protocol}

We are primarily interested in analyzing Protocol~\ref{pro_prototype} when $\mathcal{D}=\mathcal{U}$ is the uniform distribution, i.e. $f(i,u)$ are chosen from $[\Delta]$ independently and uniformly at random for all $i$ and $u$.

\paragraph{Approximation via Random Walks.}

To analyze the runtime of Protocol~\ref{pro_prototype},  we compare the process of rumor spreading with a random walk on a branching program.
 For random walks, a walk always stays at a single node throughout the process, although this node keeps changing. On the other hand, in the process of rumor spreading, each informed node $u$ randomly sends the rumor to one of its neighbors $v$ in each round, and then $u,v$ are both informed subsequently. So we may think of rumor spreading as many random walks in parallel: When node $u$ sends the rumor  to $v$,
 one random walk moves from $u$ to $v$ whereas another one stays at $u$.
 In order to characterize this behavior, we introduce the notion of forward and reversed random walks. For any round $i\in [T]$ and node $u\in V[G]$, denote  by $\tilde{f}(i,u)$ the $f(i,u)$th neighbor of $u$ in its adjacency list.

\begin{defi}[Forward random walks]\label{defi_walk}
Consider a random rumor spreading process in $T$ rounds on a graph $G$ using Protocol \ref{pro_prototype} determined by $f\sim\mathcal{D}=\mathcal{U}$.
A forward random walk of length $k\in [T]$ with pattern $S=(s_0,\dots,s_{k-1})\in\mathcal{C}_k\triangleq
\{\text{lazy}, \text{non-lazy}\}^k$ is a sequence of $k+1$ nodes $(p_0,\dots,p_k)$ of $G$, such that for all $i\in [k]$: (i) if $s_i=\text{lazy}$, then $p_{i+1}=p_i$; (ii) if $s_i=\text{non-lazy}$, then $p_{i+1}=\tilde{f}(i,p_i)$.
\end{defi}

We also define reversed random walks, tailored to the idea of simulating pull using push.
Roughly speaking, a reversed random walk takes a step from node $v$ to $u$ if $u$ is the
unique node pushing to $v$. For technical reasons, we introduce auxiliary random variables $r_{i,u}$ uniformly distributed over $[0,1]$ for each $i\in [T]$ and $u\in V[G]$ to equalize the probabilities of successful steps of reversed random walks made from different nodes. These random variables only appear in the analysis, not in the protocol constructions.
Then the reversed random walks are determined by the randomness $f\sim\mathcal{D}$ together with $r_{i,u}$, whereas the forward walks are solely determined by $f$.
See \defref{RevsersedRW} for the formal definition of reversed random walks.

%For $k\in [T/2]$\NOTE{$T/4$?}, $u\in V[G]$ and $S\in\mathcal{C}_k=\{\text{lazy}, \text{non-lazy}\}^k$, the forward and the reversed
%random walk of length $k$ with pattern $S$ starting from $u$ (i.e., $p_0=u$) uniquely exist, and are uniquely determined by $f\sim\mathcal{D}$ and the auxiliary randomness $\{r_{i,u}\}$. Let $\tilde{\mathcal{D}}$ be the distribution of the whole randomness $r=(f,\{r_{i,u}\})$, which is the product of $\mathcal{D}$ with copies of uniform distributions over $[0,1]$.
For $k\in [T/4]$, $u,v\in V[G]$ and $S\in\mathcal{C}_k=\{\text{lazy}, \text{non-lazy}\}^k$,
let $X^S_{u,v}$ (resp. $Y^S_{u,v}$) be the indicator random variable of the event that the unique forward (resp. reversed) walk with pattern $S$ and initial node $u$ is at node $v$ in the $k$th round.
For $\gamma\in (0,1)$, let $\mathcal{D}_{\gamma,k}$ be the distribution over $\mathcal{C}_k$ where entries are independently chosen to be lazy with probability $1-\gamma$.

\newcommand{\lembound}{
For Protocol~\ref{pro_prototype} with $\mathcal{D}=\mathcal{U}$ and initial node $s$, any $0\leq k\leq T/4$, and $\gamma \in (0,1)$, a node $w$ is informed in $T$ rounds with probability at least
\begin{equation}\label{eq:lowerboundprob}
\frac{\sum_{u,v\in V[G]}\EXX{r,S}{X^S_{s,u}}\EXX{r,S}{X^{S}_{s,v}}
\EXX{r,S}{Y^S_{w,u}}\EXX{r,S}{Y^{S}_{w,v}}}
{\sum_{u,v\in V[G]}\EXX{r,S,S'}{X^S_{s,u}X^{S'}_{s,v}}
\EXX{r,S,S'}{Y^S_{w,u}Y^{S'}_{w,v}}}
\end{equation}
where $r$, $S$ and $S'$ are independent with distributions $\tilde{\mathcal{D}}$,  $\mathcal{D}_{\gamma,k}$ and $\mathcal{D}_{\gamma,k}$ respectively.}

We fix an arbitrary node $w\in V[G]$, and study the probability that node $w$ is informed in $T$ rounds. Clearly, if there exist a forward random walk $p$ from $s$ to some node $u$ and a reversed random walk $p'$ from $w$ to $u$, then the rumor is sent from $s$ to $u$ following $p$ and then from $u$ to $w$ following the reversal of $p'$.
Also note that the two walks exist   if and only if $X^S_{s,u}Y^{S'}_{w,u}>0$ for some $S$, $S'$ and $u$.
Therefore it holds for any $k\in[T/4]$ that
\begin{equation}\label{eq:lowerboundsingle}
\Pro{w \text{ receives the message in $T$ rounds}}\geq \Pro{\sum_{S,S'\in\mathcal{C}_k, u\in V[G]} {X}^S_{s,u}{Y}^{S'}_{w,u}>0},
\end{equation}
where the probability is taken over the randomness $f\sim\mathcal{D}$ and $r_{i,u}$.

We want to reduce the global event $\sum_{S,S'\in\mathcal{C}_k, u\in V[G]} {X}^S_{s,u}{Y}^{S'}_{w,u}>0$ to local events $X^S_{s,u}$ and $Y^{S'}_{w,u}$. By using Cauchy-Schwarz inequality, and linearity of expectation, we show that \eq{lowerboundsingle} is lower bounded by
\begin{equation}\label{eq:Expexpansion}
\frac{\sum_{u,v\in V[G]}\EXX{r,S}{X^S_{s,u}}\EXX{r,S}{X^{S}_{s,v}}
\EXX{r,S}{Y^S_{w,u}}\EXX{r,S}{Y^{S}_{w,v}}}
{\sum_{u,v\in V[G]}\EXX{r,S,S'}{X^S_{s,u}X^{S'}_{s,v}}
\EXX{r,S,S'}{Y^S_{w,u}Y^{S'}_{w,v}}}.
\end{equation}
Hence the runtime of Protocol~\ref{pro_prototype} can be derived by analyzing multiple random walks individually or pairwisely.
See \lemref{lembound} for detailed analysis.

\paragraph{Analysis using Markov Chains.}

We study the expectations in
\eq{Expexpansion} in terms of finite-state Markov chains.
For simplicity, we represent these Markov chains by stochastic matrices.
Recall that a stochastic matrix $\mat{M}''\in\mathbb{R}^{n\times n}\otimes\mathbb{R}^{n\times n}$ is a \emph{coupling} of $\mat{M},\mat{M}'\in\mathbb{R}^{n\times n}$ if (i)  $\sum_{x\in [n]}\mat{M}''_{(u,w)(v,x)}=\mat{M}_{u,v}$ for any $u,w,v\in [n]$, and (ii) $\sum_{v\in [n]}\mat{M}''_{(u,w)(v,x)}=\mat{M}'_{w,x}$ for any $u,w,x\in [n]$.

We  define the ``bi-lazy'' analogue of lazy Markov chains with respect to a coupling where the two chains choose to be lazy or non-lazy independently.

\begin{defi}
For $\gamma\in[0,1]$, let $\mathcal{L}_{\gamma}(\mat{M})\triangleq(1-\gamma)\mat{I}+\gamma\mat{M}$
be the lazy Markov chain.
\end{defi}

\begin{defi}[Lazy coupling]
Let $\mat{M}''$ be a coupling of $\mat{M},\mat{M}'\in\mathbb{R}^{n\times n}$.
For $\gamma,\gamma'\in [0,1]$, define
$
\mathcal{L}_{\gamma,\gamma'}(\mat{M}'')\triangleq(1-\gamma)(1-\gamma')(\mat{I}\otimes\mat{I})
+(1-\gamma)\gamma'(\mat{I}\otimes\mat{M}')
+\gamma(1-\gamma')(\mat{M}\otimes\mat{I})
+\gamma\gamma'\mat{M}''.
$
That is, $\mathcal{L}_{\gamma,\gamma'}(\mat{M}'')$ is  a coupling of $\mathcal{L}_\gamma(\mat{M})$ and $\mathcal{L}_{\gamma'}(\mat{M}')$.
\end{defi}

\begin{defi}[Doeblin coupling \cite{Lin02}]
Let $\mat{M}\in\mathbb{R}^{n\times n}$ be a stochastic matrix. The Doeblin coupling $\mathcal{Q}(\mat{M})$ of two copies of $\mat{M}$ is defined as
$$
\mathcal{Q}(\mat{M})_{(u,w)(v,x)}\triangleq\begin{cases}
(\mat{M}\otimes \mat{M})_{(u,w)(v,x)} & u\neq w,\\
\mat{M}_{uv} & u=w, v=x,\\
0 & u=w, v\neq x.
\end{cases}
$$
\end{defi}

Using the above definitions, we are able to characterize the expectations in \eqref{eq:Expexpansion} in terms of Markov chains. For instance, the first and the second moments $\EXX{r,S}{X^S_{u,v}}$ and $\EXX{r,S,S'}{X^S_{u,v}X^{S'}_{w,x}}$ about forward random walks are characterized by the chains $\mathcal{L}_{\gamma}\left(\mat{M}_{\reg(G)}\right)$ and
$\mathcal{L}_{\gamma,\gamma}\circ\mathcal{Q}\left(\mat{M}_{\reg(G)}\right)$
respectively, and similar results hold for reversed walks. Hence we reduce the problem of lower bounding \eqref{eq:lowerboundsingle} to the study of these Markov chains.

Notice that matrix $\mathcal{Q}(\mat{M})$ agrees with $\mat{M}\otimes \mat{M}$ except on the rows indexed by $(u,u)$, $u\in V[G]$. This is a manifestation of the fact that the ``non-lazy" steps from the same node made by two different forward/reversed random walks are not independent, i.e., every informed node can only send the rumor to one neighbor in each round. Despite this complication, we show that $\mathcal{Q}(\mat{M})$ is actually quite close to $\mat{M}\otimes \mat{M}$:

\newcommand{\lemmarkov}{
Let $r$, $S$ and $S'$ be independent with distributions $\tilde{\mathcal{D}}$ (induced by $\mathcal{D}=\mathcal{U}$), $\mathcal{D}_{\gamma,k}$ and $\mathcal{D}_{\gamma,k}$ respectively.
Then for stochastic matrices
$\mat{M}_1=\mathcal{L}_{\gamma}\left(\mat{M}_{\reg(G)}\right)$,
$\mat{M}_2=\mathcal{L}_{\gamma,\gamma}\circ\mathcal{Q}\left(\mat{M}_{\reg(G)}\right)$,
$\mat{M}_3=\mathcal{L}_{\gamma}\circ\mathcal{L}_{\gamma'}\left(\mat{M}_{\reg(G)}
\right)$,
$\mat{M}_4=\mathcal{L}_{\gamma,\gamma}\circ\mathcal{Q}
\circ\mathcal{L}_{\gamma'}\left(\mat{M}_{\reg(G)}\right)$,
$\gamma'\triangleq (1-1/\Delta)^{\Delta-1}$,
and any $u,v,w,x\in V[G]$, the following statements hold:
\begin{enumerate}
\item $\EXX{r,S}{X^S_{u,v}}
=\left\langle\vec{e}_u\mat{M}_1^k,\vec{e}_v\right\rangle$,
\item $\EXX{r,S,S'}{X^S_{u,v}X^{S'}_{w,x}}=\left\langle\vec{e}_{(u,w)}\mat{M}_2^k, \vec{e}_{(v,x)}\right\rangle$,
\item $\EXX{r,S}{Y^S_{u,v}}
=\left\langle\vec{e}_u\mat{M}_3^k,\vec{e}_v\right\rangle$,
 and
\item $\EXX{r,S,S'}{Y^S_{u,v}Y^{S'}_{w,x}}=\left\langle\vec{e}_{(u,w)}\mat{M}_4^k, \vec{e}_{(v,x)}\right\rangle$.
\end{enumerate}}
%\mylemma{lemmarkov}{\lemmarkov}

\newcommand{\lemregularization}{Suppose graph $G$ has spectral gap $\alpha$ and irregularity $\beta$. Then $\reg(G)$ has spectral gap at least $\beta^{-2}\alpha$.
}

\newcommand{\lemmixing}{
Suppose $\mat{M}\in\mathbb{R}^{n\times n}$ is a doubly-stochastic matrix with spectral gap $\alpha>0$, and suppose $\mat{M}_{uv}\leq \eta$ for any distinct $u,v\in V[G]$. Then
for any distribution $\vec{u}$ over $V[G]\times V[G]$, $k\in\mathbb{N}$, and $0\leq\gamma\leq\min\left\{1/3, \alpha \eta^{-1/2}/9\right\}$, we have
$$
\left\|\vec{u}\left(\mathcal{L}_{\gamma,\gamma}\circ\mathcal{Q}(\mat{M})\right)^k-\vec{\pi}\otimes \vec{\pi}\right\|_2
\leq (1-\gamma\alpha/2)^k + 2\sqrt{2}\gamma\alpha^{-1} n^{-3/2},
$$
where $\vec{\pi}$ denotes the uniform distribution over $V[G]$.
}
\mylemma{lemmixing}{\lemmixing}

One corollary of \lemref{lemmixing} states  that the stationary distribution of the Markov chain $\mathcal{L}_{\gamma,\gamma}
\circ\mathcal{Q}(\mat{M})$ is very close to $\vec{\pi}\circ\vec{\pi}$, and its mixing rate is comparable to that of $\mat{M}\otimes\mat{M}$ (see Corollary~\ref{cor_approx}).
Using the rapid mixing of $\mathcal{L}_{\gamma}\left(\mat{M}_{\reg(G)}\right)$ and $\mathcal{L}_{\gamma,\gamma}
\circ\mathcal{Q}(\mat{M})$ (and similar chains for reversed random walks), we obtain an upper bound of the runtime of Protocol~\ref{pro_prototype}, which holds for general graphs with spectral gap $\alpha$ and irregularity $\beta$. Our result in this subsection is summarized as follows:

\newcommand{\thmrandom}{
Suppose $G$ has spectral gap $\alpha$ and irregularity $\beta$. Using Protocol \ref{pro_prototype} with distribution $\mathcal{D}=\mathcal{U}$, any node gets the rumor in $T=O(C \log n)$ rounds with probability at least $1-O(n^{-2c})$ where $C=(1/\alpha)\cdot\beta^2\max\{1, 1/(\alpha\cdot\Delta^{0.5-c})\}$ and $c>0$ is an arbitrary small constant.}

%mylemma{thmrandom}{\thmrandom}

\begin{thm}\label{thm:FullRandomness}
Suppose $G$ has spectral gap $\alpha$ and irregularity $\beta$. Using Protocol \ref{pro_prototype} with distribution $\mathcal{D}=\mathcal{U}$,  with high probability all nodes get the rumor in $T=O(C \log n)$ rounds, where $C=(1/\alpha)\cdot\beta^2\max\{1, 1/(\alpha\cdot\Delta^{0.499})\}$.
\end{thm}

We remark that our analysis above provides a fundamentally new approach to analyze the rumor spreading time of general graphs and, as shown in \thmref{FullRandomness}, the result
 is tight for certain graph families, e.g. $T=O(\log n)$ for any expander graph with $n$ nodes and $\beta=O(1)$.

\subsection{A Randomness-Efficient Protocol}

The discussion above relates rumor spreading processes to multiple random walks.
The transitions of these random walks from a fixed node only depend on local information and are characterized by combinatorial rectangles. Moreover the memoryless feature of random walks/Markov chains allow us to compute them in log-space, or branching programs with polynomial width. Using \PRG\ for combinatorial rectangles and those for branching programs, we obtain a distribution that is samplable with a short seed and  has almost the same performance as the distribution $\mathcal{D}=\mathcal{U}$ in Protocol~\ref{pro_prototype}. This gives Protocol~\ref{pro0} that  corresponds  to \thmref{mainresult1}.

\begin{protocol}\label{pro0}
%Let $m=n^{\Theta(1)}$ be a sufficiently large power of 2.
Pick the following objects:
\begin{itemize}\itemsep -0.3pt
\item an explicit $\epsilon$-\textsf{PRG} $\gen=(\gen_0,\dots,\gen_{n-1}):\{0,1\}^{\ell}\to [m]^n$ for $\mathsf{CR}_{[m]^n}$
    with seed length $\ell$, and
\item an explicit $\epsilon'$-\textsf{PRG} $\gen'=(\gen'_0,\dots,\gen'_{T/2-1}):\{0,1\}^{\ell'}\to \left(\{0,1\}^\ell\right)^{T/2}$ for $(T/2,n^2,2^\ell)$-branching programs with seed length $\ell'$
\end{itemize}
where $\epsilon^{-1},\epsilon'^{-1},m=n^{\Theta(1)}$ are sufficiently large.
%These two objects $\gen$ and $f$ can be uniquely constructed from $n^c$, and hence are known to every informed node.

The initial node having the rumor independently chooses random strings $x,y\in\{0,1\}^{\ell'}$.
These random strings are appended with the rumor and sent to other nodes.
\begin{itemize}
\item In the $i$th round for $0\leq i < T/2$, an informed node $u$ sends the rumor to the neighbor with index $\gen_u(\gen'_i(x)) \bmod \Delta$ in its adjacency list.
\item In the $i$th round for $T/2\leq i<T$, let $j=\lfloor \frac{T-i-1}{2}\rfloor$. For $u\in V[G]$, let $(r_0,r_1)=\gen_u(\gen'_j(y)) \bmod \Delta^2 \in [\Delta]^2$.
Then $u$ sends the rumor to the $r_0$th neighbor if $i=T-1-2j$, and to the $r_1$th neighbor if $i=T-2-2j$.
\end{itemize}
\end{protocol}

Setting $C=(1/\alpha)\cdot\beta^2\max\{1, \alpha^{-1}/\Delta^{0.499}\}$, Protocol \ref{pro0} uses $2\ell'$ random bits, and with high probability informs all nodes in $T=O(C \log n)$ rounds.
As a consequence, we obtain the following reduction:
\begin{thm}\label{cor_reduction}
Given an explicit $\epsilon$-\textsf{PRG} for $\mathsf{CR}_{[m]^n}$
with seed length $\ell$ and an explicit $\epsilon'$-\textsf{PRG} for $(T/2,n^2,2^\ell)$-branching programs with seed length $\ell'$, where $\epsilon^{-1},\epsilon'^{-1},m=n^{\Theta(1)}$ are sufficiently large,
there exists an explicit protocol using $2\ell'$ random bits such that,  with high probability all nodes get the rumor in $T=O(C \log n)$ rounds.
In particular, given an explicit $\epsilon$-\textsf{PRG} for $(L,W,D)$-branching programs with seed length $O(\log n)$ where $L=\max\{T/2, n\}$, $W=n^2$, and $D,\epsilon^{-1}=n^{\Theta(1)}$ sufficiently large, there exists an explicit protocol using $O(\log n)$ random bits, and with high probability informs all nodes in $T=O(C \log n)$ rounds.\footnote{This follows from the simple observation that combinatorial rectangles in $[m]^n$ can be computed by $(n,2,m)$-branching programs.}
\end{thm}

Combining the reduction above with known explicit constructions of \PRG\ (Theorem \ref{thm_cr}, Theorem \ref{thm:inw}), we obtain Theorem \ref{thm:mainresult1}.

\begin{rem} We remark here that, by allowing
every node to have $O(\Delta)$ preprocessing time before the protocol starts, the rumor spreading time can be improved to $T=O((1/\alpha)\cdot\beta^2\log n)$, which corresponds to \thmref{ResultAssumption}. See \secref{SimplifiedProtocol} for formal discussions.
\end{rem}

\section{Two-Level Hashing Protocols\label{sec:general}}

In this section we present two protocols. Our protocols are based on pairwise independent generators and unbalanced expanders with near-optimal expansion. Here different rounds use different random bits. In contrast to $O(n\log n)$ random bits per round used in the truly random protocol, we show that  $O(\log\log n+\log \Delta)$ random bits per round suffice to spread the rumor efficiently on
general graphs $G$. In contrast to protocols in Section~\ref{sec_expander}, the protocols in this section do not need to assume that nodes have initial IDs, and we can combine the protocols with an ID distribution mechanism so that every node gets a unique ID once it gets the rumor.
Formally,  in round $0$ there is one arbitrary node having the rumor, and the ID of this node is set to be $0$.
We assume that node $0$ knows the maximum degree $\Delta$, and an upper bound $n'\triangleq n^{c}~(c\geq 1)$ of the number of nodes $n$.
Moveover, node $0$ chooses a binary string, called \emph{seed}, uniformly
at random, and the seed  is appended to the rumor. In subsequent rounds, whenever one node with ID $u$ sends the rumor to one of its neighbors in round $t$, it also sends a unique string consisting of the ID $u$, parameters $n', \Delta$, and current round number $t$. A node is \emph{uninformed} as long as it has not received a rumor. Once a node receives the first rumor from an informed node with ID $u$ in round $t$, it becomes \emph{informed} and gets a unique ID defined by $g_t(u)\triangleq 2^{t-1}+u$. If one node  becomes informed from multiple informed nodes, then
this node  chooses an arbitrary node with ID $u$ that informs it and uses $g_t(u)$ as its ID. It was shown in \cite{conf/stacs/GSSW12} that, through this protocol above, all informed nodes have different IDs, and all the IDs are in $[2^T]$ if the protocol finishes in $T$ rounds.

\subsection{Protocol For Graphs with Certain Conductance}

Our first protocol in this section corresponds to \thmref{gen_graph_result}, and holds for graphs with conductance $\phi$. Formally, for a graph $G$ of $n$ nodes, the \emph{conductance}  $\phi(G)$ of $G$ is defined by
\[
\phi(G)\triangleq\min_{S\subseteq V, 0<|S|<n}\frac{e(S,V\setminus S)}{ \min\{ \vol(S), \vol(V\setminus S) \}},
\]
where $\vol(S)\triangleq\sum_{u\in S}\deg(u)$ is the volume of $S$, and $e(S,T)\triangleq |\{ \{u,v\}: u\in S\mbox{ and } v\in T \}|$ is the number of edges between $S$ and $T$.
The formal description of our  protocol is as follows:

\begin{protocol}[Protocol for Graphs with Certain Conductance]\label{pro1}
Let $\epsilon=\Delta^{-\Theta(1)}$ be sufficiently small and $m=2^{\lceil\log(4/\epsilon)\rceil}$. Pick the following objects:
\begin{itemize}\itemsep -0.3pt

\item An explicit $(K, (1-\epsilon^2/4)D)$-expander $\Gamma: [n^c]\times [D] \to \bigsqcup_{i\in [D]}[M_i]$, where $K=2$, $D=\left((\log n)/{\epsilon}\right)^{O(1)}$ and $M_0=\dots=M_{D-1}=M\leq D$.

\item An explicit pairwise independent generator $\gen=(\gen_1,\dots,\gen_M):\{0,1\}^{\ell}\to [m]^M$, where $\ell=O(\log m+\log M)=O(\log\log n+\log\Delta)$.
\end{itemize}
These two objects $\gen$ and $\Gamma$ can be uniquely constructed from $n^c$ and $\Delta^{\Theta(1)}$, and hence are known to every informed node.

The initial node having the rumor chooses a random string $(s_1,\dots,s_{T})$ where every $s_i$ is of the form $(x_i,y_i)\in [D]\times\{0,1\}^{\ell}$.
This random string is appended with the rumor and sent to other nodes. Once one node gets the rumor, it gets the ID $u$. In the $i$th round, node $u$ computes $r=\Gamma(u,x_i)$ that is in $[M_u]$, the $u$th copy of $[M]$. Node $u$ computes
$y\triangleq\gen_r(y_i) \bmod \Delta$, and chooses the neighbor with index $y$ in its adjacency list to send the rumor if $y\leq \deg(u)$.
\end{protocol}

Protocol~\ref{pro1}  presents a nice ``two-level hashing" framework: The first level is based on a pairwise independent generator $\gen$. While the \textsf{PRG}-based protocol in \cite{conf/stacs/GSSW12} needs to generate $O(n)$ blocks and different nodes need to use different blocks, our protocol only needs $M=(\Delta\log n)^{O(1)}$ blocks and hence $O(\log\log n+\log \Delta)$ random bits suffice for this purpose. The second level uses unbalanced expanders to map the node with ID  $u\in[n^c]$ to  $r\in\left[\Delta^{O(1)}\right]$ by using $O(\log\log n+\log\Delta)$ random bits. After these, node $u$ uses the value of the $r$th block of $\gen$ to choose the neighbors. It is easy to see that every informed node $u$ only needs $O(\mathrm{poly}\log n)$ arithmetic operations per round in order to determine its neighbor.

\begin{figure}[hbt]
\centering
\begin{tikzpicture}[>=stealth, ->]
\tikzstyle{tmtape}=[draw, minimum width=1cm, minimum height=0.5cm]
\tikzstyle{comp}=[draw, thick, rounded corners=2mm, inner sep=0.3cm]

\begin{scope}[start chain=1 going right,node distance=-0.15mm]
    \node (b1) [on chain=1,tmtape] {};
    \node (b2) [on chain=1,tmtape] {};
    \node (mid_node) [on chain=1,tmtape, minimum width=1.5cm] {$\ldots$ };
    \node (b3) [on chain=1,tmtape, fill=black!20] {};
    \node (b4) [on chain=1,tmtape] {};
    \node [on chain=1, outer sep=0.2cm] {\small $(\Delta\log n)^{O(1)}$ blocks};
\end{scope}

\begin{scope}[node distance=0.5cm]
\node (prg) [above=of mid_node, comp, minimum width=2cm] {\textsf{PRG} $\mathcal{G}$};
\node (t1) [above=of prg] {\small $O(\log\log n+\log\Delta)$ random bits}; {selector};
\end{scope}

\begin{scope}[node distance=1cm]
\node (cond) [left=of b1, comp, minimum width=1cm] {$\Gamma$};
\node (t5) [below=of b3] {};
\end{scope}

\begin{scope}[node distance=0.5cm]
\node (t2) [above=of cond, text width=3cm, text centered] {\small $O(\log\log n+\log\Delta)$ random bits};
\node (t3) [left=of cond, text width=1.5cm, text centered] {\small node ID $u\in [n^c]$};
\node (t4) [below=of cond] {\small index $r$};
\end{scope}

\draw (t1) -- (prg);
\draw (t2) -- (cond);
\draw (t3) -- (cond);
\draw (cond) -- (t4);
\draw (b3) -- (t5) node [left, near end] {\small $r$th block};
\draw (prg.230) -- (b1.north);
\draw (prg.250) -- (b2.north);
\draw (prg.270) -- (mid_node);
\draw (prg.290) -- (b3.north);
\draw (prg.310) -- (b4.north);
\end{tikzpicture}

\caption{Illustration of the protocol for general graphs.
Every node $u$ uses an unbalanced expander $\Gamma$ to generate an index $r$, and uses the $r$th block of \textsf{PRG} $\gen$ to choose a neighbor to send the rumor.
\label{p1}}

\end{figure}
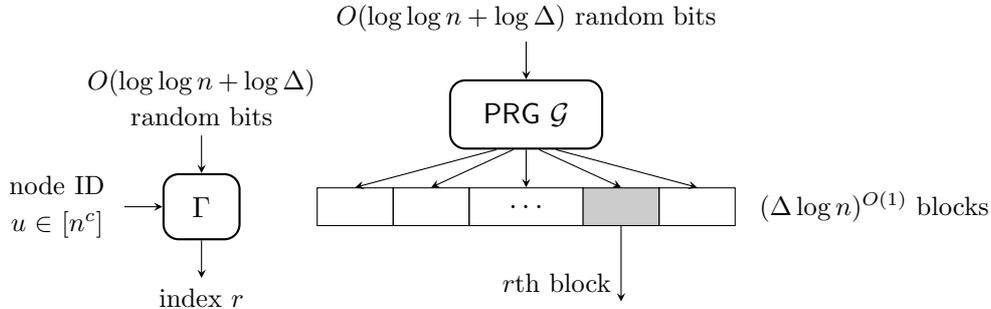

\begin{pro}\label{pro:protocol1}
Assume that Protocol~\ref{pro1} finishes in $T$ rounds. Then it uses $O(T\cdot(\log\log n+\log\Delta))$ random bits in total.
\end{pro}

\begin{rem}
Using the explicit constructions of unbalanced expanders in \cite{GUV09} and pairwise independent generators in \cite{CW79}, our protocol is very simple and can be described as follows:
Assign each node with ID $u\in[n^c]$ with a distinct polynomial $p_u$ of degree at most $\lceil c\log_q n\rceil$ over a finite field $\mathbb{F}_q$ of size $q=\left(\Delta\log n\right)^{\Theta(1)}$.
The protocol then uses the random string $(s_1,\dots,s_{T})$ where every $s_i$ is of the form $(x_i,a_i,b_i)\in \mathbb{F}_q^3$. Then node $u$ computes $z=a_i\cdot p_u(x_i)+b_i$ (over $\mathbb{F}_q$) in the $i$th round, and chooses the neighbor with index $\left(z\bmod \deg(u)\right)$ in its adjacency list to send the rumor.
\end{rem}

\subsection{Protocol For Strong Expander Graphs}

In this subsection we present one protocol for strong expander graphs, and prove \thmref{expander_result}.

Let $\mathcal{G}=\{G\}_i$ be a family of graphs. We call $\mathcal{G}$ a family of  \emph{strong expander graphs} if every $G_i$ in $\mathcal{G}$ has spectra gap $\alpha=1-o(1)$, and irregularity $\beta=1+o(1)$. This graph family includes several interesting graphs, e.g. Ramanujan graphs,  complete graphs, random graphs $G(n,p)$ with $p=\omega(\log n/ n)$, and random $d$-regular graph where $d$ is any increasing function of $n$.
The formal description of our  protocol is as follows:

\begin{protocol}[Protocol for Strong Expander Graphs]\label{pro_PRG}
Let $\epsilon=\Delta^{-\Theta(1)}$ be sufficiently small, $\epsilon'=2^{-\sqrt{\log\log n}}$, and $m=\Theta((\log n)/\epsilon)$ a power of 2. Pick the following objects:
\begin{itemize}\itemsep -0.3pt
\item An explicit $(\mathord{\leq} K, (1-\epsilon^2/4)D)$-expander $\Gamma: [n^c]\times [D] \to \bigsqcup_{i\in [D]}[M_i]$, where $K=\Delta$, $D=\left(({\log n})/{\epsilon}\right)^{O(1)}$ and $M_0=\dots=M_{D-1}=M\leq \max\{D,\Delta^{O(1)}\}$.

\item An explicit function $\gen=(\gen_1,\dots,\gen_M):\{0,1\}^{\ell}\to [m]^M$ that is both a pairwise independent generator and an $\epsilon'$-\textsf{PRG} for $\mathsf{CR}_{[m]^M}$, where $\ell=O(\log m+\log M)+\tilde{O}(\log(1/\epsilon'))=O(\log\log n+\log\Delta)$.
\end{itemize}
These two objects $\gen$ and $\Gamma$ can be uniquely constructed from $n^c$ and $\Delta^{\Theta(1)}$, and hence are known to every informed node.

The initial node having the rumor chooses a random string $(s_1,\dots,s_{T})$ where every $s_i$ is of the form $(x_i,y_i)\in [D]\times\{0,1\}^{\ell}$.
This random string  is appended with the rumor and sent to other nodes. Once one node gets the rumor, it gets the ID $u$. In the $i$th round, node $u$ computes $r=\Gamma(u,x_i)$ that is in $[M_u]$, the $u$th copy of $[M]$. It then chooses the neighbor with index $\gen_r(y_i) \bmod \deg(u)$ in its adjacency list to send the rumor.
\end{protocol}

\begin{pro}
Assume that Protocol~\ref{pro_PRG} finishes in $T$ rounds. Then it uses $O(T\cdot(\log\log n+\log\Delta))$ random bits in total.
\end{pro}

\paragraph{Acknowledgement.}
We are grateful to  Chris Umans for many hours of stimulating discussion and improving the presentation of the paper. We
would like to thank Luca Trevisan and Avi Wigderson for helpful discussion about our work.

\bibliographystyle{plain}
\bibliography{pseudoRS}
\appendix

\section{Notations \& Useful Lemmas}

In this section we list all notations used in the paper.
Let $G=(V,E)$ be a connected, undirected, and simple graph with $n$ nodes. For any node $u$, $\deg(u)$ stands for the degree of $u$. The  maximum, minimum, and average degree of $G$ are represented by $\Delta$, $\delta$, and $d$. Let $\beta\triangleq\Delta/\delta$ be the \emph{irregularity} of graph $G$. The set of neighbors of an node $u$ is represented by $N(u)$. Moreover, for any set $S\subseteq V$,
 let
 $N(S)\triangleq\bigcup_{u\in S} N(u)$, and $\mathrm{vol}(S)\triangleq\sum_{u\in S}\deg(u)$. For any set $S, T\subseteq V$, we define $E(S,T)\triangleq\{ \{u,v\}: u\in S\mbox{~ and ~} v\in T \}$ and $e(S,T)\triangleq|E(S,T)|$.

 We use $\mat{A}_G$ to express the adjacency matrix of $G$. Let  $\mat{D}$  the $n\times n$ diagonal matrix defined by $\mat{D}_{uu}=\deg(u)$ for $u\in V[G]$.
Let $\mat{M}_G=\mat{D}^{-1}\mat{A}_G$ be the transition matrix for the random walk over $G$, and $
\mat{N}_G\triangleq \mat{D}^{-1/2}\mat{A}_G\mat{D}^{-1/2}$.
Define the $n$ real eigenvalues of $\mat{N}_G$ by $1=\lambda_1\geq\cdots\geq \lambda_n\geq -1$, and let $\lambda_{\max}\triangleq\max\{ \lambda_2, |\lambda_i|\}$. The spectral gap $\alpha$ is defined by $\alpha\triangleq 1-\lambda_2$, whereas the absolute spectral gap is defined as $1-\lambda_{\max}$. For simplicity, we also use $\alpha$ to express the spectral expansion of a reversible Markov chain if the chain is clear from the context.

For $m\in\mathbb{N}$, vector $\vec{u}\in \mathbb{R}^m$ and real number $p\geq 1$, define the $\ell_p$-norm $\|\vec{u}\|_p=\left(\sum_{i=1}^m |\vec{u}_i|^p\right)^{1/p}$. In addition, we define $\|\vec{u}\|_\infty=\max_{1\leq i\leq m}|\vec{u}_i|$.
The inner product of two vectors $\vec{u},\vec{v}\in\mathbb{R}^{m}$ is $\langle\vec{u},\vec{v}\rangle=\sum_{i=1}^m\vec{u}_i\vec{v}_i$.
We write $\vec{1}_m$ for the vector in $\mathbb{R}^m$ having ones in all entries, or simply $\vec{1}$ if the dimension is clear from the context. Similarly write $\vec{0}_m$ or $\vec{0}$ for the zero vector.
Let $\vec{e}_i$ be the vector that has an one in the $i$th entry and zero elsewhere.
Write $\mat{I}_m$ or $\mat{I}$ for the $m\times m$ identity matrix.
For a matrix $\mat{M}\in\mathbb{R}^{m\times m'}$, we use $\mat{M}_{ij}$ to denote the entry on $\mat{M}$'s $i$th row and $j$th column.
For $p\in [1,\infty) \cup\{\infty\}$, define $$
\|\mat{M}\|_p=\sup_{\vec{u}\in\mathbb{R}^m\setminus\{\vec{0}\}}\frac{\|\vec{u}\mat{M}\|_p}{\|\vec{u}\|_p}.
$$
It is easy to show that $\|\mat{M}\|_1$ equals the maximum of the $\ell_1$-norms of the rows of $\mat{M}$.
And $\|\mat{M}\|_\infty$ equals the maximum of the $\ell_1$-norms of the columns of $\mat{M}$, or equivalently $\|\mat{M}^\intercal\|_1$.
We say a square matrix $\mat{M}$ is stochastic if all of its entries are non-negative and all of its rows have $\ell_1$-norm 1. Clearly if $\mat{M}$ is stochastic, then $\|\mat{M}\|_1=1$. We say $\mat{M}$ is doubly-stochastic if both $\mat{M}$ and $\mat{M}^\intercal$ are stochastic.

 By $\log x$ we denote the binary logarithm of $x$.
For any integer $m$, define $[m]\triangleq\{0,\ldots,m-1\}$.
%The product distribution of two distributions $X$ and $Y$ is denoted by $X\times Y$.
The disjoint union of a family of sets $\{A_i: i\in I\}$ indexed by $I$ is denoted by $\bigsqcup_{i\in I} A_i\triangleq\bigcup_{i\in I}\{(x,i):x\in A_i\}$.
%$:= \bigcup_{i\in I}\{(x,i):x\in A_i\}$. We usually abuse the notion by saying $A_i$ is a subset of $\bigsqcup_{i\in I} S_i$, where we are identifying $A_i$ with $\{(x,i): x\in A_i\}$.
With high probability stands for with probability $1-n^{-\Theta(1)}$.

\begin{lem}\label{lem:geometric}
Fix any $0< p <1$ and let $X_1,\ldots,X_n$ be
independent geometric random variables on $\mathbb{N}$ with $\Pro{X_i = k} =
(1-p)^{k-1} p$ for every $k \in \mathbb{N}$. Let $X = \sum_{i=1}^n
X_i$, and $\mu=\Ex{X}$. Then it holds for all $\beta > 0$ that
\[
  \Pro { X \ge  (1+\beta) \mu } \le \mathrm{e}^{-n\beta^2/(2 (1+\beta))}.
\]
\end{lem}

\begin{fact}[\cite{SL97}] \label{fact_spectralgap}
The spectral gap of a graph $G$ satisfies
$$
\alpha=\inf_{\vec{u}\not\parallel\vec{1}}\frac{\mathcal{E}_{\vec{\pi},G}(\vec{u},\vec{u})}{\mathrm{Var}_{\vec{\pi}}(\vec{u})}
$$
where $\pi$ is the stationary distribution of $\mat{M}_G$, and the quantities
$$
\mathrm{Var}_{\vec{\pi}}(\vec{u})=\frac{1}{2}\sum_{u,v\in V[G]} \vec{\pi}_u\vec{\pi}_v(\vec{u}_u-\vec{u}_v)^2,
\qquad
\mathcal{E}_{\vec{\pi},G}(\vec{u},\vec{u})=\frac{1}{2}\sum_{u,v\in V[G]} \vec{\pi}_u(\mat{M}_G)_{uv}(\vec{u}_u-\vec{u}_v)^2
$$
are known as the global variance and the local variance (or Dirichlet form) of $\vec{u}$ respectively.
\end{fact}

We also need an operation on graphs, called {\em regularization}. Formally speaking, for an undirected graph $G$ with maximal degree $\Delta$, let $\reg(G)$ be the regular graph obtained from $G$ by adding $\Delta-\deg(u)$ self-loops to each node $u\in V[G]$.

\mylemma{lemregularization}{\lemregularization}

\begin{proof}
Let $\vec{\pi}$ and $\vec{\pi}'$ be the stationary distributions of $\mat{M}_G$ and $\mat{M}_{\reg(G)}$ respectively, i.e. $\vec{\pi}_u=\deg(u)/(n\cdot d)$ and $\vec{\pi}'_u=1/n$ for any node $u\in V$. Then for any $\vec{u}\not\parallel\vec{1}$, we have
$$
\frac{\mathcal{E}_{\vec{\pi}',\reg(G)}(\vec{u},\vec{u})}{\mathcal{E}_{\vec{\pi},G}(\vec{u},\vec{u})}
\geq \min_{u\neq v}
\frac{\vec{\pi}'_u(\mat{M}_{\reg(G)})_{uv}}
{\vec{\pi}_u(\mat{M}_G)_{uv}}
= \min_{u\in V[G]}\frac{d\cdot \Delta^{-1}}{\deg(u)\cdot (\deg(u))^{-1}}=d/\Delta
$$
and
$$
\frac{\mathrm{Var}_{\vec{\pi}'}(\vec{u})}{\mathrm{Var}_{\vec{\pi}}(\vec{u})}
\leq \min_{u\neq v}\frac{\vec{\pi}'_u\vec{\pi}'_v}{\vec{\pi}_u\vec{\pi}_v}
\leq\min_{u\neq v} \frac{(1/n)\cdot (1/n)}{ (\deg(u)/nd)\cdot (\deg(v)/nd) }
 =
 \min_{u\neq v}\frac{d^2}{\deg(u)\deg(v)}\leq d^2/\delta^2.
$$
So
$$
\left.\frac{\mathcal{E}_{\vec{\pi}',\reg(G)}(\vec{u},\vec{u})}{\mathrm{Var}_{\vec{\pi}'}(\vec{u})}
\middle/
\frac{\mathcal{E}_{\vec{\pi},G}(\vec{u},\vec{u})}{\mathrm{Var}_{\vec{\pi}}(\vec{u})}
\right.
\geq (d/\Delta)\cdot(\delta^2/d^2)\geq \beta^{-2}
$$
and the claim follows from Fact \ref{fact_spectralgap}.
\end{proof}

\section{Existential Proof}

In this section we show that $O(\log n)$ random bits are sufficient in rumor spreading for many classes of graphs (e.g. complete graphs, strong expanders, graphs with good conductance, etc.) if we do not care about the computational complexity. We will prove the following general statement:

\begin{lem}\label{lem:existence}
Let $\mathcal{C}$ be a class of graphs on $n$ nodes with no multi-edges. Let $T'=n^{O(1)}$ be an upper bound of spreading time. Suppose the spreading time for any graph in $\mathcal{C}$ is at most $T$ with probability $p$ for fully-random push protocol. Then there exists a (non-explicit) function
$$
f: \{0,1\}^{\ell} \times [n]\times [T']\times [\Delta] \to [\Delta]
$$
such that
\begin{enumerate}
\item $f(x,u,t,d)\in [d]$ for all $(x,u,t,d)\in\{0,1\}^{\ell} \times [n]\times [T']\times [\Delta]$.
\item $\ell=\max\{ \log\log |C|, \log n+\log\Delta+\log\log\Delta \}+2\log(1/\epsilon)+O(1)$.
\item for $x$ uniformly chosen from $\{0,1\}^\ell$, the spreading time for any graph  $G\in\mathcal{C}$ is at most $T$ with probability $p-\epsilon$ if node $u$ uses  $f(x,u,t,\deg(u))\in [\deg(u)]$ as the index of its receiver in its adjacency list in round $t$.
\end{enumerate}

In particular, $\ell$ is bounded by $2\log n+\log\log n+2\log(1/\epsilon)+O(1)$ since $|\mathcal{C}|\leq 2^{n^2}$ and $\Delta\leq n$.
\end{lem}

\begin{proof}
Choose $f(x,u,t,d)\in [d]$ independently and uniformly at random for each $(x,u,t,d)\in\{0,1\}^{\ell} \times [n]\times [T']\times [\Delta]$.
Fix a graph $G\in\mathcal{C}$ and an initial node in $[n]$.
For each node $u$ in the graph of degree $\deg(u)$, there are $\deg(u)!$ possible orders of neighbors of $u$ in its adjacency list. We also fix the order for each node $u$.
Observe that for any fixed $x$, the random variables $f(x,u,t,\deg(u))$ for all pairs $(u,t)$ are  independent and uniformly distributed. Let $I(x)$ be the indicator random variable that equals $1$ if the spreading time of $G$ is at most $T$ when node $u$ uses $f(x,u,t,\deg(u))$ to decide its receiver in round $t$. Then $\Prob{f}{I(x)=1}\geq p$ for any $x$ and hence $\EXX{f}{I(x)}\geq p$.
Also note that $I(x)$'s are independent. By the Chernoff bound it holds that
$$
\Prob{f}{\left|2^{-\ell}\sum_x I(x) - 2^{-\ell}\sum_x \EXX{f}{I(x)}\right| \geq \epsilon} \leq 2\exp(-2^\ell\epsilon^2/4).
$$
So with probability at least $1-2\exp(-2^\ell\epsilon^2/4)$, we have $\EXX{x}{I(x)}\geq \EXX{x}{\EXX{f}{I(x)}}-\epsilon\geq p-\epsilon$.
By the union bound, the probability that $\EXX{x}{I(x)}\geq p-\epsilon$ holds for all graphs in $\mathcal{C}$, arbitrary neighboring list of nodes, and all start nodes is at least
$$
1-n|\mathcal{C}|\cdot (\Delta!)^n\cdot 2\exp(-2^\ell\epsilon^2/4),
$$
which is greater than zero for sufficiently large $\ell=\max\{ \log\log |C|, \log n+\log\Delta+\log\log\Delta \}+2\log(1/\epsilon)+O(1)$. So there exists one function $f$ such that $\EXX{x}{I(x)}\geq p-\epsilon$ holds for all graphs in $\mathcal{C}$, i.e. the spreading time for any graph $G\in\mathcal{C}$ is at most $T$ with probability $p-\epsilon$ over the choices of $x$, if node $u$ uses $f(x,u,t,\deg(u))\in [\deg(u)]$ to choose its receiver in round $t$.
\end{proof}

The same result also holds for pull protocols and push-pull protocols, and can be shown using similar arguments.

The following result follows from \lemref{existence} directly.

\begin{cor}[Existential Result]\label{thm:existence} Let $\mathcal{G}=\{G_n\}_{ n\geq 1}$ be a family of graphs such that for any $G_n\in\mathcal{G}$ with $n$ nodes the truly random protocol finishes in $T=n^{O(1)}$ rounds with high probability. Then there is a protocol which finishes in $T$ rounds with high probability and uses $3\log n$ random bits in total.
\end{cor}

\section{Lower Bounds on Randomness Complexity}

We address the randomness requirement of rumor spreading protocols. We first introduce the pull model, which is a symmetric version of the push model, and the formal description is as follows: In round $t\geq 0$, every node $u$ that does not yet have the rumor selects a neighbor $v$ uniformly at random and asks for the rumor, and gets the rumor if $v$ received the rumor before. In the push-pull model, in every round $t$, every node $u$ chooses a random neighbor to perform \emph{push} if node $u$ has the rumor, or perform \emph{pull} if $u$ has not received the rumor.

We prove the following lower bound on the number of random bits needed for any protocol in the push-pull model:

\begin{thm}\label{thm:lb1}
Let $G$ be any graph with $n$ nodes and sufficiently large minimum degree $\delta=\Omega(\log n)$. Then any protocol in the push-pull model that is oblivious of the order of adjacency lists of $G$ and informs at least half of the nodes of $G$ in $T$ rounds with nonzero probability has to use more than $\log\delta-\log T-2$ random bits. In particular, $\Theta(\log n)$ random bits are necessary when $\delta=\Theta(n)$ and $T=O(n^{1-\epsilon})$ for some constant $\epsilon>0$.
\end{thm}

Here we even allow the protocol access to the ID of the initial node and the structure of $G$, i.e., the sets of neighbors of nodes as {\em unordered sets}.
In addition, we allow each node access to the randomness even before it obtains the rumor.
All we assume is that the protocol is oblivious of the {\em order} of the adjacency lists.

\begin{proof}
Suppose $V[G]=[n]$.
Let $\Delta$ be the maximum degree of $G$ and $s$ be the initial node.
We first claim that there exists a subset of nodes $S$ of size $n/2$ (for simplicity assume $n$ is even) such that $\deg(u)/4\leq |S\cap N(u)|\leq 3\deg(u)/4$ for all $u\in [n]$: If we pick a random subset $S$ of size $n/2$, then for any fixed $u$ the condition $\deg(u)/4\leq |S\cap N(u)|\leq 3\deg(u)/4$ holds, by the Chernoff bound, with probability at least $1-\ce^{-\Theta(\delta)}>1-1/n$ for $\delta=\Omega(\log n)$ sufficiently large. The claim then follows by taking the union bound. Pick such a subset $S$ with the claimed property. Note that $[n]\setminus S$ has the same property. We may therefore assume $s\in S$ by swapping $S$ and $[n]\setminus S$ if necessary.

A protocol for $G$ using $\ell$ random bits in $T$ rounds is uniquely characterized by a pair of functions
$$
f_1, f_2: \{0,1\}^{\ell} \times [n]\times [T]\times [\Delta] \to [\Delta]
$$
satisfying $f_1(x,u,t,d),f_2(x,u,t,d)\in [d]$ for all $(x,u,t,d)\in\{0,1\}^{\ell} \times [n]\times [T]\times [\Delta]$,
in the sense that given the random string $x$, node $u$ chooses a neighbor with index $f_1(x,u,t,\deg(u))$ (resp. $f_2(x,u,t,\deg(u))$) in its adjacency list to push (resp. pull) the message in round $t$ if it is informed (resp. uninformed).
For each $u\in [n]$, define $I_u\subseteq [n]$ as
$$
I_u=\begin{cases}\{f_1(x,u,t,\deg(u)): x\in\{0,1\}^\ell, t\in [T]\} & u\in S\\
\{f_2(x,u,t,\deg(u)): x\in\{0,1\}^\ell, t\in [T]\} & u\not\in S.
\end{cases}
$$
Assume to the contrary that $\ell\leq \log\delta-\log T-2$. Then the size of $I_u$ is at most $2^\ell\cdot T\leq\delta/4\leq \min\{|S\cap N(u)|, |([n]\setminus S)\cap N(u)|\}$ for each $u\in [n]$.
So it is possible to order the adjacency list of each $u\in [n]$ such that the neighbors picked by $u$ using index set $I_u$ are all in $S\cap N(u)$ if $u\in S$, or in $([n]\setminus S)\cap N(u)$ if $u\in [n]\setminus S$. Then in the rumor spreading process, nodes in $S$ push messages only to those also in $S$, and nodes in $[n]\setminus S$ pull messages only from those also in $[n]\setminus S$. As $s\in S$, the nodes in $[n]\setminus S$ never get informed.
\end{proof}

For the push model and the pull model we may drop the assumption that $\delta=\Omega(\log n)$ is sufficiently large, and also simplify the proof.

\begin{thm}\label{thm:lb_push}
Let $G$ be any graph with $n$ nodes. Then any protocol in the push model that is oblivious of the order of adjacency lists of $G$ and informs all the nodes of $G$ in $T$ rounds with nonzero probability has to use more than $\log(\delta-1)-\log T$ random bits.
\end{thm}
\begin{proof}
The protocol is now characterized by a single function $f_1$ describing how rumors are pushed.
Define $I_u=\{f_1(x,u,t,\deg(u)): x\in\{0,1\}^\ell, t\in [T]\}$ for each $u\in [n]$.
Pick $v\in [n]\setminus\{s\}$.
Assume to the contrary that $\ell\leq \log(\delta-1)-\log T$. Then the size of $I_u$ is at most $2^\ell\cdot T\leq\delta-1\leq |N(u)\setminus\{v\}|$ for each $u\in [n]$.
So it is possible to order the adjacency list of each $u\in [n]$ such that the neighbors picked by $u$ using index set $I_u$ are all in $N(u)\setminus \{v\}$. Then the node $v$ never gets informed.
\end{proof}

\begin{thm}\label{thm:pushlowerbound}
Let $G$ be any graph with $n$ nodes. Then any protocol in the pull model that is oblivious of the order of adjacency lists of $G$ and informs more than one node of $G$ in $T$ rounds with nonzero probability has to use more than $\log(\delta-1)-\log T$ random bits.
\end{thm}
\begin{proof}
The protocol is now characterized by a single function $f_2$ describing how rumors are pulled.
Define $I_u=\{f_2(x,u,t,\deg(u)): x\in\{0,1\}^\ell, t\in [T]\}$ for each $u\in [n]$.
Assume to the contrary that $\ell\leq \log(\delta-1)-\log T$. Then the size of $I_u$ is at most $2^\ell\cdot T\leq\delta-1\leq |N(u)\setminus\{s\}|$ for each $u\in [n]$.
So it is possible to order the adjacency list of each $u\in [n]$ such that the neighbors picked by $u$ using index set $I_u$ are all in $N(u)\setminus \{s\}$. Then the nodes in $[n]\setminus\{s\}$ never get informed.
\end{proof}

\section{Omitted Details in Section~\ref{sec_expander}}

\subsection{Preliminaries}

We first list  definitions and results about pseudorandom generators.

\paragraph{Pairwise Independent Generators.}

\begin{defi}[Pairwise Independent Generator]
We say $X_0,\dots,X_{d-1}$ with $X_i$ distributed over $[m_i]$ are $\epsilon$-pairwise independent if
\begin{itemize}
\item $\left|\Pro{X_i=x}-\frac{1}{m_i}\right|\leq\epsilon$ for all $i\in [d]$ and $x\in [m_i]$, and
\item $\left|\Pro{X_i=x\wedge X_j=x'}-\frac{1}{m_i\cdot m_j}\right|\leq\epsilon$
for all distinct $i,j \in [d]$ and all $x\in [m_i]$, $x'\in [m_j]$.
\end{itemize}
We say they are pairwise independent  if $\eps=0$.
We say $\gen:\{0,1\}^{\ell}\to [m_0]\times\dots\times [m_{d-1}]$ is an ($\epsilon$-)pairwise independent generator if its outputs are ($\epsilon$-)pairwise independent given a uniformly distributed seed.
\end{defi}

\begin{thm}[\cite{CW79}]\label{thm:cw}
There exists an explicit pairwise independent generator $\gen:\{0,1\}^\ell\to [m]^d$ with seed length $\ell=O(\log m+\log d)$.
\end{thm}

\newcommand{\pairwisemodulo}{
Suppose $\gen=(\gen_0,\dots,\gen_{d-1})$ is a pairwise independent generator where $\gen_i:\{0,1\}^{\ell}\to [m]$. Define $\gen'=(\gen'_0,\dots,\gen'_{d-1})$ where $\gen'_i(x)=\gen_i(x)\bmod m_i$ for $i\in [d]$.
Then $\gen':\{0,1\}^{\ell}\to[m_0]\times\cdots\times[m_{d-1}]$ is an $\epsilon$-pairwise independent generator where $\epsilon=2/m$.}
\mylemma{pairwisemodulo}{\pairwisemodulo}

\begin{proof}
For distinct $i,j\in [d]$ and $x\in [m_i]$, $x'\in [m_j]$, let $B$ (resp. $B'$) be the preimages of $x$ (resp. $x'$) under the map $s \mapsto s\bmod m_i$ (resp. $s \mapsto s\bmod m_j$). Then $||B|-m/m_i|\leq 1$ and $||B'|-m/m_j|\leq 1$.
So $\mathbf{Pr}_s[\gen'_i(s)=x]=|B|/m$ which differs from $1/m_i$ by at most $1/m$.
Similarly $\mathbf{Pr}_s[\gen'_i(s)=x \wedge \gen'_j(s)=x']=|B||B'|/m^2$ which differs from $1/(m_i m_j)$ by at most $2/m$.
\end{proof}

\newcommand{\rectanglemodulo}{
Suppose $\gen=(\gen_0,\dots,\gen_{d-1})$ is an $\epsilon$-\textsf{PRG} for $\mathsf{CR}_S$ where $S=[m]^d$.
Define $\gen'=(\gen_0',\dots,\gen_{d'-1}')$ where $\gen'_j(x)=\gen_{i_j}(x)\bmod m_j$ for $j\in [d']$ and $i_0,\dots,i_{d'-1}\in [d]$.
Then $\gen'$ is an $(\epsilon+\sum_{i\in [d']}m_i/m)$-\textsf{PRG} for $\mathsf{CR}_{S'}$
where $S'=\prod_{i\in [d']}[m_i]$.
}
\mylemma{rectanglemodulo}{\rectanglemodulo}

\begin{proof}
By definition, $\gen'=\pi\circ\gen$ with $\pi:(x_0,\dots,x_{d-1})\mapsto (x_{i_0} \bmod m_0,\dots, x_{i_{d'-1}} \bmod m_{d'-1})$.
For $A=\prod_{i\in [d']} A_i \in\mathsf{CR}_{S'}$, let $B=\pi^{-1}(A)=\prod_{i\in [d]} B_i \in\mathsf{CR}_{S}$.
Then $\mathbf{Pr}_s\left[ \gen'(s)\in A\right]=\mathbf{Pr}_s\left[\gen_i(s)\in B\right]$ which differs from $|B|/|S|$ by at most $\epsilon$ since $\gen$ is an $\epsilon$-\textsf{PRG} for  $\mathsf{CR}_{S}$.
Note that $|B_{i_j}|/m$ differs from $|A_j|/m_j$ by at most $m_j/m$ for $j\in [d']$, and $B_i=[m]$ for $i\in [d]\setminus\{i_0,\dots,i_{d'-1}\}$.
A simple induction shows that $|B|/|S|$ differs from $|A|/|S'|$ by at most
\[
\sum_{j\in [d']} ||B_{i_j}|/m-|A_j|/m_j|\leq \sum_{i\in [d']}m_i/m.\qedhere
\]
\end{proof}

\begin{thm}[\cite{GMRTV12}]\label{thm_cr}
Let $S=[m]^d$.
There exists an explicit $\epsilon$-\textsf{PRG} for $\mathsf{CR}_S$ with seed length $O(\log m+\log d)+\tilde{O}(\log(1/\epsilon))$.
\footnote{In \cite{GMRTV12} the seed length is presented as $O((\log\log m)(\log m+\log d+\log(1/\epsilon)))+\tilde O(\log(1/\epsilon))$. But there are techniques of reducing $m$ and $d$ to $m'=(1/\epsilon)^{O(1)},d'=(1/\epsilon)^{O(1)}$ using $O(\log m+\log d)$ randomness, cf.~\cite{ASWZ96, Lu02}.}
\end{thm}

\newcommand{\prg}{
There exists an explicit function $\gen:\{0,1\}^\ell\to [m]^d$ that is both a pairwise independent generator and an $\epsilon$-\textsf{PRG} for $\mathsf{CR}_{[m]^d}$ with seed length $O(\log m+\log d)+\tilde{O}(\log(1/\epsilon))$.
}
\mylemma{prg}{\prg}

\begin{proof}
Pick an explicit pairwise independent generator $\gen^{\flat}:\{0,1\}^{\ell_1}\to [m]^d$ with see length $\ell_1=O(\log m+\log d)$ and an explicit $\epsilon$-\textsf{PRG} $\gen^{\sharp}:\{0,1\}^{\ell_2}\to [m]^d$ for $\mathsf{CR}_{[m]^d}$ with seed length $\ell_2=O(\log m+\log d)+\tilde{O}(\log(1/\epsilon))$.
Identify $[m]^d$ with $\mathbb{Z}_m^d$ and define $\gen:\{0,1\}^{\ell_1+\ell_2}\to [m]^d$ using addition in $\mathbb{Z}_m^d$: $\gen(x,y)=\gen^{\flat}(x)+\gen^{\sharp}(y)$.
The definition of pairwise independent generators implies that the function $x\mapsto \gen^{\flat}(x) + z$ is a pairwise independent generator for any fixed $z\in [m]^d$, i.e., the property is preserved under addition of any fixed element in $\mathbb{Z}_m^d$. Then the same is true for random $z=\gen^{\sharp}(y)$. So $\gen$ is a pairwise independent generator. A similar argument shows that it is also an $\epsilon$-\textsf{PRG} for $\mathsf{CR}_{[m]^d}$.
\end{proof}

\begin{defi}[Branching Programs]\label{def:BranchingProgram}
A branching program of length $L$, width $W$ and degree $D$, or an $(L,W,D)$-branching program, is a directed (multi)-graph with  node set $[W]\times \{0,\dots,L\}$. We say the nodes in $[W]\times \{i\}$ are on the $i$th layer for $0\leq i\leq L$.
Each node $(u,i)$ except those on the last layer has $D$ outgoing edges to nodes on the next layer, and these $D$ edges are associated with $D$ distinct labels from $[D]$.
\end{defi}

\begin{thm}[\cite{INW94}]\label{thm:inw}
There exists an explicit $\epsilon$-\textsf{PRG} for $(L,W,D)$-branching programs with seed length $O(\log L(\log W+\log L+\log(1/\epsilon))+\log D)$.
\end{thm}

The following lemma about Markov chains will be used in the analysis. For an ergodic Markov chain represented by the stochastic matrix $\mat{M}$ and $\epsilon>0$, define its $\ell_2$-mixing time as
$$
\tau_{\mat{M}}(\epsilon)=\max_{\vec{u}}\min\{k: \|\vec{u}\mat{M}^k-\vec{\pi}\|_2\leq\epsilon\},
$$
where $\vec{\pi}$ is the stationary distribution of $\mat{M}$ and $\vec{u}$ ranges over all distributions over the state set of the chain.

\begin{lem}[\cite{SL97}]\label{lemmixing0}
Suppose $\mat{M}\in\mathbb{R}^{V[G]\times V[G]}$ represents a reversible Markov chain with absolute spectral gap $\alpha>0$.
Then $\tau_{\mat{M}}(\epsilon)<\log_{1-\alpha}\epsilon + 1$.
\end{lem}

\subsection{Analysis of the Prototype Protocol}

In this subsection we give the detailed analysis of Protocol~\ref{pro_prototype}. We start with the formal definition of reversed random walks.
The basic idea is to view a push operation (or one step of a forward walk) as a pull operation (or one step of a reversed walk). However, there are several complications: (1) we let $v$ ``pull the from $u$'' only when $u$ is the unique node pushing to $v$, since $v$ is not allowed to pull from multiple nodes at the same time; (2) we need to use auxiliary randomness $r_{i,u}$ to equalize the probabilities of successful pulls made by different nodes;\footnote{The auxiliary randomness only appears in the analysis.} (3) we want the pull operations to be pairwise independent. In particular two nodes $u$ and $v$ pull from their common neighbor $w$ at the same time with probability $1/\Delta^2$. To realize this, we combine two rounds into one so that $w$ can send two messages, say to $a$ and $b$ at the same time. Also, note that there are two cases when $w$ pushes to both $u$ and $v$, or equivalently $u$ and $v$ both pull from $w$: $(a,b)=(u,v)$ or $(a,b)=(v,u)$. We admit only one of them, so that the event occurs with probability $1/\Delta^2$ rather than $2/\Delta^2$.

As before, for $f\sim\mathcal{D}$, we denote by $\tilde{f}(i,u)$ the $f(i,u)$th neighbor of $u$ in its adjacency list.

\begin{defi}[Reversed random walks]\label{def:RevsersedRW}
Consider a random rumor spreading process in $T$ rounds on a graph $G$ using Protocol \ref{pro_prototype} determined by its own randomness $f\sim\mathcal{D}=\mathcal{U}$.
Pick real numbers $r_{i,u}$  independently and uniformly from $[0,1]$ for all $i\in [T/2]$ and $u\in V[G]$.
%We say $u$ is {\em active} in round $i$ if $r_{i,u}\leq\deg(u)/\Delta$, and otherwise {\em inactive}.

Fix an arbitrary total order $\preceq$ on $V[G]$.
For $i\in [T/2]$ and $u\in V[G]$, define
$$
N_{i,u}=\begin{cases}
\left\{\tilde{f}(T-1-2i,u),\tilde{f}(T-2-2i,u)\right\} &  \tilde{f}(T-1-2i,u)\preceq \tilde{f}(T-2-2i,u)\\
\emptyset & \text{otherwise}
\end{cases}
$$
and define $N_{i,u}^\vee=\left\{v\in V[G]: v\neq u \text{ and } u\in N_{i,v}\right\}$.
%Also let $\deg'(u)$ be the number of neighbors of $u$ other than $u$ itself.
%\textcolor{red}{redefine $\deg'(u)$ here.}

A reversed random walk of length $k\in [T/2]$
with pattern $S=(s_0,\dots,s_{k-1})\in\mathcal{C}_k$ is a sequence of $k+1$ nodes $(p_0,\dots,p_k)$ of $G$, such that for all $i\in [k]$: (i) if $s_i=\text{lazy}$, then $p_{i+1}=p_i$; (ii) if $s_i=\text{non-lazy}$, then $p_{i+1}=u$ if  $N_{i,p_i}^\vee=\{u\}$ is a singleton and $r_{i,p_i}\leq (1-1/\Delta)^{\Delta-\deg(u)}$, and otherwise $p_{i+1}=p_i$.
%In particular, each node itself is a forward/reversed random walk of length $0$.
%We write $p:s\leadsto v$ to indicate that the colored random walk $p$ is from $s$ to $v$, or write $p:s\leadsto_k v$ with its length $k$ specified.  For a colored random walk $p$ of length at least 1, We write $p^-$ for the colored path obtained by discarding the last node of $p$.
\end{defi}

\subsubsection{Approximation via Random Walks} We elaborate the idea of bounding runtime of Protocol~\ref{pro_prototype} with respect to multiple random walks. We will use three distributions in the following analysis:
\begin{itemize}
\item $\mathcal{D}_{\gamma, k}$ is the distribution over $\mathcal{C}_k$ where entries are independently chosen to be lazy with probability $1-\gamma$.
\item Let $r=(f, \{r_{i,u}\})$ be the whole randomness used in
Definition \ref{defi_walk} and Definition \ref{def:RevsersedRW}
%Protocol~\ref{pro_prototype}
which determines
%the rumor spreading process as well as
the random walks. Let $\tilde{\mathcal{D}}$ be the distribution of $r$, which is the product of $\mathcal{D}$ with copies of uniform distributions over $[0,1]$.
\end{itemize}

The following lemma give a lower bound of the probability that a node $w$ gets informed in $T$ rounds with respect to multiple random walks.

\mylemma{lembound}{\lembound}

\begin{proof}
Define the weight of forward or reversed random walks with pattern $S=(s_0,\dots,s_k)$ as $\mathrm{wt}(S):=(1-\gamma)^{n_1}\gamma^{n_2}>0$ where $n_1$ and $n_2$ are the number of lazy and non-lazy $s_i$ respectively.
Let $\tilde{X}^S_{u,v}=X^S_{u,v}\cdot \mathrm{wt}(S)$ and $\tilde{Y}^S_{u,v}=Y^S_{u,v}\cdot \mathrm{wt}(S)$.

Suppose $s\in V[G]$ is the initial node and also fix $w\in V[G]$. If there exist a forward walk $p$ from $s$ to some node $u$ and a reversed walk $p'$ from $w$ to $u$, then the rumor is sent from $s$ to $u$ following $p$ and then from $u$ to $w$ following the reversal of $p'$.
Also note that the two walks exist   if and only if $\tilde{X}^S_{s,u}\tilde{Y}^{S'}_{w,u}>0$ for some $S$, $S'$ and $u$.
Therefore,
\begin{equation}\label{eq:pro}
\Prob{f\sim \mathcal{D}}{t \text{ receives the message}}\geq \Prob{r\sim\tilde{\mathcal{D}}}{\sum_{S,S'\in\mathcal{C}_k, u\in V[G]} \tilde{X}^S_{s,u}\tilde{Y}^{S'}_{w,u}>0}.
\end{equation}
Furthermore,
\begin{equation}\label{eq:approx}
\begin{aligned}
\lefteqn{\Prob{r\sim\tilde{\mathcal{D}}}{\sum_{S,S'\in\mathcal{C}_k, u\in V[G]} \tilde{X}^S_{s,u}\tilde{Y}^{S'}_{w,u}>0}}\\
&=\EXX{r\sim\tilde{\mathcal{D}}}{\mathbf{1}_{\sum_{S,S'\in\mathcal{C}_k, u\in V[G]} \tilde{X}^S_{s,u}\tilde{Y}^{S'}_{w,u}>0}}\\
&\geq \frac{\left(\EXX{r\sim\tilde{\mathcal{D}}}{\sum_{S,S'\in\mathcal{C}_k, u\in V[G]} \tilde{X}^S_{s,u}\tilde{Y}^{S'}_{w,u}}\right)^2}{\EXX{r\sim\tilde{\mathcal{D}}}{\left(\sum_{S,S'\in\mathcal{C}_k, u\in V[G]} \tilde{X}^S_{s,u}\tilde{Y}^{S'}_{w,u}\right)^2}}
\\
&=\frac{\sum_{u,v\in V[G]}\EXX{r,S}{X^S_{s,u}}\EXX{r,S}{X^{S}_{s,v}}
\EXX{r,S}{Y^S_{w,u}}\EXX{r,S}{Y^{S}_{w,v}}}
{\sum_{u,v\in V[G]}\EXX{r,S,S'}{X^S_{s,u}X^{S'}_{s,v}}
\EXX{r,S,S'}{Y^S_{w,u}Y^{S'}_{w,v}}}
\end{aligned}
\end{equation}
where the subscripts $r$, $S$ and $S'$ are independent with distributions $\tilde{\mathcal{D}}$,  $\mathcal{D}_{\gamma,k}$ and $\mathcal{D}_{\gamma,k}$ respectively. The first inequality is an instance of the Cauchy-Schwarz inequality. The last equality uses the independence of $X^S_{s,u}X^{S'}_{s,v}$ and $Y^S_{w,u'}Y^{S'}_{w,v'}$ for any $u,v,u',v'\in V[G]$ as well as the fact that the weight $\mathrm{wt}(S)$ is just the probability of $S$ in $\mathcal{D}_{\gamma,k}$.
\end{proof}

The following lemma characterizes the expectations in \eq{lowerboundprob} in terms of Markov chains.

\mylemma{lemmarkov}{\lemmarkov}

\begin{proof}
We add $\Delta-\deg(u)$ self-loops to each node $u$ and hence
a non-lazy step of a forward walk is the same as a step of the random walk over $\reg(G)$.
Since $S$ has distribution $\mathcal{D}_{\gamma,k}$ where each step is chosen to be lazy with probability $1-\gamma$,
the forward walk with random pattern $S$ starting from $u$ is just a lazy random walk from $u$ with transition matrix $\mathcal{L}_{\gamma}\left(\mat{M}_{\reg(G)}\right)$. This proves the first claim.

For the second claim, note that two forward walks are independent in some round $i$ if at least one is lazy, since a lazy step is deterministic. The corresponding transition matrix is $\mat{I}\otimes\mat{I}$, $\mat{I}\otimes\mat{M}_{\reg(G)}$ or $\mat{M}_{\reg(G)}\otimes\mat{I}$, depending on which walk is lazy. When both walks are non-lazy and are at distinct nodes $u$ and $w$ respectively, they are still independent and behave according to $\mat{M}_{\reg(G)}\otimes \mat{M}_{\reg(G)}$ by the independence of $f(i,u)$ and $f(i,w)$.
%and $r_{i,u}$, $r_{i,w}$.
If $u=w$, then the two walks move the same node according to $\mat{M}_{\reg(G)}$. So the case for two lazy steps is exactly characterized by the Doeblin coupling $\mathcal{Q}\left(\mat{M}_{\reg(G)}\right)$. And when the two walks have independent random patterns $S,S'\sim \mathcal{D}_{\gamma,k}$, the corresponding transition matrix is
$\mathcal{L}_{\gamma,\gamma}\circ\mathcal{Q}\left(\mat{M}_{\reg(G)}\right)$ by definition. The second claim follows.

For the third claim, we consider the probability that a node $u\neq v$ is included in $N_{i,v}$. We divide it into two cases: the case that $N_{i,v}=\{u\}$ (i.e. $\tilde{f}(T-1-2i,v)=\tilde{f}(T-2-2i,v)=u$) and the case that $N_{i,v}=\{u,u'\}$ where $u\neq u'$. The first case occurs with probability $1/\Delta^2$. For the second one, we have $\left(\tilde{f}(T-1-2i,v),\tilde{f}(T-2-2i,v)\right)=(u,u')$ or $(u',u)$ for some $u'\neq u$. And exactly one of them is counted by the condition $\tilde{f}(T-1-2i,v)\preceq \tilde{f}(T-2-2i,v)$.
As they occur with the same probability we may assume it is the first one that is counted.
Summing over $u'\neq u$, we conclude that this case occurs with probability $1/\Delta \cdot (1-1/\Delta)$. So $u$ is included in $N_{i,v}$ with probability $1/\Delta$ for any $u\neq v$. And $N_{i,u}^\vee=\{v\}$ occurs when $u\in N_{i,v}$ and $u\not\in N_{i,v'}$ for all $v'\in N(u)\setminus\{v\}$, whose probability is $1/\Delta\cdot (1-1/\Delta)^{\deg(u)-1}$. Taking the condition $r_{i,u}\leq (1-1/\Delta)^{\Delta-\deg(u)}$ into account, we see that the reversed walk extends from $u$ to each neighbor $v\in N(u)$ with probability $1/\Delta\cdot (1-1/\Delta)^{\Delta-1}=\left(\mathcal{L}_{\gamma'}\left(M_{\reg(G)}\right)\right)_{uv}$.
So a non-lazy step of a reversed walk is the same as a step of the random walk over $\mathcal{L}_{\gamma'}\left(M_{\reg(G)}\right)$. And the reversed walk with a random pattern chosen from $\mathcal{D}_{\gamma,k}$ corresponds to the transition matrix $\mathcal{L}_{\gamma}\circ\mathcal{L}_{\gamma'}\left(\mat{M}_{\reg(G)}
\right)$, similarly to the first claim.

The proof to the last claim is similar to the second one. The only non-trivial part is to show that when two walks are both non-lazy and are at distinct nodes $u$ and $w$ respectively, they behave independently according to $\mathcal{L}_{\gamma'}\left(M_{\reg(G)}\right)\otimes \mathcal{L}_{\gamma'}\left(M_{\reg(G)}\right)$.
Note that $r_{i,u}$ and $r_{i,w}$ are independent.
So it suffices to show the probability that $N_{i,u}^\vee=\{v\}$ and $N_{i,w}^\vee=\{x\}$ both occur equals the product of their individual probabilities for all $v\in N(u)$ and $x\in N(w)$.
For $a\in V[G]$ and $b\in N(a)$, let $\mathcal{I}_{a,b}$ be the event that $a\in N_{i,b}$. The claim follows if the events $\mathcal{I}_{a,b}$ are independent for all $a\in\{u,w\}$ and neighbor $b\in N(a)$. Note that $\mathcal{I}_{a,b}$ depends solely on $\tilde{f}(T-1-2i,b)$ and $\tilde{f}(T-2-2i,b)$. And those for different $b$ are independent. So we reduce to proving $\mathcal{I}_{u,b}$ and $\mathcal{I}_{w,b}$ are independent for fixed $b\in N(u)\cap N(w)$. Each occurs with probability $1/\Delta$, as shown in the proof to the third claim. Both occurs exactly when $\left(\tilde{f}(T-1-2i,b),\tilde{f}(T-2-2i,b)\right)$ equals $(u,w)$ if $u\preceq w$, or $(w,u)$ if $w\preceq u$. So the probability that both events occur equals $1/\Delta^2$, as desired.
\end{proof}

\subsubsection{Proof of \lemref{lemmixing}}

\againlemma{lemmixing}{\lemmixing}

To prove \lemref{lemmixing}, we show that $\mathcal{L}_{\gamma,\gamma}\circ\mathcal{Q}(\mat{M})$ behaves similarly as $\mathcal{L}_\gamma(\mat{M})\otimes \mathcal{L}_\gamma(\mat{M})$, in the sense that it almost preserves the vector $\vec{\pi}\otimes \vec{\pi}$ and shrinks vectors orthogonal to $\vec{\pi}\otimes \vec{\pi}$.
For a distribution $\vec{u}$ over $V[G]\times V[G]$, we
have the decomposition $\vec{u}=\vec{\pi}\otimes \vec{\pi}+\vec{u}^\perp$
where $\vec{u}^\perp\triangleq\vec{u}-\vec{\pi}\otimes \vec{\pi}$ is orthogonal to $\vec{\pi}\otimes \vec{\pi}$.

\begin{lem}\label{lem:parallel}
Let $\mat{M}$, $\vec{\pi}$ and $\gamma$ be as in \lemref{lemmixing}. Then
$\left\|((\vec{\pi}\otimes \vec{\pi})\mathcal{L}_{\gamma,\gamma}\circ\mathcal{Q}(\mat{M}))^\perp\right\|_2\leq \sqrt{2}\gamma^2n^{-3/2}$.
\end{lem}

\begin{proof}
Let $\mat{E}=\mathcal{L}_{\gamma,\gamma}\circ\mathcal{Q}(\mat{M})-\mathcal{L}_\gamma(\mat{M})\otimes \mathcal{L}_\gamma(\mat{M})$.
Note that
$((\vec{\pi}\otimes \vec{\pi})\mathcal{L}_{\gamma,\gamma}\circ\mathcal{Q}(\mat{M}))^\perp
=((\vec{\pi}\otimes \vec{\pi})\mat{E})^\perp
$, since $\mathcal{L}_\gamma(\mat{M})\otimes \mathcal{L}_\gamma(\mat{M})$ fixes $\vec{\pi}\otimes \vec{\pi}$.
Also note that $\|((\vec{\pi}\otimes \vec{\pi})\mat{E})^\perp\|_2\leq\|(\vec{\pi}\otimes \vec{\pi})\mat{E}\|_2$. So it suffices to prove
$\|(\vec{\pi}\otimes \vec{\pi})\mat{E}\|_2\leq \sqrt{2}\gamma^2n^{-3/2}$.
By definition, we have
$$
\mat{E}_{(u,w)(v,x)}=\begin{cases}
0 & u\neq w,\\
\gamma^2\left(\mat{M}_{uv}-(\mat{M}\otimes \mat{M})_{(u,w)(v,x)}\right) & u=w, v=x,\\
-\gamma^2(\mat{M}\otimes \mat{M})_{(u,w)(v,x)} & u=w, v\neq x.
\end{cases}
$$
So $\mat{E}_{(u,w)(v,x)}=0$ for $u\neq w$, and $\left|\mat{E}_{(u,w)(v,x)}\right|\leq\gamma^2\mat{M}_{uv}$ for $u=w$.
Then for any $v,x\in V[G]$, we have
$$
\left|((\vec{\pi}\otimes \vec{\pi})\mat{E})_{(v,x)}\right|\leq
\sum_{u\in V[G]} (1/n^2)\cdot\gamma^2\mat{M}_{uv}
=\gamma^2/n^2.
$$
So $\|(\vec{\pi}\otimes \vec{\pi})\mat{E}\|_\infty\leq \gamma^2/n^2$.
We also have
\begin{align*}
\|(\vec{\pi}\otimes \vec{\pi})\mat{E}\|_1
&=\sum_{v,x\in V[G]}\left|\sum_{u\in V[G]}(1/n^2)\mat{E}_{(u,u)(v,x)}\right|\\
&\leq\sum_{u,v\in V[G]} (1/n^2)
\gamma^2\mat{M}_{uv}
+
\sum_{u,v,x\in V[G]} (1/n^2)
\gamma^2(\mat{M}\otimes \mat{M})_{(u,u)(v,x)}\\
&=2\gamma^2/n.
\end{align*}
By H\"{o}lder's inequality, we have
$$
\|(\vec{\pi}\otimes \vec{\pi})\mat{E}\|_2
\leq\sqrt{\|(\vec{\pi}\otimes \vec{\pi})\mat{E}\|_1
\|(\vec{\pi}\otimes \vec{\pi})\mat{E}\|_\infty}\leq \sqrt{2}\gamma^2n^{-3/2}.
$$
\end{proof}

\begin{lem}\label{lem:perp}
Let $\mat{M}$, $\vec{\pi}$ and $\gamma$ be as in  \lemref{lemmixing}.
For any vector $\vec{u}\in\mathbb{R}^n\otimes\mathbb{R}^n$ orthogonal to $\vec{\pi}\otimes \vec{\pi}$, we have
$\vec{u}\mathcal{L}_{\gamma,\gamma}\circ\mathcal{Q}(\mat{M})\perp\vec{\pi}\otimes \vec{\pi}$ and
$$
\|\vec{u}\mathcal{L}_{\gamma,\gamma}\circ\mathcal{Q}(\mat{M})\|_2\leq \left(
1-(1-\gamma)\gamma\alpha+\gamma^2\sqrt{2\eta}
\right)\|\vec{u}\|_2.
$$
\end{lem}

\begin{proof}
Since $\mathcal{L}_{\gamma,\gamma}\circ\mathcal{Q}(\mat{M})$ is stochastic, we have
$$
\langle\vec{u}\mathcal{L}_{\gamma,\gamma}\circ\mathcal{Q}(\mat{M}),  \vec{\pi}\otimes \vec{\pi}\rangle=\langle\vec{u},  (\vec{\pi}\otimes \vec{\pi})(\mathcal{L}_{\gamma,\gamma}\circ\mathcal{Q}(\mat{M}))^\intercal\rangle=\langle\vec{u},\vec{\pi}\otimes \vec{\pi}\rangle=0.
$$
To prove the second claim,
we write $\mathcal{L}_{\gamma,\gamma}\circ\mathcal{Q}(\mat{M})=\mat{R}_1+\mat{R}_2$
where
$$
\mat{R}_1=(1-\gamma)^2(\mat{I}\otimes\mat{I})+(1-\gamma)\gamma(\mat{I}\otimes\mat{M})
+\gamma(1-\gamma)(\mat{M}\otimes\mat{I})
$$
and $\mat{R}_2=\gamma^2\mathcal{Q}(\mat{M})$. Then we bound $\|\vec{u}\mat{R}_1\|_2$ and $\|\vec{u}\mat{R}_2\|_2$ individually.

Observe that $\mat{R}_1=(1-\gamma^2)\mathcal{L}_{\gamma_0}(\mat{R}_0)$ where $\mat{R}_0$ is the stochastic matrix $(\mat{I}\otimes\mat{M}+\mat{M}\otimes\mat{I})/2$ and $\gamma_0=2\gamma/(1+\gamma)$.
Recall that $\mat{M}$ has $n$ normalized orthogonal eigenvectors $\vec{v}_1,\dots, \vec{v}_n$ in $\mathbb{R}^n$ associated with $n$ real eigenvalues $1=\lambda_1>1-\alpha\geq \lambda_2\geq\dots\geq \lambda_n\geq -1$ respectively, and $\vec{v}_1$ is parallel to $\vec{\pi}$.
Then $\mat{R}_0$ has $n^2$ normalized orthogonal eigenvectors $\vec{v}_{i}\otimes\vec{v}_{j}$ associated with eigenvalues $(\lambda_j+\lambda_i)/2$, $i,j=1,\dots,n$. And $\mathcal{L}_{\gamma_0}(\mat{R}_0)$ has the same set of eigenvectors, with the $(i,j)$th eigenvalue replaced by $(1-\gamma_0)+\gamma_0(\lambda_j+\lambda_i)/2$. These eigenvalues are all non-negative, since $(\lambda_j+\lambda_i)/2\geq -1$ and $\gamma_0=2\gamma/(1+\gamma)\leq 1/2$ (from the condition $\gamma\leq 1/3$). So the absolute spectral gap of $\mathcal{L}_{\gamma_0}(\mat{R}_0)$ is
\begin{align*}
\lefteqn{1-\max_{(i,j)\neq (1,1)} \left((1-\gamma_0)+\gamma_0(\lambda_j+\lambda_i)/2\right)}\\
& = 1-\max_{(i,j)\neq (1,1)} \left( 1+
 \left( \frac{\lambda_i+\lambda_j}{2}-1 \right)\gamma_0\right) \\
&= \min_{ (i,j)\neq (1,1)} \left(\frac{1-\lambda_i+1-\lambda_j}{2} \right)\gamma_0 \geq \gamma_0\alpha/2.
\end{align*}
As $\vec{u}$ is parallel to $\vec{\pi}\otimes\vec{\pi}$, or equivalently $\vec{v}_1\otimes\vec{v}_1$, we have
$$
\|\vec{u}\mat{R}_1\|_2 =
 (1-\gamma^2)\left\|\vec{u}\mathcal{L}_{\gamma_0}(\mat{R}_0)\right\|_2
 \leq (1-\gamma^2)(1-\gamma_0\alpha/2)\|\vec{u}\|_2
 =(1-\gamma^2)(1-\gamma\alpha/(1+\gamma))\|\vec{u}\|_2.
$$

Then we bound $\|\vec{u}\mat{R}_2\|_2=\gamma^2\|\vec{u}\mathcal{Q}(\mat{M})\|_2$.
By permutating the rows (resp. columns) of $\mathcal{Q}(\mat{M})$, we assume its first $n$ rows (resp. $n$ columns) are indexed by the diagonal elements $\{(u,u):u\in V[G]\}$. By definition, we have
$$
\mathcal{Q}(\mat{M})=\begin{pmatrix}
\mat{M} & \mat{0}\\
\mat{A} & \mat{B}
\end{pmatrix}
$$
where $\begin{pmatrix}\mat{A} & \mat{B}\end{pmatrix}$ are the last $n^2-n$ rows of $\mat{M}\otimes \mat{M}$ (we permutate the rows and columns of $\mat{M}\otimes \mat{M}$ in the same way as we did for $\mathcal{Q}(\mat{M})$.
Write $\vec{u}=\begin{pmatrix}\vec{u}_1 & \vec{u}_2\end{pmatrix}$ where $\vec{u}_1\in\mathbb{R}^n$ and $\vec{u}_2\in\mathbb{R}^{n^2-n}$, consisting of entries indexed by $(u,w)$, $u=w$ and $u\neq w$ respectively.
Then
\begin{equation}\label{eqsum}
\begin{aligned}
\|\vec{u}\mathcal{Q}(\mat{M})\|^2_2
&=\left\|\begin{pmatrix}\vec{u}_1\mat{M}+\vec{u}_2\mat{A} &
\vec{u}_2\mat{B}
\end{pmatrix}\right\|^2_2\\
&=\|\vec{u}_1\mat{M}\|_2^2+\|\vec{u}_2\mat{A}\|_2^2+
\|\vec{u}_2\mat{B}\|_2^2
+2\langle\vec{u}_1\mat{M}, \vec{u}_2\mat{A}\rangle\\
&=\|\vec{u}_1\mat{M}\|_2^2+\|\begin{pmatrix}\vec{0} & \vec{u}_2\end{pmatrix}(\mat{M}\otimes\mat{M})\|_2^2
+2\langle\vec{u}_1\mat{M}, \vec{u}_2\mat{A}\rangle\\
&\leq \|\vec{u}_1\|_2^2+\|\vec{u}_2\|_2^2+2\langle\vec{u}_1\mat{M}, \vec{u}_2\mat{A}\rangle\\
&=\|\vec{u}\|_2^2+2\langle\vec{u}_1\mat{M}, \vec{u}_2\mat{A}\rangle\\
&\leq \|\vec{u}\|_2^2+2\|\vec{u}_1\mat{M}\|_2\|\vec{u}_2\mat{A}\|_2\\
&\leq \|\vec{u}\|_2^2+2\|\vec{u}\|_2^2\|\mat{A}\|_2\\
&\leq \|\vec{u}\|_2^2\left(1+2\sqrt{\|\mat{A}\|_1\|\mat{A}\|_\infty}\right)
\end{aligned}
\end{equation}
The third equality uses the fact that
$\begin{pmatrix}\vec{0} & \vec{u}_2\end{pmatrix}(\mat{M}\otimes\mat{M})=\begin{pmatrix}\vec{u}_2\mat{A} & \vec{u}_2\mat{B}\end{pmatrix}$. The first inequality uses the fact that $\|\mat{M}\|_2, \|\mat{M}\otimes\mat{M}\|_2\leq 1$. The second inequality is an instance of the Cauchy-Schwarz inequality. The third one uses the facts that $\|\mat{M}\|_2\leq 1$ and $\|\vec{u}_1\|_2,\|\vec{u}_2\|_2\leq \|\vec{u}\|_2$. And the last one uses the inequality
$\|\mat{A}\|_2\leq \sqrt{\|\mat{A}\|_1\|\mat{A}\|_\infty}$.

We have $\|\mat{A}\|_\infty\leq\|\mat{M}\otimes \mat{M}\|_\infty=1$.
To bound $\|\mat{A}\|_1$, observe that $\|\mat{A}\|_1$ is by definition the maximum of
the $\ell_1$-norm of rows of $\mat{A}$.
%For $u,v\in V[G]$, let $e(u,v)$ be the number of edges between $u$ and $v$. Then $e(u,v)=1$ for $\{u,v\}\in E, u\neq v$, since we assume $G$ has no multi-edge other than self-loops.
Then
\begin{align*}
\|\mat{A}\|_1&=\max_{\substack{u,w\in V[G]\\u\neq w}}\sum_{v\in V[G]} \mat{M}_{uv}\mat{M}_{wv}\\
&\leq \max_{\substack{u,w\in V[G]\\u\neq w}}\left(\eta\mat{M}_{ww}+\eta\sum_{v\in V[G]\setminus\{w\}}\mat{M}_{uv}\right)\leq 2\eta.
\end{align*}
Combining it with \eqref{eqsum}, we obtain
$$
\|\vec{u}\mat{R}_2\|^2_2=
\gamma^4\|\vec{u}\mathcal{Q}(\mat{M})\|^2_2\leq \gamma^4\left(1+2\sqrt{2\eta}\right)\|\vec{u}\|^2_2\leq
\gamma^4\left(1+\sqrt{2\eta}\right)^2\|\vec{u}\|^2_2
.
$$
Therefore
\begin{align*}
\|\vec{u}\mathcal{L}_{\gamma,\gamma}\circ\mathcal{Q}(\mat{M})\|_2
&\leq \|\vec{u}\mat{R}_1\|_2+\|\vec{u}\mat{R}_2\|_2\\
&\leq
(1-\gamma^2)(1-\gamma\alpha/(1+\gamma))\|\vec{u}\|_2
+
\gamma^2\left(1+\sqrt{2\eta}\right)\|\vec{u}\|_2\\
&=\left(
1-(1-\gamma)\gamma\alpha+\gamma^2\sqrt{2\eta}
\right)\|\vec{u}\|_2.
\end{align*}

\end{proof}

\begin{proof}[Proof of \lemref{lemmixing}]
Note that we are bounding the $\ell_2$-norm of $\vec{u}\left(\mathcal{L}_{\gamma,\gamma}\circ\mathcal{Q}(\mat{M})\right)^k-\vec{\pi}\otimes\vec{\pi}=\left(\vec{u}\left(\mathcal{L}_{\gamma,\gamma}\circ\mathcal{Q}(\mat{M})\right)^k\right)^\perp$.
The proof is based on the induction on $k$. When $k=0$, we have
$$
\left\|\left(\vec{u}\left(\mathcal{L}_{\gamma,\gamma}\circ\mathcal{Q}(\mat{M})\right)^k\right)^\perp\right\|_2
\leq \left\|\vec{u}\left(\mathcal{L}_{\gamma,\gamma}\circ\mathcal{Q}(\mat{M})\right)^k\right\|_2\leq 1,
$$
and hence the claim holds.
For $k>0$, assume the claim holds for $k'<k$.
Let $\vec{v}=\vec{u}\left(\mathcal{L}_{\gamma,\gamma}\circ\mathcal{Q}(\mat{M})\right)^{k-1}$.
We have
\begin{align*}
\left(\vec{u}\left(\mathcal{L}_{\gamma,\gamma}\circ\mathcal{Q}(\mat{M})\right)^k\right)^\perp
&=\left(\vec{v}\mathcal{L}_{\gamma,\gamma}\circ\mathcal{Q}(\mat{M})\right)^\perp\\
&=\left((\vec{\pi}\otimes\vec{\pi})\mathcal{L}_{\gamma,\gamma}\circ\mathcal{Q}(\mat{M})\right)^\perp
+\left(\vec{v}^\perp\mathcal{L}_{\gamma,\gamma}\circ\mathcal{Q}(\mat{M})\right)^\perp.
\end{align*}

By \lemref{parallel},
we have
$$
\|\left((\vec{\pi}\otimes\vec{\pi})\mathcal{L}_{\gamma,\gamma}\circ\mathcal{Q}(\mat{M})\right)^\perp\|_2\leq \sqrt{2}\gamma^2n^{-3/2}.
$$
And by \lemref{perp}, we have $\left(\vec{v}^\perp\mathcal{L}_{\gamma,\gamma}\circ\mathcal{Q}(\mat{M})\right)^\perp
=\vec{v}^\perp\mathcal{L}_{\gamma,\gamma}\circ\mathcal{Q}(\mat{M})$ whose $\ell_2$-norm is at most
$$
\left(
1-(1-\gamma)\gamma\alpha+\gamma^2\sqrt{2\eta}
\right)\|\vec{v}\|_2
\leq
(1-\gamma\alpha/2)\|\vec{v}\|_2
$$
where we use the condition $\gamma\leq \left\{1/3,\alpha \eta^{-1/2}/9\right\}$.
This is bounded by
$$
(1-\gamma\alpha/2)\left((1-\gamma\alpha/2)^{k-1} + 2\sqrt{2}\gamma\alpha^{-1} n^{-3/2}\right)
=(1-\gamma\alpha/2)^k + 2\sqrt{2}\gamma\alpha^{-1} n^{-3/2}(1-\gamma\alpha/2)
$$
by the induction hypothesis.
Then
\begin{align*}
\left\|\left(\vec{u}\left(\mathcal{L}_{\gamma,\gamma}\circ\mathcal{Q}(\mat{M})\right)^k\right)^\perp\right\|_2
&\leq
\left\|\left((\vec{\pi}\otimes\vec{\pi})\mathcal{L}_{\gamma,\gamma}\circ\mathcal{Q}(\mat{M})\right)^\perp\right\|_2
+\left\|\left(\vec{v}^\perp\mathcal{L}_{\gamma,\gamma}\circ\mathcal{Q}(\mat{M})\right)^\perp\right\|_2\\
&\leq
\sqrt{2}\gamma^2n^{-3/2}
+(1-\gamma\alpha/2)^k + 2\sqrt{2}\gamma\alpha^{-1} n^{-3/2}(1-\gamma\alpha/2)\\
&=(1-\gamma\alpha/2)^k + 2\sqrt{2}\gamma\alpha^{-1} n^{-3/2}
\end{align*}
as desired.
\end{proof}

As a side product, we show that the chain $\mathcal{L}_{\gamma,\gamma}\circ\mathcal{Q}(\mat{M})$ behaves similarly as $\mathcal{L}_\gamma(\mat{M})\otimes \mathcal{L}_\gamma(\mat{M})$ in terms of the stationary distribution and the mixing time.

\begin{cor}\label{cor_approx}
Let $\mat{M}$, $\gamma$ and $\alpha$ be as in Lemma \ref{lem:lemmixing}.
Let $\vec{\pi}'$ be the stationary distribution\footnote{The lazyness and $\alpha>0$
guarantees that $\mathcal{L}_{\gamma,\gamma}\circ\mathcal{Q}(\mat{M})$ is ergodic and has a unique stationary distribution.} of $\mathcal{L}_{\gamma,\gamma}\circ\mathcal{Q}(\mat{M})$. Then
$$
\left\|\vec{\pi}'-\vec{\pi}\otimes \vec{\pi}\right\|_2
\leq (1-\gamma\alpha/2)^k + 2\sqrt{2}\gamma\alpha^{-1} n^{-3/2}.
$$
Define the $\ell_1$-mixing time $\bar{\tau}(\epsilon):=\max_{\vec{u}}\min\{k: \|\vec{u}\left(\mathcal{L}_{\gamma,\gamma}\circ\mathcal{Q}(\mat{M})\right)^k-\vec{\pi}\otimes\vec{\pi}\|_1\leq\epsilon\}$ where $\vec{u}$ ranges over all distributions over $V[G]\times V[G]$. Assuming $\gamma\alpha^{-1}=O(n^{1/2-c})$ for some constant $c>0$, we have
$\bar{\tau}(\epsilon) = O(\gamma^{-1}\alpha^{-1}(\log n+\log\epsilon^{-1}))$.
\end{cor}

\begin{proof}
The first claim follows directly from Lemma \ref{lem:lemmixing}.
We also have
$$
\left\|\vec{u}\left(\mathcal{L}_{\gamma,\gamma}\circ\mathcal{Q}(\mat{M})\right)^k-\vec{\pi}\otimes \vec{\pi}\right\|_1
\leq n\left\|\vec{u}\left(\mathcal{L}_{\gamma,\gamma}\circ\mathcal{Q}(\mat{M})\right)^k-\vec{\pi}\otimes \vec{\pi}\right\|_2
\leq n^{-c}$$
for sufficiently large $k=O(\gamma^{-1}\alpha^{-1}\log n)$, again by Lemma \ref{lem:lemmixing}. So $\bar{\tau}(n^{-c})=O(\gamma^{-1}\alpha^{-1}\log n)$. The second claim then follows from the well-known fact that $\bar{\tau}(\epsilon)\leq\bar{\tau}(\delta)\lceil\log_{\delta}\epsilon\rceil$ for $\epsilon,\delta>0$.
\end{proof}

We know that the stationary distribution of $\mathcal{Q}(\mat{M})$
is the uniform distribution over the set of diagonal entries $\{(u,u): u\in V[G]\}$.
So is the stationary distribution of the lazy chain $\mathcal{L}_{\gamma}\circ\mathcal{Q}(\mat{M})$
for any $\gamma\in (0,1]$.
Interestingly, Corollary \ref{cor_approx} tells us that the ``bi-lazy'' chain $\mathcal{L}_{\gamma,\gamma}\circ\mathcal{Q}(\mat{M})$ behaves very differently, as its stationary distribution is close to $\vec{\pi}\otimes \vec{\pi}$ instead.

\subsubsection{Proof of \thmref{FullRandomness}}

We are now ready to derive a bound on the runtime of Protocol~\ref{pro_prototype}.

\mylemma{thmrandom}{\thmrandom}

\begin{proof}
Let $s\in V[G]$ be the initial node and fix a target node $w\in V[G]$.
Let $c>0$ be any constant.
Choose $\gamma=\min\left\{1/3, \Delta^{0.5-c}\alpha/9\right\}
\leq n^{0.5-c}\alpha/9$.
Choose $k=(\gamma\gamma'\alpha)^{-1}\beta^2\log n + 1$ and let $T=4k$.
So $T=O(C\log n)$.
Define the distributions
$\vec{u}=\vec{e}_s\mat{M}_1^k$,
$\vec{v}=\vec{e}_{(s,s)}\mat{M}_2^k$,
$\vec{u}'=\vec{e}_w\mat{M}_3^k$,
and
$\vec{v}'=\vec{e}_{(w,w)}\mat{M}_4^k$, where $\mat{M}_1,\dots,\mat{M}_4$ are as in  \lemref{lemmarkov}.
Let $\vec{\pi}$ be the uniform distribution over $V[G]$.
As before, let $\vec{u}^\perp=\vec{u}-\vec{\pi}$ and
$\vec{v}^\perp=\vec{v}-\vec{\pi}\otimes\vec{\pi}$, and similarly for $\vec{u}'$ and $\vec{v}'$.
By \lemref{lembound} and \lemref{lemmarkov}, the probability that $w$ gets the rumor in $k$ rounds is lower bounded by
\begin{equation}\label{eqprob}
\begin{aligned}
&\frac{\sum_{u,v\in V[G]}
\langle \vec{u}, \vec{e}_u\rangle
\langle \vec{u}, \vec{e}_v\rangle
\langle \vec{u}', \vec{e}_u\rangle
\langle \vec{u}', \vec{e}_v\rangle}
{\sum_{u,v\in V[G]}
\left\langle \vec{v}, \vec{e}_{(u,v)}\right\rangle
\left\langle \vec{v}', \vec{e}_{(u,v)}\right\rangle}
=
\frac{\langle\vec{u},\vec{u}'\rangle^2}{\langle\vec{v},\vec{v}'\rangle}\\
&=
\frac{\left(\langle\vec{\pi},\vec{\pi}\rangle
+\langle\vec{u}^\perp,\vec{\pi}\rangle
+\langle\vec{\pi},\vec{u}'^\perp\rangle
+\langle\vec{u}^\perp,\vec{u}'^\perp\rangle
\right)^2}
{\langle\vec{\pi}\otimes\vec{\pi},\vec{\pi}\otimes\vec{\pi}\rangle
+\langle\vec{v}^\perp,\vec{\pi}\otimes\vec{\pi}\rangle
+\langle\vec{\pi}\otimes\vec{\pi},\vec{v}'^\perp\rangle
+\langle\vec{v}^\perp,\vec{v}'^\perp\rangle
}\\
&=
\frac{\left(
1/n+\langle\vec{u}^\perp,\vec{u}'^\perp\rangle
\right)^2}
{1/n^2+\langle\vec{v}^\perp,\vec{v}'^\perp\rangle}.
\end{aligned}
\end{equation}
Note that $\mat{M}_1=\mathcal{L}_{\gamma}\left(\mat{M}_{\reg(G)}\right)$ and
$\mat{M}_3=\left(\mathcal{L}_{\gamma}\circ\mathcal{L}_{\gamma'}\left(\mat{M}_{\reg(G)}
\right)\right)$ have absolute spectral gaps $\gamma\alpha\beta^{-2}$ and $\gamma\gamma'\alpha\beta^{-2}$ respectively. This follows from \lemref{lemregularization} and the definition of lazy Markov chains (Also, the lazyness guarantees that the eigenvalues are all non-negative, and hence the bounds are about absolute spectral gaps, not just spectral gaps).
By Lemma \ref{lemmixing0} and the fact that
$k\geq  (\gamma\gamma'\alpha)^{-1}\beta^2\log n + 1
\geq \log_{1-\gamma\gamma'\alpha\beta^{-2}}(1/n) + 1$, we have
$|\langle\vec{u}^\perp,\vec{u}'^\perp\rangle|\leq \left\|\vec{u}^\perp\right\|_2
\left\|\vec{u}'^\perp\right\|_2\leq 1/n^2$.
By \lemref{lemmixing} (with $\eta=1/\Delta$), we have
$$
|\langle\vec{v}^\perp,\vec{v}'^\perp\rangle|\leq \left\|\vec{v}^\perp\right\|_2
\left\|\vec{v}'^\perp\right\|_2\leq \left((1-\gamma\alpha/2)^k + 2\sqrt{2}\gamma\alpha^{-1} n^{-3/2}\right)^2
\leq 1/n^{2+2c}.
$$
So \eqref{eqprob} is lower bounded by $\frac{(1/n-1/n^2)^2}{1/n^2-1/n^{2+2c}}=1-O(n^{-2c})$.
\end{proof}

\thmref{FullRandomness} is obtained by repeating the protocol $O(1)$ times and apply the union bound.

\subsection{Analysis of Protocol~\ref{pro0}}

Let $\mathcal{P}$ be the distribution over the set of functions $f:[T]\times V[G]\to [\Delta]$ associated with Protocol~\ref{pro0}.  The values $f(i,u)$ in the $i$th round are generated using the \textsf{PRG} $\gen$, and the seeds of $\gen$ in different rounds are generated by the \textsf{PRG} $\gen'$.
In this section we show that Protocol \ref{pro_prototype} with distribution $\mathcal{D}=\mathcal{P}$ has almost the same performance as the one with $\mathcal{D}=\mathcal{U}$.
As an intermediate step, we consider the distribution $\mathcal{P}'$ defined as follows:
the values of $f$ in each round are determined by the \textsf{PRG} $\gen$ in the same way as for $\mathcal{P}$ but the seeds of $\gen$ in different rounds are now independent and random, instead of being generated by $\gen'$.
With $\mathcal{D}=\mathcal{P}'$, Definition \ref{defi_walk} are still valid, and \lemref{lembound} still holds by exactly the same proof.
Moreover,  \lemref{lemmarkov} ``almost holds'' in the following sense.

\begin{lem}\label{lem_markov2}
Let $r$, $S$ and $S'$ be independent with distributions $\tilde{\mathcal{D}}$ (induced by $\mathcal{D}=\mathcal{P}'$),  $\mathcal{D}_{\gamma,k}$ and $\mathcal{D}_{\gamma,k}$ respectively.
Then there exist stochastic matrices $\mat{M}'_1,\mat{M}'_3\in\mathbb{R}^{n\times n}$,
$\mat{M}'_2,\mat{M}'_4\in\mathbb{R}^{n\times n}\otimes\mathbb{R}^{n\times n}$
such that $\|\mat{M}'_i-\mat{M}_i\|_1\leq 12\gamma\Delta^2(\epsilon + 2\Delta^3/m)$ for $1\leq i\leq 4$, where $\mat{M}_i$ are as in \lemref{lemmarkov} and $\epsilon,m$ are as in Protocol \ref{pro0}. Moreover,
for any $u,v,w,x\in V[G]$, the following statements hold:
\begin{enumerate}
\item $\EXX{r,S}{X^S_{u,v}}
=\left\langle\vec{e}_u\mat{M}_1'^k,\vec{e}_v\right\rangle$,
\item $\EXX{r,S,S'}{X^S_{u,v}X^{S'}_{w,x}}=\left\langle\vec{e}_{(u,w)}\mat{M}_2'^k, \vec{e}_{(v,x)}\right\rangle$,
\item $\EXX{r,S}{Y^S_{u,v}}
=\left\langle\vec{e}_u\mat{M}_3'^k,\vec{e}_v\right\rangle$,
\item $\EXX{r,S,S'}{Y^S_{u,v}Y^{S'}_{w,x}}=\left\langle\vec{e}_{(u,w)}\mat{M}_4'^k, \vec{e}_{(v,x)}\right\rangle$.
\end{enumerate}
\end{lem}

\begin{proof}
Let $\mat{M}'_1$ (resp. $\mat{M}'_3$) be the transition matrix of a forward (reversed) random walk with random pattern $S\sim\mathcal{D}_{\gamma,k}$.
Let $\mat{M}'_2$ (resp. $\mat{M}'_4$) be the joint transition matrix of two forward (reversed) random walks with random patterns $S,S'\sim\mathcal{D}_{\gamma,k}$. This is exactly the same setting as in  \lemref{lemmarkov}, except that now $\mathcal{D}=\mathcal{P}'$. Since the randomness $f(i,u)$ and $r_{i,u}$ in different rounds are independent, Items 1 -- 4 clearly hold. It remains to show that $\|\mat{M}'_i-\mat{M}_i\|_1\leq 12\gamma\Delta^2(\epsilon + 2\Delta^3/m)$ for $1\leq i\leq 4$.

Recall that the $\ell_1$-norm of a matrix equals the maximal sum of absolute values of entries in a row. So we may fix the row index $u$ (or $(u,v)$) maximizing the sum.
Also fix the auxiliary randomness $\{r_{i,u}\}$ and since if we have a bound for all fixed $\{r_{i,u}\}$, the same bound applies when they are random.

%Note that $\mat{M}'_i-\mat{M}_i$ has at most $(\Delta+1)^2$ nonzero entries in a row: A random walk at a node $u$ only moves to $N(u)\cup\{u\}$, and two walks at $u$ and $v$ respectively only move to a pair of nodes in $(N(u)\cup\{u\})\times (N(v)\cup\{v\})$.

Consider the $i$th step of a forward walk with random pattern $S\sim\mathcal{D}_{\gamma,k}$ from node $u$. The walk stays at $u$ if that step is lazy for $\mathcal{D}=\mathcal{P}'$ and also for $\mathcal{D}=\mathcal{U}$. So we may assume the step is non-lazy which occurs with probability $\gamma$.
The event that the walk moves to $v$ is determined solely by $f(i,u)$ and hence characterized by a combinatorial rectangle of dimension one.
By Lemma \ref{lem:rectanglemodulo}, we have
$\left|\left(\mat{M}'_1-\mat{M}_1\right)_{uv}\right|\leq \gamma(\epsilon + \Delta/m)$ (note that the difference is counted only when the step is non-lazy).
Note that the walk always moves to a node in $N(u)\cup\{u\}$. Taking the sum of differences, we have $\left\|\mat{M}'_1-\mat{M}_1\right\|_1\leq \gamma(\Delta+1)(\epsilon + \Delta/m)$.

Now consider the $i$th step of two forward walks from $u$ and $w$ respectively.
We may assume at least one of them has a non-lazy step which occurs with probability $2(1-\gamma)\gamma+\gamma^2\leq 2\gamma$.
The event that the first walk moves to some node $v$ is determined by $f(i,u)$
whereas the event that the second walk moves to some $x$ is determined by $f(i,w)$.
Each is characterized by a combinatorial rectangle in $\prod_{a\in\{u,w\}}[\Delta]$ of dimension one (if $u=w$) or two (if $u\neq w$).
The conjunction of these two events is characterized by the intersection of the two combinatorial rectangles, which is again a combinatorial rectangle in $\prod_{a\in\{u,w\}}[\Delta]$.
By Lemma \ref{lem:rectanglemodulo}, we have
$\left|\left(\mat{M}'_2-\mat{M}_2\right)_{(u,w)(v,x)}\right|\leq 2\gamma(\epsilon + 2\Delta/m)$.
Also the only possible $(v,x)$ are in $(N(u)\cup\{u\})\times(N(w)\cup\{w\})$. Taking the sum of differences, we have $\left\|\mat{M}'_2-\mat{M}_2\right\|_1\leq 2\gamma(\Delta+1)^2(\epsilon + 2\Delta/m)$.

Now consider the $i$th step of a reversed walk with random pattern $S\sim\mathcal{D}_{\gamma,k}$ from a node $u$. Again assume the step is non-lazy which occurs with probability $\gamma$.
Let $i_0=T-2i-1$ and $i_1=T-2i-2$.
The event $u\in N_{i,v}$ for $v\in N(u)$ is determined by whether $(f(i_0,v),f(i_1,v))\in S_{v,u}$ for some $S_{v,u}\subseteq [\Delta^2]$. Then the event whether $N^\vee_{i,u}=\{v\}$ for $v\in N(u)$ is characterized by the combinatorial rectangle $\prod_{w\in N(u)} S_w\subseteq \prod_{w\in N(u)}[\Delta^2]$ of dimension $\deg(u)\leq\Delta$ where $S_w$ equals $S_{w,u}$ if $w=v$, and equals $[\Delta^2]\setminus S_{w,u}$ if $w\neq v$.
By Lemma \ref{lem:rectanglemodulo}, we have
$\left|\left(\mat{M}'_3-\mat{M}_3\right)_{uv}\right|\leq \gamma(\epsilon + \Delta^3/m)$ for $v\in N(u)$.
When $v=u$, we have $\left|\left(\mat{M}'_3-\mat{M}_3\right)_{uv}\right|\leq
\sum_{w\in N(u)}\left|\left(\mat{M}'_3-\mat{M}_3\right)_{uw}\right|
\leq \gamma\Delta(\epsilon + \Delta^3/m)$ since $\sum_{w\in N(u)\cup\{u\}}\left(\mat{M}'_3-\mat{M}_3\right)_{uw}
=\sum_{w\in N(u)\cup\{u\}}(\mat{M}'_3)_{uw}-\sum_{w\in N(u)\cup\{u\}}(\mat{M}_3)_{uw}
=1-1=0$. Taking the sum of differences, we have $\left\|\mat{M}'_3-\mat{M}_3\right\|_1\leq 2\gamma\Delta(\epsilon + \Delta^3/m)$.

Finally consider the $i$th step of two reversed walks from $u$ and $w$ respectively.
We may assume at least one of them has a non-lazy step which occurs with probability $2(1-\gamma)\gamma+\gamma^2\leq 2\gamma$.
Similar to the case of two forward walks, using the fact that the family of combinatorial rectangles is closed under intersection, we know the event that the two walks move to some nodes $v\in N(u)$ and $x\in N(u)$ respectively is characterized by a combinatorial rectangle in $\prod_{a\in N(u)\cup N(w)} [\Delta^2]$ of dimension at most $2\Delta$.
By Lemma \ref{lem:rectanglemodulo}, we have
$\left|\left(\mat{M}'_4-\mat{M}_4\right)_{(u,w)(v,x)}\right|\leq 2\gamma(\epsilon + 2\Delta^3/m)$ for $v\in N(u)$ and $x\in N(w)$.
When $u\neq v$ and $w=x$, using the fact that $\mat{M}_4$ (resp. $\mat{M}'_4$) is a coupling of two copies of $\mat{M}_3$ (resp. $\mat{M}'_3$),
we have $\sum_{x'\in N(w)}\left(\mat{M}'_4-\mat{M}_4\right)_{(u,w)(v,x')}=\left(\mat{M}'_3-\mat{M}_3\right)_{uv}$ and hence
\begin{equation}\label{eq_comb}
\begin{aligned}
\left|\left(\mat{M}'_4-\mat{M}_4\right)_{(u,w)(v,x)}\right|
&\leq
\left|\left(\mat{M}'_3-\mat{M}_3\right)_{uv}\right|+\sum_{x'\in N(w)\setminus\{w\}}\left|\left(\mat{M}'_4-\mat{M}_4\right)_{(u,w)(v,x')}\right|\\
&\leq \gamma(\epsilon + \Delta^3/m)+2\gamma\Delta(\epsilon + 2\Delta^3/m).
\end{aligned}
\end{equation}
The case that $u=v$ and $w\neq x$ is symmetric.
When $u=v$ and $w=x$, the first inequality of \eqref{eq_comb} still holds, yet the RHS of the second one becomes  $\gamma\Delta(\epsilon + \Delta^3/m) + \Delta(\gamma(\epsilon + \Delta^3/m)+2\gamma\Delta(\epsilon + 2\Delta^3/m))$.
Taking the sum of differences, we have $\left\|\mat{M}'_4-\mat{M}_4\right\|_1\leq 12\gamma\Delta^2(\epsilon + 2\Delta^3/m)$.
\end{proof}

Next we consider the case $\mathcal{D}=\mathcal{P}$. Again Definition \ref{defi_walk} is still valid and \lemref{lembound} still holds by the same proof. Furthermore we show that the expectations are almost the same as in $\mathcal{D}=\mathcal{P}'$ since they can be computed by small-width branching programs:

\begin{lem}\label{lem_bp}
For any $u,w\in V[G]$, the quantities
\begin{equation}\label{firstmoment}
\sum_{v\in V[G]}\left|\EXX{r\sim\tilde{\mathcal{P}}',S}{X^S_{u,v}}
-\EXX{r\sim\tilde{\mathcal{P}},S}{X^S_{u,v}}\right|
\end{equation}
and
\begin{equation}\label{secondmoment}
\sum_{v,x\in V[G]}\left|\EXX{r\sim\tilde{\mathcal{P}}',S,S'}{X^S_{u,v}X^{S'}_{w,x}}
-\EXX{r\sim\tilde{\mathcal{P}},S,S'}{X^S_{u,v}X^{S'}_{w,x}}\right|
\end{equation}
are bounded by $\epsilon'$,
where $\tilde{\mathcal{P}}$ (resp. $\tilde{\mathcal{P}}'$) is the distribution of $r$ induced by $\mathcal{P}$ (resp. $\mathcal{P}'$), $S,S'$ in the subscripts are independent and have distribution $\mathcal{D}_{\gamma,k}$, and $\epsilon'$ is as in Protocol \ref{pro0}.
The same statement holds with $X^S_{u,v}$ and $X^S_{w,x}$ replaced by $Y^S_{u,v}$
and $Y^S_{w,x}$ respectively.
\end{lem}

\begin{proof}
It suffices to bound the quantities with $S$, $S'$ and the auxiliary randomness $\{r_{i,u}\}$ fixed. Then \eqref{firstmoment} becomes
$
\sum_{v\in V[G]}\left|\EXX{f\sim\mathcal{P}'}{X^S_{u,v}}
-\EXX{f\sim\mathcal{P}}{X^S_{u,v}}\right|
$.
Note that for both cases $f\sim\mathcal{P}$ and $f\sim\mathcal{P}'$ we can view $f$ as a random variable determined by a sequence of seeds $y=(y_0,\dots,y_{k-1})\in\left(\{0,1\}^{\ell}\right)^k$. In the former case $y$ is truly random whereas in the latter case it is generated by the \textsf{PRG} $\gen'$. So we may rewrite \eqref{firstmoment} as
$$
\sum_{v\in V[G]}\left|\EXX{y\in\{0,1\}^{\ell'}}{X^S_{u,v}(\gen'(y))}
-\EXX{y\in\left(\{0,1\}^{\ell}\right)^k}{X^S_{u,v}(y)}\right|,
$$
where $X^S_{u,v}(y)$ denotes the value of $X^S_{u,v}$ determined by the sequence of seeds $y$. We claim that $X^S_{u,v}(y)$ is computed by a $(k,n,2^\ell)$-branching program $\mathcal{B}$. More specifically,
it holds that $X^S_{u,v}(y)=1$ iff $\mathcal{B}(u,y)=v$. The branching program $\mathcal{B}$ is easy to construct: we use the set of nodes $[n]=V[G]$ in the $i$th level to keep track of the where the  random walk is at the $i$th step. This location together with the seed $y_i$ (which is used as the label of the outgoing edge in $\mathcal{B}$) uniquely determines the next node. Then the fact that $y$ is generated by an $\epsilon'$-\textsf{PRG} for $(T/2,n^2,2^\ell)$-branching program $\mathcal{B}$ easily implies the bound.
%\footnote{Note that an $\epsilon$-\textsf{PRG} for $(L,W,D)$-branching program is also an $\epsilon$-\textsf{PRG} for $(L',W',D)$-branching program for $L'\leq L$ and $W'\leq W$, by using dummy nodes and levels.}
The bound for \eqref{secondmoment} is derived in the same way, except that we use a $(k,n^2,2^\ell)$-branching program to keep track of two random walks simultaneously.
The cases for $Y^S_{u,v}$
and $Y^S_{u,v}Y^S_{w,x}$ are the same, except that the time is reversed.
\end{proof}

Now we are ready to prove a derandomized version of
\lemref{thmrandom}.
%This is done by replacing the distributions $\vec{u},\vec{u}',\vec{v},\vec{v}'$
%defined in the proof of Theorem \ref{thm_random} with
%very close distributions $\tilde{\vec{u}}, \tilde{\vec{v}}, \tilde{\vec{u}}', \tilde{\vec{v}}'$.
%The closeness is derived from Lemma \ref{lem_markov2} and Lemma \ref{lem_bp} and is measured using the $\ell_1$-norm.
%This may seem insufficient since the proof of Theorem \ref{thm_random} uses the $\ell_2$-norm which may be too large (relatively) even when the $\ell_1$-norm is small.
%We solve this problem using the trick of removing a small fraction of  entries where the distributions mostly differ.

\begin{thm}\label{thm_main}
Suppose $G$ has spectral gap $\alpha$ and irregularity $\beta$. Using Protocol \ref{pro_prototype} with distribution $\mathcal{D}=\mathcal{P}$, any node gets the rumor in $T=O(C \log n)$ rounds with probability at least $1-n^{-2c}$, where $C=(1/\alpha)\cdot\beta^2\max\{1, 1/(\alpha\cdot\Delta^{0.5-c})\}$ and $c>0$ is an arbitrary small constant.
\end{thm}

\begin{proof}
Let $s\in V[G]$ be the initial node and fix a target node $w\in V[G]$.
Let $c,\gamma,k,T,\vec{\pi},\vec{u},\vec{u}',\vec{v},\vec{v}'$ be as in the proof of \lemref{thmrandom} and $T=O(C\log n)$.
Define
$\bar{\vec{u}}=\vec{e}_s\mat{M}_1'^k$,
$\bar{\vec{v}}=\vec{e}_{(s,s)}\mat{M}_2'^k$,
$\bar{\vec{u}}'=\vec{e}_w\mat{M}_3'^k$,
and
$\bar{\vec{v}}'=\vec{e}_{(w,w)}\mat{M}_4'^k$, where $\mat{M}_1',\dots,\mat{M}_4'$ are as in Lemma \ref{lem_markov2}.
Then
$$
\|\bar{\vec{u}}-\vec{u}\|_1
=\left\|\vec{e}_s\left(\mat{M}_1'^k-\mat{M}_1^k\right)\right\|_1
\leq \left\|\mat{M}_1'^k-\mat{M}_1^k\right\|_1
\leq k\left\|\mat{M}_1'-\mat{M}_1\right\|_1\leq k\epsilon_0
$$
where $\epsilon_0 = 12\gamma\Delta^2(\epsilon + 2\Delta^3/m)$ (c.f. Lemma \ref{lem_markov2}). Here the second inequality holds by a simple induction on $k$.
Similarly $\|\bar{\vec{u}}'-\vec{u}'\|_1, \|\bar{\vec{v}}-\vec{v}\|_1, \|\bar{\vec{v}}'-\vec{v}'\|_1\leq k\epsilon_0$.
Define $\tilde{\vec{u}}, \tilde{\vec{u}}'\in\mathbb{R}^n$ and $\tilde{\vec{v}}, \tilde{\vec{v}}'\in\mathbb{R}^n\otimes\mathbb{R}^n$ such that
$\tilde{\vec{u}}_u=\EXX{r,S}{X^S_{s,u}}$,
$\tilde{\vec{u}}'_u=\EXX{r,S}{Y^S_{w,u}}$,
$\tilde{\vec{v}}_{u,v}=\EXX{r,S,S'}{X^S_{s,u}X^S_{s,v}}$ and
$\tilde{\vec{v}}'_{u,v}=\EXX{r,S,S'}{Y^S_{w,u}Y^S_{w,v}}$
where $r$, $S$ and $S'$ are independent with distributions $\tilde{\mathcal{P}}$ (induced by $\mathcal{P}$),  $\mathcal{D}_{\gamma,k}$ and $\mathcal{D}_{\gamma,k}$ respectively.
Then Lemma \ref{lem_markov2} and Lemma \ref{lem_bp} altogether imply that
$\|\tilde{\vec{u}}-\bar{\vec{u}}\|_1\leq\epsilon'$
and hence $\|\tilde{\vec{u}}-\vec{u}\|_1\leq k\epsilon_0+\epsilon'$.
Obviously we have $\|\tilde{\vec{u}}-\vec{u}\|_\infty\leq 1$.
Therefore by H\"{o}lder's inequality, we have
$\|\tilde{\vec{u}}-\vec{u}\|_2\leq \sqrt{k\epsilon_0+\epsilon'}$.
Similarly,
$$
\|\tilde{\vec{u}}'-\vec{u}'\|_2\leq \sqrt{k\epsilon_0+\epsilon'}, \|\tilde{\vec{v}}-\vec{v}\|_2\leq \sqrt{k\epsilon_0+\epsilon'}, \|\tilde{\vec{v}}'-\vec{v}'\|_2\leq \sqrt{k\epsilon_0+\epsilon'}.
$$
As shown in the proof of \lemref{thmrandom},
we have $\left\|\vec{u}^\perp\right\|_2, \left\|\vec{u}'^\perp\right\|_2\leq n^{-1}$, and
$\left\|\vec{v}^\perp\right\|_2, \left\|\vec{v}'^\perp\right\|_2\leq n^{-(1+c)}$.
Note that
$$
\tilde{\vec{u}}^\perp=\tilde{\vec{u}}-\vec{\pi}
=(\tilde{\vec{u}}-\vec{u})+(\vec{u}-\vec{\pi})
=(\tilde{\vec{u}}-\vec{u})+\vec{u}^\perp.
$$
So we have $\left\|\tilde{\vec{u}}^\perp\right\|_2\leq \sqrt{k\epsilon_0+\epsilon'} +  n^{-1}$ and similarly $\left\|\tilde{\vec{u}}'^\perp\right\|_2\leq \sqrt{k\epsilon_0+\epsilon'} +  n^{-1}$, and $\left\|\tilde{\vec{v}}^\perp\right\|_2, \left\|\tilde{\vec{v}}'^\perp\right\|_2\leq \sqrt{k\epsilon_0+\epsilon'} +  n^{-(1+c)}$.

By \lemref{lembound}, the probability that $t$ gets the rumor in $k$ rounds is lower bounded by
\begin{equation}\label{eq_bound}
\begin{aligned}
&\frac{\sum_{u,v\in V[G]}
\langle \tilde{\vec{u}}, \vec{e}_u\rangle
\langle \tilde{\vec{u}}, \vec{e}_v\rangle
\langle \tilde{\vec{u}}', \vec{e}_u\rangle
\langle \tilde{\vec{u}}', \vec{e}_v\rangle}
{\sum_{u,v\in V[G]}
\left\langle \tilde{\vec{v}}, \vec{e}_{(u,v)}\right\rangle
\left\langle \tilde{\vec{v}}', \vec{e}_{(u,v)}\right\rangle}
=
\frac{\left\langle\tilde{\vec{u}},\tilde{\vec{u}}'\right\rangle^2}
{\left\langle\tilde{\vec{v}},\tilde{\vec{v}}'\right\rangle}\\
&=
\frac{\left(
\left\langle\vec{\pi},\vec{\pi}\right\rangle
+\left\langle\tilde{\vec{u}}^\perp,\vec{\pi}\right\rangle
+\left\langle\vec{\pi},\tilde{\vec{u}}'^\perp\right\rangle
+\left\langle\tilde{\vec{u}}^\perp,\tilde{\vec{u}}'^\perp\right\rangle
\right)^2}
{\left\langle\vec{\pi}\otimes\vec{\pi},\vec{\pi}\otimes\vec{\pi}\right\rangle
+\left\langle\tilde{\vec{v}}^\perp,\vec{\pi}\otimes\vec{\pi}\right\rangle
+\left\langle\vec{\pi}\otimes\vec{\pi},\tilde{\vec{v}}'^\perp\right\rangle
+\left\langle\tilde{\vec{v}}^\perp,\tilde{\vec{v}}'^\perp\right\rangle
}\\
&=
\frac{\left(
\left\langle\vec{\pi},\vec{\pi}\right\rangle
+\left\langle\tilde{\vec{u}}^\perp,\tilde{\vec{u}}'^\perp\right\rangle
\right)^2}
{\left\langle\vec{\pi}\otimes\vec{\pi},\vec{\pi}\otimes\vec{\pi}\right\rangle
+\left\langle\tilde{\vec{v}}^\perp,\tilde{\vec{v}}'^\perp\right\rangle
}\\
&=
\frac{\left(
1/n+\left\langle\tilde{\vec{u}}^\perp,\tilde{\vec{u}}'^\perp\right\rangle
\right)^2}
{1/n^2+\left\langle\tilde{\vec{v}}^\perp,\tilde{\vec{v}}'^\perp\right\rangle.
}
\end{aligned}
\end{equation}
We have
\begin{align*}
\left|\left\langle\tilde{\vec{u}}^\perp,\tilde{\vec{u}}'^\perp\right\rangle\right|
\leq \left\|\tilde{\vec{u}}^\perp\right\|_2 \left\|\tilde{\vec{u}}'^\perp\right\|_2
= O\left(k\epsilon_0+\epsilon'+n^{-2}\right),\\
\left|\left\langle\tilde{\vec{v}}^\perp,\tilde{\vec{v}}'^\perp\right\rangle\right|
\leq \left\|\tilde{\vec{v}}^\perp\right\|_2\left\|\tilde{\vec{v}}'^\perp\right\|_2
= O\left(k\epsilon_0+\epsilon'+n^{-(2+2c)}\right).
\end{align*}
So \eqref{eq_bound} is lower bounded by $1-O(n^2(k\epsilon_0+\epsilon')+n^{-2c})$ where $\epsilon_0 = 12\gamma\Delta^2(\epsilon + 2\Delta^3/m)$.
The claim follows since we pick $\epsilon^{-1},\epsilon'^{-1},m=n^{\Theta(1)}$ sufficiently large in Protocol \ref{pro0}.
\end{proof}

By repeating the protocol $O(1)$ times and apply the union bound, we obtain Theorem \ref{cor_reduction}.

\section{Simplified Protocol with $O(\Delta)$ Preprocessing Time\label{sec:SimplifiedProtocol}}

\subsection{Description of the Protocol}

\begin{protocol}\label{pro_new}
Let $m$ be a prime power. Pick the following objects:
\begin{itemize}\itemsep -0.3pt
\item an explicit pairwise independent generator $\gen=(\gen_0,\dots,\gen_{n-1}):\{0,1\}^{\ell}\to [m]^n$ with seed length $\ell$, and
\item an explicit $\epsilon$-\textsf{PRG} $\gen'=(\gen'_0,\dots,\gen'_{T-1}):\{0,1\}^{\ell'}\to \left(\{0,1\}^\ell\right)^{T}$ for $(T,n^2,2^\ell)$-branching programs with seed length $\ell'$
\end{itemize}
where $\epsilon^{-1},m=n^{\Theta(1)}$ are sufficiently large.

The initial node having the rumor independently chooses a random string $x\in\{0,1\}^{\ell'}$ which is appended with the rumor and sent to other nodes. Once one node gets the rumor, it gets the ID $u$.
Let $y=(y_0,\dots,y_{T-1})$ be the sequence of seeds generated by $\gen'$, i.e., $y_i=\gen'_i(x)$.
For $i\in [T]$ and $u\in V[G]$, define $(w_{u,i},z_{u,i})=\gen_u(y_i) \bmod 4\Delta \in [2\Delta] \times \{\mathrm{active},\mathrm{inactive}\}$. We say $u$ is active in the $i$th round if $z_{u,i}$ is active, and otherwise inactive. We say $u$ {\em selects} $v$ if $v$ is the $w_{u,i}$th neighbor of $u$.
In the $i$th round, an informed node $u$ sends the rumor to the unique neighbor $v$ (if exist) if $\{u,v\}$ is a good pair, where we call $\{u,v\}$ is a good pair if (i) $u$ is active, $v$ is inactive, and $u$ is the unique node selecting $v$, or (ii) the same holds with $u$ and $v$ swapped.
\end{protocol}

Checking the conditions requires
$u$ and $v$ knowing its index in the lists of its neighbors as well as the IDs of its neighbors. One can deterministically use $O(\Delta)$ preprocessing time to guarantee this assumption. Then Condition~(ii) can be checked directly by $u$. For Condition~(i), note that an active node $u$ can send the rumor and the seed to its unique inactive neighbor $v$ specified by $w_{i,u}$ and then $v$ can check if the condition is met, i.e., if $u$ is the unique node selecting $v$.
\footnote{The uniqueness requirement in Condition~(i) is necessary only for analyzing the associated averaging algorithm. For the sake of rumor spreading, dropping the requirement only make the rumor spread faster.}

\begin{thm}\label{thm_main_new}
Let $G$ be any graph with spectral gap $\alpha$ and irregularity $\beta$. Then
Protocol \ref{pro_new} uses $2\ell$ random bits, and with high probability informs all nodes of $G$ in $T=O(\beta^2\alpha^{-1}\log n)$ rounds.
\end{thm}

As a consequence, we obtain the following reduction:
\begin{cor}
Assume each node knows its index in the lists of its neighbors as well as the IDs of its neighbors. Then the following statements hold:
\begin{enumerate}
\item
Given an explicit $\epsilon$-\textsf{PRG} for $(T/2,n^2,2^\ell)$-branching programs with seed length $\ell'$, where $\epsilon^{-1}=n^{\Theta(1)}$ and $\ell=O(\log n)$ are sufficiently large,
there exists an explicit protocol using $2\ell'$ random bits, and  with high probability informs all nodes in $T=O((1/\alpha)\cdot\beta^2\log n)$ rounds.
\item
In particular, given an explicit $\epsilon$-\textsf{PRG} for $(T/2,n^2,\epsilon)$-branching programs with seed length $O(\log n)$ where $\epsilon^{-1}=n^{\Theta(1)}$ is sufficiently large, there exists an explicit protocol using $O(\log n)$ random bits, and with high probability informs all nodes  in $T=O((1/\alpha)\cdot\beta^2\log n)$ rounds.
\end{enumerate}
\end{cor}

Combining the reduction above with known explicit constructions of \PRG\ (Theorem \ref{thm:inw}), we obtain Theorem \ref{thm:ResultAssumption}.

We study Protocol \ref{pro_new} by analyzing the following associated averaging protocol, which is closely related to other gossip processes, e.g. random-matching model of load balancing processes. In the following, let  $\vec{v}(k)\in\mathbb{R}^{V[G]}$ denote the values of nodes after $k$ rounds.

\begin{protocol}[\textbf{Averaging Protocol}]\label{pro_new2}
Each node $u$ has a value $\vec{v}(0)_u$ specified by the distribution $\vec{v}(0)=\vec{e}_s$ where $s$ is the initial node.
Proceed as in Protocol~\ref{pro_new}. When node $u$ sends the rumor to node $v$, set the both values of $u$ and $v$ as the average of their original values.
\end{protocol}

We define the {\em averaging time} $\tau_{\mathsf{avg}}(\delta)$ of the protocol
as the smallest $k\in\mathbb{N}$ such that $\Pro{\|\vec{v}(k)^\perp\|_2 < \delta} > 1-\delta$ for any distribution $\vec{v}$, or $\infty$ if there is no such $k$.

\begin{thm}\label{thm_main_new2}
For $\delta>0$, assume $2\epsilon<\delta^2$ where $\epsilon$ is as in Protocol \ref{pro_new}.
Then Protocol \ref{pro_new2} uses $2\ell'$ random bits with $\tau_{\mathsf{avg}}(\delta)=O((1/\alpha)\cdot\beta^2\log (1/\delta))$.
\end{thm}

Theorem \ref{thm_main_new} is simple corollary of Theorem \ref{thm_main_new2} with $\delta=1/n$, since when $\|\vec{v}(k)^\perp\|_2<1/n$ then all $\vec{v}(k)_u$ must be nonzero, and $\vec{v}(k)_u\neq 0$ implies that $u$ is informed in $k$ rounds.

In Theorem \ref{thm_main_new2} we only consider initial values specified by $\vec{v}(0)=\vec{e}_s$. Assuming $\epsilon/\delta^2=n^{-\Theta(1)}$ is sufficiently small, it is easy to establish a upper bound $O(1/\alpha\cdot\beta^2(\log n + \log (1/\delta)))$ on the averaging time regarding a general distribution $\vec{v}(0)$: first use $T=O(1/\alpha\cdot\beta^2\log (1/\delta))$ rounds to inform all the nodes with high probability. Then set the new initial values $\vec{v}'(0)=\vec{v}(T)$, and run the averaging protocol for another $O(1/\alpha\cdot\beta^2(\log n + \log (1/\delta)))$ rounds.
The process with initial value distribution $\vec{v}'(0)$ can be viewed as a convex combination of those with initial value distribution $\vec{e}_u$, $u\in V[G]$ (note that each node $u$ is already informed).
With high probability, for all initial value distributions $\vec{e}_u$, the values converge to the average up to $\ell_2$-distance $\delta$. So the same is true for $\vec{v}'(0)$.

\subsection{Analysis of the Protocol}

For $x\in\{0,1\}^\ell$, define the following matrix
$$
\mat{M}(x)_{uv}=\begin{cases}
1/2 & \text{$u\neq v$ and $\{u,v\}$ is a good pair,}\\
1/2 & \text{$u=v$ and $\{u,v'\}$ is a good pair for some $v'\in V[G]$,}\\
1 & \text{$u=v$ and $\{u,v'\}$ is not a good pair for any $v'\in V[G]$,}\\
0 & \text{$u\neq v$ and $\{u,v\}$ is not a good pair}
\end{cases}
$$
where the set of good pairs are determined by the seed $y_i=x$ (see Protocol \ref{pro_new}, where the definition of good pairs are the same for all round number $i$).
It is easy to check that $\mat{M}(x)$ is
doubly stochastic, symmetric and $\mat{M}(x)^2=\mat{M}(x)$ for all $x\in\{0,1\}^\ell$. Moreover it characterizes the averaging operations using the seed $y_i=x$.
\begin{lem}\label{lem_faithful}
It holds that $\vec{v}(i+1)=\vec{v}(i)\mat{M}(y_i)$ for any $i\in[T]$.
\end{lem}
\begin{proof}
By definition, $\mat{M}(y_i)$ acts on $\mathbb{R}^{V[G]}$ by averaging the values of $u$ and $v$ for each good pair $\{u,v\}$. Protocol \ref{pro_new2} guarantees that averaging operations are performed for each good pair $\{u,v\}$, where $u$ or $v$ are already informed. If neither $u$ nor $v$ is informed, their values are both zero (by induction with the base case $\vec{v}(0)=\vec{e}_s$) and hence the averaging operation between them can be safely ignored.
\end{proof}

Let $\mat{M}=\EXX{x\in\{0,1\}^{\ell}}{\mat{M}(x)}$. Then $\mat{M}$ is doubly-stochastic. We have the following lemma:

\begin{lem}\label{lem_positivity}
$\mat{M}_{uv}\geq c\cdot\mathcal{L}_{1/2}\left(\mat{M}_{\reg(G)}\right)_{uv}$ for some constant $c\in (0,1)$.
\end{lem}

\begin{proof}
Each edge $\{u,v\}$ with $u\neq v$ is a good pair if either of the two mutually exclusive conditions (c.f. Protocol \ref{pro_new}) is met. The first one holds with probability at least
$$
\Prob{x\in\{0,1\}^{\ell}}{\text{$u$ is active and selects $v$}}-\sum_{u'\in N(v)\setminus\{u\}}\Prob{x\in\{0,1\}^{\ell}}{\text{$u$ is active and both $u, u'$ select $v$}}
$$
taken over the seed $y_i=x$.
As $\gen$ is a pairwise independent generator, by Lemma \ref{lem:pairwisemodulo}, this probability is lower bounded by $\left(\frac{1}{4\Delta}-\frac{2}{m}\right)-
\Delta\cdot\left(\frac{1}{4\Delta}\cdot\frac{1}{2\Delta}+\frac{2}{m}\right)\geq \frac{c}{2\Delta}$ for some $c>0$ and $m=\Omega(\Delta^{2})$.  The case for the second condition is the same.
So $\{u,v\}$ is a good pair with probability at least $\frac{c}{\Delta}$. Note that $\mat{M}(x)_{uv}=1/2$ whenever $\{u,v\}$ is a good pair. Therefore $$
\mat{M}_{uv}=\EXX{x\in\{0,1\}^{\ell}}{\mat{M}(x)_{uv}}\geq \frac{c}{2\Delta}=c\mathcal{L}_{1/2}\left(\mat{M}_{\reg(G)}\right)_{uv}.
$$
For $u=v$, note that $\mat{M}_{uv}\geq 1/2$ by definition and $\mathcal{L}_{1/2}\left(\mat{M}_{\reg(G)}\right)_{uv}\leq 1$.
\end{proof}

Again let $\vec{\pi}\in\mathbb{R}^{V[G]}$ denote the uniform distribution over $V[G]$.
\begin{lem}\label{lem_shrink}
For any  $\vec{v}\in\mathbb{R}^{V[G]}$ orthogonal to $\vec{\pi}$, it holds that $0\leq \EXX{x\in\{0,1\}^{\ell}}{\left\|\vec{v}\mat{M}(x)\right\|_2}  \leq  (1-c\beta^{-2}\alpha)\|\vec{v}\|_2$ for some constant $c\in (0,1)$.
\end{lem}
\begin{proof}
The non-negativity is obvious. For the upper bound, we have
\begin{align*}
\EXX{x\in\{0,1\}^{\ell}}{\left\|\vec{v}\mat{M}(x)\right\|_2}
&=\EXX{x\in\{0,1\}^{\ell}}{\vec{v}\mat{M}(x)\mat{M}(x)^{\intercal}\vec{v}^{\intercal}}\\
&=\vec{v}\EXX{x\in\{0,1\}^{\ell}}{\mat{M}(x)\mat{M}(x)^{\intercal}}\vec{v}^{\intercal}\\
&=\vec{v}\EXX{x\in\{0,1\}^{\ell}}{\mat{M}(x)}\vec{v}^{\intercal}\\
&=\vec{v}\mat{M}\vec{v}^{\intercal}.
\end{align*}
Let $\mat{M}'=\mat{M}-c\cdot \mathcal{L}_{1/2}(\mat{M}_{\reg(G)})$ where $c$ is as in Lemma \ref{lem_positivity}. Then $\mat{M}'$ is a non-negative matrix by Lemma \ref{lem_positivity}. As both $\mat{M}$ and $\mathcal{L}_{1/2}\left(\mat{M}_{\reg(G)}\right)$ are doubly-stochastic, so is $\mat{M}'/(1-c)$. Then $\lambda_{\max}(\mat{M'})\leq \|\mat{M}'\|_2\leq 1-c$.
Note that $\lambda_{\max}\left(\mathcal{L}_{1/2}\left(\mat{M}_{\reg(G)}\right)\right)\leq 1-\beta^{-2}\alpha/2$. Therefore
$$
\lambda_{\max}(\mat{M})\leq \lambda_{\max}(\mat{M'})+c\cdot \lambda_{\max}\left(\mathcal{L}_{1/2}\left(\mat{M}_{\reg(G)}\right)\right)
\leq 1-(c/2)\beta^{-2}\alpha
$$
and the claim follows.
\end{proof}

\begin{lem}\label{lem_shrink2}
For any $\vec{v}\in\mathbb{R}^{n}$ orthogonal to $\vec{\pi}$ and $k\in [T]$, it holds that
$$
\EXX{y_0,\dots,y_{k-1}\in\{0,1\}^{\ell}}{\left\|\vec{v}\prod_{i=0}^{k-1}\mat{M}(y_i)\right\|_2}  \leq  (1-c\beta^{-2}\alpha)^k\|\vec{v}\|_2
$$ for some constant $c\in (0,1)$.
%\footnote{$\prod_{i=0}^{k-1}\mat{M}(s_i)$ means $\mat{M}(s_0)\dots\mat{M}(s_{k-1})$.}
\end{lem}

\begin{proof}
Induct on $k$. The claim is trivial for $k=0$.
For $k>0$, assume the claim holds for $k'<k$. Let
$\vec{v}\in\mathbb{R}^n$ be a vector orthogonal to $\vec{\pi}$, and define
 $\vec{v}'=\vec{v}\prod_{i=0}^{k-2}\mat{M}(y_i)$. Then $\vec{v}'$ is also orthogonal to $\vec{\pi}$. So
\begin{align*}
\EXX{y_0,\dots,y_{k-1}\in\{0,1\}^{\ell}}{\left\|\vec{v}\prod_{i=0}^{k-1}\mat{M}(y_i)\right\|_2}
&=\EXX{y_0,\dots,y_{k-2}\in\{0,1\}^{\ell}}{\EXX{y_{k-1}\in\{0,1\}^{\ell}}{\left\|\vec{v}'\mat{M}(y_{k-1})\right\|_2}}\\
&\leq \EXX{y_0,\dots,y_{k-2}\in\{0,1\}^{\ell}}{(1-c\beta^{-2}\alpha)\left\|\vec{v}'\right\|_2}\\
&\leq (1-c\beta^{-2}\alpha)^k\|\vec{v}\|_2.
\end{align*}
The first inequality uses Lemma \ref{lem_shrink} and the second one uses the induction hypothesis.
\end{proof}

Let $\mathcal{P}$ be the distribution of $y=(y_0,\dots,y_{T-1})$ in Protocol \ref{pro_new}. Then we have

\begin{lem}\label{lem_bp_new}
For any $u\in V[G]$,
$$
\left|\EXX{y\sim\mathcal{P}}{\left\|\vec{e}_u\prod_{i=0}^{T-1}\mat{M}(y_i)\right\|_2}
-\EXX{y\in\left(\{0,1\}^\ell\right)^T}{\left\|\vec{e}_u\prod_{i=0}^{T-1}\mat{M}(y_i)\right\|_2}\right|\leq \epsilon
$$
where $\epsilon$ is as in Protocol \ref{pro_new}.
\end{lem}

\begin{proof}
For $x\in\{0,1\}^{\ell}$, write $\mat{M}(x)=\frac{1}{2}\mat{M}_{\text{lazy}}(x)+
\frac{1}{2}\mat{M}_{\text{non-lazy}}(x)$
where $\mat{M}_{\text{lazy}}(x)$ is simply the identity matrix $\mat{I}$, and $\mat{M}_{\text{non-lazy}}(x)$ is the following permutation matrix:
$$
\left(\mat{M}_{\text{non-lazy}}(x)\right)_{uv}=\begin{cases}
1 & \text{$u\neq v$ and $\{u,v\}$ is a good pair}\\
0 & \text{$u=v$ and $\{u,v'\}$ is a good pair for some $v'\in V[G]$,}\\
1 & \text{$u=v$ and $\{u,v'\}$ is not a good pair for any $v'\in V[G]$,}\\
0 & \text{$u\neq v$ and $\{u,v\}$ is not a good pair.}
\end{cases}
$$
As before, let $\mathcal{C}_T=\{\text{lazy},\text{non-lazy}\}^T$.
Note that for any $y=(y_0,\dots,y_{T-1})\in\left(\{0,1\}^\ell\right)^T$,
we have
\begin{align*}
\left\|\vec{e}_u\prod_{i=0}^{T-1}\mat{M}(y_i)\right\|_2
&=
\vec{e}_u\left(\prod_{i=0}^{T-1}\mat{M}(y_i)\right)\left(\prod_{i=0}^{T-1}\mat{M}(y_i)\right)^\intercal\vec{e}_u^\intercal\\
&=2^{-2T}\sum_{c,c'\in\mathcal{C}_T}\vec{e}_u\left(\prod_{i=0}^{T-1}\mat{M}_{c_i}(y_i)\right)\left(\prod_{i=0}^{T-1}\mat{M}_{c'_i}(y_i)\right)^\intercal\vec{e}_u^\intercal\\
&=2^{-2T}\sum_{c,c'\in\mathcal{C}_T,v\in V[G]}(\vec{e}_u\otimes\vec{e}_u)
\prod_{i=0}^{T-1}\left(\mat{M}_{c_i}(y_i)\otimes\mat{M}_{c'_i}(y_i)\right)(\vec{e}_v\otimes\vec{e}_v)^\intercal.
\end{align*}
For any $c,c'\in\mathcal{C}_T$,  it is easy to construct a $(T,n^2,2^\ell)$-branching program $\mathcal{B}_{c,c'}$ that has state set $V[G]\times V[G]$, such that for any node $v\in V[G]$ and input $y=(y_0,\dots,y_{T-1})$,
it holds that $\mathcal{B}_{c,c'}((u,u),y)=(v,v)$ (resp. $\mathcal{B}_{c,c'}((u,u),y)\neq (v,v)$) iff
$$
(\vec{e}_u\otimes\vec{e}_u)
\prod_{i=0}^{T-1}\left(\mat{M}_{c_i}(y_i)\otimes\mat{M}_{c'_i}(y_i)\right)(\vec{e}_v\otimes\vec{e}_v)^\intercal.
$$
equals 1 (resp. 0).
More specifically, The transition matrix between the $i$th and the $(i+1)$st layer of $\mathcal{B}_{c,c'}$ with edge label $y_i$ is just $\mat{M}_{c_i}(y_i)\otimes\mat{M}_{c'_i}(y_i)$.
Then the absolute difference between
$\EXX{y\sim\mathcal{P}}{\left\|\vec{v}\prod_{i\in [k]}\mat{M}(y_i)\right\|_2}$ and
$\EXX{y\in\left(\{0,1\}^\ell\right)^T}{\left\|\vec{v}\prod_{i\in [k]}\mat{M}(y_i)\right\|_2}$ is bounded by
$$
2^{-2T}\sum_{c,c'\in\mathcal{C}_T,v\in V[G]}\left|
\Prob{y\sim\mathcal{P}}{\mathcal{B}_{c,c'}((u,u),y)=(v,v)}-
\Prob{y\in\left(\{0,1\}^\ell\right)^T}{\mathcal{B}_{c,c'}((u,u),y)=(v,v)}
\right|
$$
which is bounded by $\epsilon$ since $\gen$ is an $\epsilon$-\textsf{PRG} for $(T,n^2,2^\ell)$-branching programs.
\end{proof}

\begin{proof}[Proof of Theorem \ref{thm_main_new2}]
By Lemma \ref{lem_shrink2}, we have
$$
\EXX{y\in\left(\{0,1\}^\ell\right)^T}{\left\|\vec{e}_s^{\perp}\prod_{i=0}^{T-1}\mat{M}(y_i)\right\|_2} \leq  (1-c\beta^{-2}\alpha)^T.
$$
Combining this with Lemma \ref{lem_bp_new} and using the fact that $\|\vec{v}\|_2=\|\vec{v}^\perp\|_2+\|\vec{\pi}\|_2$ for any distribution $\vec{v}$, we obtain
$$
\EXX{y\sim\mathcal{P}}{\left\|\left(\vec{e}_s\prod_{i=0}^{T-1}\mat{M}(y_i)\right)^{\perp}\right\|_2}
=
\EXX{y\sim\mathcal{P}}{\left\|\vec{e}_s^{\perp}\prod_{i=0}^{T-1}\mat{M}(y_i)\right\|_2}\leq  (1-c\beta^{-2}\alpha)^T+\epsilon < \delta^2
$$
for sufficiently large $T=O(\beta^2\alpha^{-1}\log \delta^{-1})$. The claim then follows from Lemma \ref{lem_faithful} and the Markov's inequality.
\end{proof}

\section{Omitted Details in \secref{general}}

\subsection{Preliminaries\label{sec:OmittedBackground}}

In this subsection we list all necessary definitions and results that are used to construct the protocols in \secref{general}.

\paragraph{Unbalanced Expanders with Near-Optimal Expansion}

We consider the following kind of left-regular bipartite graphs.
\begin{defi}\label{def:bigraph}
Let $\Gamma: [N]\times [D]\to \bigsqcup_{i\in [D]}[M_i]$ be a function
%where $M_0=\dots=M_{D-1}=M$ and hence $[M_i]$'s are identical copies of $[M]$. Moreover suppose
where $\Gamma(x,y)\in [M_y]$ for any $x\in [N]$, $y\in [D]$. Function
$\Gamma$ specifies a left-degree $D$ bipartite graph with left vertex set $[N]$ and right vertex set $\bigsqcup_{i\in [D]}[M_i]$ in the following way: for $x\in [N]$ and $y\in [D]$, the $y$th neighbor of $x$ is given by $\Gamma(x,y)$.
\end{defi}

We are interested in graphs $\Gamma$ exhibiting excellent expansion properties. This leads to the notion of unbalanced expanders \cite{TUZ07, GUV09}.
\begin{defi}[Unbalanced expanders \cite{TUZ07, GUV09}]
Let $\Gamma: [N]\times [D]\to \bigsqcup_{i\in [D]}[M_i]$ be as in \defref{bigraph}.
We call $\Gamma$ a $(K,A)$-expander if for any set $S\subseteq [N]$ of size $K$, it holds that $|N(S)|\geq AK$. We call $\Gamma$  a $(\mathord{\leq} K, A)$-expander if it is a $(K', A)$-expander for all $K'\leq K$.
\footnote{The definition here is slightly different from \cite{TUZ07, GUV09} as we require $\Gamma(x,y)\in [M_y]$. This is analogous to the difference between standard and strong condensers.}
\end{defi}

In particular we are interested in  $(K,A)$-expanders, where the parameter $A=(1-\eps)D$ for small $\eps$, i.e. for any subset $S$ of size $K$ from the left set $[N]$, there is almost no collision among the neighbors of nodes in $S$.
Explicit constructions of such unbalanced expanders with near-optimal expansion are known.

\begin{thm}[\cite{GUV09}]\label{thm:losslessexp}
For any $N\in\mathbb{N}$, $K\leq N$, and $\epsilon>0$, there is
an explicit $(K, (1-\epsilon)D)$-expander $\Gamma: [N]\times [D]\to \bigsqcup_{i\in [D]}[M_i]$
with $D=\left(\frac{\log N}{\epsilon}\right)^{O(1)}$ and $M_0=\dots=M_{D-1}\leq \max\left\{D,K^{O(1)}\right\}$.
\end{thm}

Assume that $\Gamma:[N]\times [D]\to \bigsqcup_{i\in [D]}[M_i]$ is a $(K,(1-\epsilon)D)$-expander. We consider the map $\Gamma(\cdot, U)$ applied on any $K$ elements of $[N]$ where $U$ is uniformly distributed over $[D]$. The following lemma states that with high probability these $K$ elements are mapped into $\bigsqcup_{i\in [D]}[M_i]$ with almost no collision.

\newcommand{\condenserp}{
Let $\Gamma:[N]\times [D]\to \bigsqcup_{i\in [D]}[M_i]$ be a $(K,(1-\epsilon)D)$-expander.
Let $S$ be a subset of $[N]$ of size $K$.
Then for at least $\left(1-\sqrt{\epsilon}\right)$-fraction of $y\in [D]$, it holds that
$|\{\Gamma(x,y):x\in S\}|\geq (1-\sqrt{\epsilon})K$.
}

\mylemma{condenserp}{\condenserp}

\begin{proof}
The size of $N(S)=\bigsqcup_{y\in [D]}\{\Gamma(x,y): x\in S\}$ is at least $(1-\epsilon)DK$ as $\Gamma$ is a $(K,(1-\epsilon)D)$-expander. So $\mathbf{E}_y\left[|\{\Gamma(x,y): x\in S\}|\right]\geq (1-\epsilon)K$ with $y$ uniformly distributed over $[D]$.
Also note that $|\{\Gamma(x,y): x\in S\}|\leq |S|=K$ for any $y\in [D]$. Applying Markov's inequality on $K-|\{\Gamma(x,y): x\in S\}|$, we have $\mathbf{Pr}_y[|\{\Gamma(x,y): x\in S\}|<(1-\sqrt{\epsilon})K]\leq \sqrt{\epsilon}$.
\end{proof}

\subsection{Analysis of Protocol~\ref{pro1}}

We start by analyzing a single round $t$ and see the properties of our protocol. Let $I_t$ be the set of informed nodes after round $t$, and $U_t$ the set of uninformed nodes after round $t$. Remember that all the random choices in round $t$ are determined by $(x_t,y_t)$.

We need the following lemma:

\newcommand{\pairwise}{
Fix any round $0\leq t<T$. For any $u\in U_t$, $v\in I_t$, let $X_{v\rightarrow u}$ be the boolean random variable whose value is $1$ iff $v$ informs $u$ in round $t+1$. Then it holds that
\begin{enumerate}
\item $|\Ex{X_{v\rightarrow u}}-1/\Delta|\leq\epsilon$ for any $u\in U_t$, $v\in I_t$;
\item $\Cov{X_{v\rightarrow u}, X_{v'\rightarrow u'}}\leq \epsilon$ for any $u,u'\in U_t$, $v,v'\in I_t$ satisfying $(u,v)\neq (u',v')$.
\end{enumerate}}
\mylemma{pairwise}{\pairwise}

\begin{proof}
For any $u\in U_t$ and  $v\in I_t$, suppose the index of $u$ in the adjacency list of $v$ is $z$. By construction, $X_{v\rightarrow u}$ equals $1$ iff $\gen_{\Gamma(v,x_t)}(y_t)\bmod\Delta=z$. Fix $x_t$. The fact that $\gen$ is a pairwise independent generator together with \lemref{pairwisemodulo} shows that $|\Ex{X_{v\rightarrow u}}-1/\Delta|\leq 2/m\leq\epsilon$.

For any $u,u'\in U_t$ and $v,v'\in I_t$,  first assume $v\neq v'$.
Suppose the index of $u$ (resp. $u'$) in the adjacency list of $v$ (resp. $v'$) is $z$ (resp $z'$). By construction, $X_{v\rightarrow u}$ equals $1$ iff $\gen_{\Gamma(v,x_t)}(y_t)\bmod\Delta=z$, and similarly for $X_{v'\rightarrow u'}$.
By \lemref{condenserp} and the fact that $\Gamma$ is a $(K, (1-\eps^2/4)D)$-expander, the event $|\{\Gamma(v,x_t), \Gamma(v',x_t)\}|\geq (1-\epsilon/2)\cdot 2>1$ occurs with probability at least $1-\epsilon/2$ over the choices of $x_t$. Condition on any $x_t$ such that this event occurs. We have $\Gamma(v,x_t)\neq \Gamma(v',x_t)$. Using the fact that $\gen$ is pairwise independent together with \lemref{pairwisemodulo}, we have $\Cov{X_{v\rightarrow u}, X_{v'\rightarrow u'}}\leq 2/m$. For the other choices of $x_t$, we have $\Cov{X_{v\rightarrow u}, X_{v'\rightarrow u'}}\leq 1$ since $X_{v\rightarrow u}, X_{v'\rightarrow u'}$ are boolean.
Therefore $\Cov{X_{v\rightarrow u}, X_{v'\rightarrow u'}}\leq (1-\epsilon/2)(2/m)+(\epsilon/2)\leq\epsilon$ for random $x_t$.

Now assume $v=v'$ and hence $u\neq u'$. We have
\begin{align*}
\Cov{X_{v \to u }, X_{v \to u' } }&=\Ex{X_{v \to u } \cdot X_{v \to u'    } } - \Ex{ X_{v \to u }   } \cdot \Ex{ X_{v \to u' }   }\\
&=0 - \Ex{ X_{v \to u }   } \cdot \Ex{ X_{v \to u' }   }\leq 0.\qedhere
\end{align*}
\end{proof}

Next we prove the following lemma:

\newcommand{\keylemma}{
Fix a round $0\leq t < T$ and the set $I_{t}$ of informed nodes before round $t+1$.
Fix also an arbitrary set of edges $F \subseteq E(I_t,U_t)$.
Let $J$ be the set of nodes that become informed in round $t+1$ if we consider only transmissions of the rumor along the edges in $F$.
\begin{enumerate}
\item \label{case:Jempty}
$
\Pro{ J \neq \emptyset}
\geq c_1 \min\{|F|/\Delta, 1\}
$ for some constant $c_1>0$.

\item \label{case:general}
If $|F|=\Omega(\Delta)$ then $\Pro{ | J | \geq c_2 |F|/\Delta} \geq c_3$ for some constant $c_2, c_3>0$.
\end{enumerate}
}
\mylemma{keylemma}{\keylemma}

\begin{proof}
Let $X_{v\rightarrow u}$ be the boolean random variable whose value is $1$ iff $v$ informs $u$ in round $t+1$.

We first prove \eqref{case:Jempty}. Let $k=|F|$ and suppose $F=\{(v_0,u_0),\dots,(v_{k-1},u_{k-1})\}$. Let $X=\sum_{i\in [k]}X_{v_i\rightarrow u_i}$.
Then by Cauchy-Schwarz inequality, $\Ex{\mathbf{1}_{X>0}}\geq (\Ex{X})^2/\Ex{X^2}$. By \lemref{pairwise}, it holds that
$$
\Ex{X}=\sum_{i\in [k]}\Ex{X_{v_i\rightarrow u_i}}\geq k(1/\Delta-\epsilon)=\Omega(|F|/\Delta)
$$
and
\begin{align*}
\Ex{X^2}&=\sum_{i,j\in [k]}\Ex{X_{v_i\rightarrow u_i}X_{v_j\rightarrow u_j}}\\
&=\sum_{i\in [k]}\Ex{X_{v_i\rightarrow u_i}}+\sum_{\substack{i,j\in [k]\\i\neq j}}(\Ex{X_{v_i\rightarrow u_i}}\Ex{X_{v_j\rightarrow u_j}}+\Cov{X_{v_i\rightarrow u_i},X_{v_j\rightarrow u_j}})\\
&\leq k(1/\Delta+\epsilon)+(k^2-k)((1/\Delta+\epsilon)^2+\epsilon)= O(|F|/\Delta+|F|^2/\Delta^2)
\end{align*}
where we use the condition that $\epsilon=\Delta^{-\Theta(1)}$ is sufficiently small.
So
$$
\Pro{J\neq\emptyset}=\Ex{\mathbf{1}_{X>0}}\geq (\Ex{X})^2/\Ex{X^2}=\Omega(\min\{|F|/\Delta,1\}),
$$
and the first statement  follows.

Next we prove the second statement. For $u\in U_t$, let $F_u$ be the set of edges in $F$ incident to $u$,  $Z_u$ be the boolean random variable whose value is $1$ iff $u$ is informed in round $t+1$ via edges in $F_u$, and $X_u=\sum_{(v,u)\in F_u} X_{v\rightarrow u}$. So $Z_u=\mathbf{1}_{X_u>0}$ and $|J|=\sum_{u\in U_t} Z_u$. For $u\in U_t$, $\Ex{Z_u}=\Ex{\mathbf{1}_{X_u>0}}\geq (\Ex{X_u})^2/\Ex{X_u^2}=\Omega(|F_u|/\Delta)$ by a similar argument as above. So $\Ex{|J|}=\Omega(\sum_{u\in U_t}|F_u|/\Delta)=\Omega(|F|/\Delta)$. Suppose $\Ex{|J|}\geq c|F|/\Delta$ for constant $c>0$.

On the other hand, for any $c_2\geq 0$, we have
\begin{align*}
\Ex{|J|}&=\Ex{\mathbf{1}_{|J|\geq c_2|F|/\Delta}\cdot |J|}+\Ex{\mathbf{1}_{|J|< c_2|F|/\Delta}\cdot |J|}\\
&\leq \Ex{\mathbf{1}_{|J|\geq c_2|F|/\Delta}\cdot |J|}+\Ex{\mathbf{1}_{|J|< c_2|F|/\Delta}}\cdot c_2|F|/\Delta
\end{align*}
and hence $\Ex{\mathbf{1}_{|J|\geq c_2|F|/\Delta}\cdot |J|}\geq \Ex{|J|}-\Ex{\mathbf{1}_{|J|< c_2|F|/\Delta}}\cdot c_2|F|/\Delta\geq (c-c_2)|F|/\Delta$. Pick $c_2=c/2$. By Cauchy-Schwarz inequalty, we have
\begin{equation}
\label{eq:boundJ}
\Pro{|J|\geq c_2|F|/\Delta}=\Ex{\mathbf{1}_{|J|\geq c_2|F|/\Delta}}\geq
\frac{\left(\Ex{\mathbf{1}_{|J|\geq c_2|F|/\Delta}\cdot |J|}\right)^2}{\Ex{|J|^2}}
\geq \frac{((c-c_2)|F|/\Delta)^2}{\Ex{|J|^2}}.
\end{equation}
Note that
\begin{align*}
\Ex{|J|^2}&=\sum_{u\in U_t}\Ex{Z_u}+\sum_{\substack{u,u'\in U_t\\u\neq u'}}\Ex{Z_u Z_{u'}}\\
&\leq\Ex{|J|}+\sum_{\substack{u,u'\in U_t\\u\neq u'}}\Ex{X_u X_{u'}}\\
&=\Ex{|J|}+\sum_{\substack{u,u'\in U_t\\u\neq u'}}\sum_{\substack{(v,u)\in F_u\\ (v',u')\in F_{u'}}}\left(\Ex{X_{v\rightarrow u}}\Ex{X_{v'\rightarrow u'}}+\Cov{X_{v\rightarrow u},X_{v'\rightarrow u'}}\right)\\
&\leq\Ex{|J|}+\sum_{\substack{u,u'\in U_t\\u\neq u'}}\left(\left(\sum_{(v,u)\in F_u}
\Ex{X_{v\rightarrow u}}\right)\left(\sum_{(v',u')\in F_{u'}}\Ex{X_{v'\rightarrow u'}}\right)+|F_u||F_{u'}|\epsilon\right)\\
&=\Ex{|J|}+O\left(\sum_{\substack{u,u'\in U_t\\u\neq u'}} |F_u||F_{u'}|/\Delta^2\right)\\
&=\Ex{|J|}+O\left(\left(\sum_{u\in U_t}|F_u|\right)^2/\Delta^2\right)\\
&=\Ex{|J|}+O\left(|F|^2/\Delta^2\right).
\end{align*}
Here $\Ex{|J|}=\sum_{u\in U_t}\Ex{Z_u}\leq \sum_{u\in U_t}\Ex{X_u}=\sum_{u\in U_t}O(|F_u|/\Delta)=O(|F|/\Delta)$. Using the condition $|F|=\Omega(\Delta)$, we have $\Ex{|J|^2}=O\left(|F|^2/\Delta^2\right)$. Substitute it in \eqref{eq:boundJ}, and then the second statement follows.
\end{proof}

Now we prove \thmref{gen_graph_result}. We first define a matrix $\mathcal{M}\in\mathbb{R}^{n\times n}$ that is associated with graph $G$. For any $u,v\in V[G]$, let $\mathcal{M}_{u,v}=1/\Delta$ if $\{u,v\}\in E[G]$, $\mathcal{M}_{u,v}=1-\mathrm{deg}(u)/\Delta$ if $u=v$,  and $\mathcal{M}_{u,v}=0$ otherwise. Notice that matrix $\mathcal{M}$ is doubly stochastic.
We further define the conductance of matrix $\mathcal{M}$ by
\[
\Phi(\mathcal{M})\triangleq
\min_{\substack{A\subset V \\ |A|\leq n/2}}
\frac{e(A,\overline{A})}{ \Delta\cdot|A|}.
\]
Notice that $ \Phi(\mathcal{M})\leq \phi(G)\leq \Phi(\mathcal{M})\cdot \beta$, where $\beta\triangleq \Delta/\delta$. Hence it suffices to work with $\Phi(\mathcal{M})$ in the following.

\begin{proof}[Proof of \thmref{gen_graph_result}]
The proof is divided into four phases, depending on the number of informed nodes $|I_t|$ after round $t$.

\medskip
{\bf Phase 1: $1 \leq |I_t| \leq 1/\Phi$}. This phase is divided into several subphases. For every $1 \leq i \leq \log (1/\phi)$, subphase $i$ begins when the number of informed nodes is at least $2^{i-1}$ and ends when this number is at least $2^i$.
Assume that we are at the beginning of the $i$th subphase. Fix an arbitrary round $t$ of the $i$th subphase and the set of informed nodes $I_t$; thus, $2^{i-1} \leq |I_t| < 2^{i}$. We consider the number of nodes that become informed in round $t+1$. Applying \lemref{keylemma}\eqref{case:Jempty} with $F=E(I_t, U_t)$ gives
$$
 \Pro{ |I_{t+1} \setminus I_t | \geq 1} \geq c_1 \min\{e(I_t, U_t)/\Delta,1\}\geq c_1\min\{\Phi\cdot |I_t|/\beta,1\},
$$
Let $p\triangleq c_1\min\{\Phi\cdot|I_t|/\beta,1\}$, and hence $p=O(\Phi\cdot|I_t|)$ since $|I_t|\leq 1/\Phi$ and $\beta\geq 1$.
Therefore, the expected time to increase $|I_t|$ from $2^{i-1}$ to $2^{i}$ is at most
$2^{i-1}/p = O(1/\Phi)$.
By Markov's inequality,
\[
  \Pro{ |I_{t+\tau}| \leq 2^{i} \, \mid \, |I_{t}| \geq 2^{i-1} } \leq 1/2
\]
for some $\tau=O(\Phi^{-1})$.
Hence the time to complete Phase 1 can be upper bounded by $\tau = \Oh( (1/\Phi) )$ multiplied with the sum of $\log (1/\Phi) = \Oh(\log n)$ independent geometric  random variables each with parameter $1/2$. Applying a Chernoff bound for the sum of independent geometric random variables yields that the number of rounds required for Phase $1$ is at most $\Oh( (1/\Phi) \cdot \log n)=\Oh((1/\phi)\cdot\beta\cdot\log n)$ with high probability.

\medskip
{\bf Phase 2: $1/\Phi\leq |I_t| \leq n/2$.} Fix a round $t$ and the set of informed nodes $I_t$. We
 apply \lemref{keylemma}\eqref{case:general}, with $F=E(I_t, U_t)$. Note that the precondition $|F|=\Omega(\Delta)$ is satisfied, as
\[
 |F| = e(I_t, U_t) \geq \Phi \cdot \Delta \cdot |I_t| \geq \Phi \cdot \Delta \cdot  (1/\Phi) = \Omega(\Delta).
\]
Hence we conclude from \lemref{keylemma}\eqref{case:general} that
\begin{align*}
\Pro{ |I_{t+1} \setminus I_t| \geq c_2\cdot  \phi\cdot \delta\cdot |I_t| /\Delta} &\geq c_3,
\end{align*}
for some constant $c_2,c_3>0$. When this event occurs, we have $|I_{t+1}|\geq (1+c_2\cdot \phi/\beta)|I_t|$.
So, the number of rounds until we have $|I_t| \leq n/2$ can be upper bounded by the sum of $\log_{1+c_2\cdot\phi/\beta} (n/2)= \Oh( (1/\phi)\cdot\beta\cdot \log n)$ independent geometric random variables with parameters $c_3$. Using again the Chernoff bound we obtain that Phase $2$ is completed within at most $\Oh( (1/\phi)\cdot \beta\cdot \log n)$ rounds with high probability.

\medskip
{\bf Phase 3: $n/2 \leq |I_t| \leq n - 1/\Phi$.} The analysis is the same as in Phase 2 with the roles of $I_t$ and $U_t$ switched.

\medskip
{\bf Phase 4: $n - 1/\Phi \leq |I_t| \leq n$.} Again, the analysis is the same as in Phase 1 with the roles of $I_t$ and $U_t$ switched.

Since each of the four phases requires only $\Oh( (1/\phi)\cdot\beta \cdot \log n)$ rounds with high probability, the result follows by applying the union bound.
\end{proof}

\subsection{Analysis of Protocol~\ref{pro_PRG}}

We first remark that the condition $\alpha=1-o(1)$ is equivalent to $\lambda\triangleq\lambda_2=o(1)$, which will be used in the following.

To relate the spectral expansion of $G$ with the expansion property, we use the following expander mixing lemma for general graphs.

\begin{lem}[Expander Mixing Lemma for General Graphs~\cite{chung1}]\label{lem:geml}
Let $G$ be a general graph. Then for any subset $X$ and $Y$ it holds that
\[
\left|
e(X,Y)-\frac{\vol(X)\cdot\vol(Y)}{\vol(G)}
\right|\leq\lambda\cdot
\frac{\sqrt{\vol(X)\cdot\vol(Y)\cdot\vol(\overline{X})\cdot\vol(\overline{Y})}}{\vol(G)}.
\]

\end{lem}

In order to prove \thmref{expander_result},  it suffices to show the following lemma:

\newcommand{\expandermainlemma}{
Let $G$ be a graph that satisfies the preconditions of \thmref{expander_result}. Then with high probability  all the following statements hold:
\begin{itemize}\itemsep -0.33pt
\item \textbf{Phase \uppercase\expandafter{\romannumeral1}}
Suppose $1\leq|I_t|\leq n/\log n$. Then there is $\tau=\log n+o(\log n)$ such that~$|I_{t+\tau}|>n/\log n$.
\item \textbf{Phase \uppercase\expandafter{\romannumeral2}} Suppose $ n/\log n \leq
|I_t| \leq n - n/\log n$. Then there is $\tau=o(\log n)$ such that~$|I_{t+\tau}|>n - n/\log n$.
\item \textbf{Phase \uppercase\expandafter{\romannumeral3}} Suppose $|I_t|\geq n- n/\log
n $. Then there is $\tau=\ln n+o(\log n)$ such that~$|I_{t+\tau}| =
n$.
\end{itemize}}
\mylemma{expandermainlemma}{\expandermainlemma}

\begin{proof}

%Throughout the whole proof, we fix an arbitrary seed $x$.
For any round $t$ and $u\in U_t$, $v\in I_t$, let $X_{v\rightarrow u}$ be the boolean random variable whose value is $1$ iff $v$ informs $u$ in round $t+1$.
Note that $\Gamma$ is a $(\mathord{\leq} K, (1-\epsilon^2/4)D)$-expander and hence a $(2, (1-\epsilon^2/4)D)$-expander. And $\gen$ is a pairwise independent generator. Then we observe that the statements in \lemref{pairwise} hold here as well by the same proof. Notice that it holds by \lemref{geml} that
\begin{align}
e(I_t, U_t) &\geq \frac{\vol(I_t)\cdot \vol(U_t)}{\vol(G)} -
\lambda\cdot\frac{\vol(I_t)\cdot \vol(U_t)}{\vol(G)}\nonumber\\
&\geq (1-\lambda)\cdot\frac{\vol(I_t)
\cdot(\vol(G)-\vol(I_t))}{\vol(G)}\label{eq:EMPexp}
\end{align}

\textbf{Phase~I.}
By \eq{EMPexp} we have
\[
e(I_t, U_t) \geq (1-\lambda)\cdot \delta \cdot |I_t| \left(1-\frac{\Delta\cdot|I_t|}{nd}\right).
\]
Since $\lambda=o(1)$ and $|I_t|\leq n/\log n$, we have
\begin{align}
e(I_t, U_t)&\geq (1-o(1))\cdot \Delta \cdot |I_t| \left(\frac{\delta}{\Delta}-\frac{\delta}{d\cdot\log n}\right)\geq \left(1- \frac{1}{\log n}-o(1)\right)\cdot \Delta \cdot |I_t|.\label{eq:lb_edgeset}
\end{align}
Hence
\[
|N(I_t)\setminus I_t| \geq \frac{e(I_t, U_t)}{\Delta} \geq \left(1- \frac{1}{\log n}-o(1)\right) \cdot |I_t|.
\]

Define $\gamma\triangleq\lambda+\frac{1}{\log n}$, and $A\triangleq\{u\in N(I_t)\setminus I_t: |N(u)\cap I_t| \geq 2d\sqrt{\gamma} \}$. Then $e(A, I_t)\geq |A|\cdot 2d\cdot\sqrt{\gamma}$. On the other hand by \lemref{geml} it holds that
\begin{align*}
e(A,I_t)&\leq \frac{\vol(A)\cdot \vol(I_t)}{\vol(G)} + \lambda\sqrt{\vol(A)\cdot\vol(I_t)} \\
& \leq \frac{\Delta^2\cdot |A|\cdot  |I_t|}{ nd } + \gamma\Delta\cdot\sqrt{|A|\cdot |I_t|}.
\end{align*}
By the definition of set $A$ we have $e(A,I_t)\geq 2d\sqrt{\gamma}\cdot|A|$, and hence
\begin{align*}
|A|\cdot 2d\cdot\sqrt{\gamma}
&\leq \frac{\Delta^2\cdot |A|\cdot  |I_t|}{ nd } + \gamma\Delta\cdot\sqrt{|A|\cdot |I_t|}\\
&\leq (1+o(1))\cdot\frac{\Delta\cdot|A|}{\log n} + \gamma\Delta\cdot\sqrt{|A|\cdot |I_t|},
\end{align*}
which implies $|A|\leq \gamma\cdot |I_t|$.

Now define $B\triangleq N(I_t)\setminus I_t\setminus A$. We have
\[
e(B, I_t)= e(N(I_t), I_t) - e(A, I_t) \geq \left(
1-\frac{1}{\log n} -o(1) -\gamma \right)\Delta\cdot |I_t|.
\]
With the above estimate at hand, we compute the expected value of $|I_t\cap B|$. Note that for any $u\in B$,
the chance that it gets informed in round $t+1$ is
$$
p_{t+1}(u)\triangleq\Pro{\bigvee_{v\in N(u)\cap I_t} \left(X_{v\rightarrow u}=1\right)},
$$
which is lower bounded by
$$
\sum_{v\in N(u)\cap I_t} \Pro{X_{v\rightarrow u}=1}
-\sum_{\substack{v_1,v_2\in N(u)\cap I_t\\ v_1 < v_2}} \Pro{\bigwedge_{i=1,2}\left(X_{v_i\rightarrow u}=1\right)}
$$
by Bonferroni inequalities. Hence
\begin{align}
p_{t+1}(u)&\geq |N(u)\cap I_t|\left(\frac{1}{\Delta}-\epsilon\right)-{|N(u)\cap I_t|\choose 2}\left(\frac{1}{\delta^2}+\epsilon\right)\nonumber\\
&\geq (1-o(1))\cdot\frac{|N(u)\cap I_t|}{\Delta}-(1+o(1))\cdot{|N(u)\cap I_t|\choose 2}\cdot \frac{1}{\Delta^2}\nonumber\\
&\geq(1-o(1))\cdot\frac{|N(u)\cap I_t|}{\Delta}\left(1-\frac{(1+o(1))\cdot|N(u)\cap I_t|}{2\Delta}\right)\nonumber\\
&\geq (1-o(1))\cdot\frac{|N(u)\cap I_t|}{\Delta},\label{eq:boundpro}
\end{align}
where the first inequality follows from  \lemref{pairwise} and the fact that $\epsilon=(1/\Delta)^{\Theta(1)}$ is sufficiently small, and
the last step uses the condition that  $|N(u)\cap I_t|\leq 2d\sqrt{\gamma}=o(\Delta)$.
Hence we  have
\begin{align*}
\Ex{|I_{t+1} \setminus I_t| } &\geq
\Ex{|I_{t+1}\cap B|}
=\sum_{u\in B} p_{t+1}(u)
\geq \sum_{u\in B}(1-o(1))\cdot\frac{|N(u)\cap I_t|}{\Delta}\\
&=(1-o(1))\cdot\frac{e(B, I_t)}{\Delta}\geq (1-o(1))\cdot|I_t|.
\end{align*}

Since $|I_{t+1} \setminus I_{t}| \leq |I_{t}|$, it follows by
using Markov's inequality (applied to $|I_{t}| - |I_{t+1} \setminus I_t|$) that
$\Pro{ |I_{t+1}| \geq (2-f(n)) |I_t| } \geq 1-g(n)$, where $f(n)$ and $g(n)$ are
both functions that tend to zero. Hence the time to reach $|I_t| \geq n/ \log n$
can be upper bounded by the sum of $\log_{2-f(n)} n $
independent, identically distributed geometric random variables with
expectation at most $1-o(1)$ each. Using the Chernoff bound from
\lemref{geometric} yields for $\tau\triangleq\log_{2} n + o(\log n)$ that
$\Pro{ |I_{t+\tau}| > n/\log n} = 1-o(1)$.

\textbf{Phase \uppercase\expandafter{\romannumeral2}} $|I_t|\in[n/\log n,
n - n/\log n]$.
We further divide this phase into the two cases $|I_t| \in [n/\log n, n/2]$
and $|I_t| \in [n/2, n - n/\log n]$. We start with the first case $|I_t| \in
[n/\log n, n/2]$.

For any $u\in N(I_t)\setminus I_t$,
the probability $p_{t+1}(u)$ that $u$ gets informed in round $t+1$ is lowered bounded by
$$
(1-o(1))\cdot\frac{|N(u)\cap I_t|}{\Delta}\left(1-\frac{(1+o(1))\cdot|N(u)\cap I_t|}{2\Delta}\right)
$$
by the same argument as in \eq{boundpro}. This is then lower bounded by
$$
(1-o(1))\cdot\frac{|N(u)\cap I_t|}{2\Delta},
$$
since we have $|N(u)\cap I_t|\leq\Delta$.

By \eq{EMPexp}, we have
\[
e(I_t, U_t)=(1-o(1))\cdot \frac{\delta}{2}|I_t|.
\]
Similar to the analysis of Phase \uppercase\expandafter{\romannumeral1}, we
can lower bound the expected number of nodes that become informed in round
$t+1$:
\begin{align*}
\Ex{|I_{t+1}\setminus I_t|}&\geq \sum_{u\in N(I_t)\setminus I_t} p_{t+1}(u)\geq (1-o(1))\sum_{u\in N(I_t)\setminus I_t} \frac{|N(u)\cap I_t|}{2\Delta} \\
&=
(1-o(1))\frac{e(I_t, U_t)}{2\Delta} \geq \frac{\delta}{8 \Delta} |I_t|.
\end{align*}

Since $|I_{t+1}| \leq 2 |I_{t}|$, we obtain that as long as $|I_t| \leq
n/2$  there are constants $\alpha, \beta > 0$ so that $\Pro{ |I_{t+1}|
\geq (1+\alpha) |I_{t}| } \geq \beta$. Hence the time to reach $|I_t| \geq
n/2$
can be upper bounded by the sum of $\log_{1+\alpha} (\log n)$
independent, identically distributed geometric random variables with
expectation at most $1/\beta$ each. Using the Chernoff bound for the sum of
geometric random variables (see \lemref{geometric}) yields that with
probability $1-o(1)$, we reach $|I_t| \geq n/2$ within at most $o(\log n)$
additional rounds.

Consider now the case $|I_t| \in [n/2, n - n/\log n]$. To analyze this case,
we examine the shrinking of $U_t = V \setminus I_t$.
Note that for any $u\in U_t$,
the probability $p_{t+1}(u)$ that $u$ gets informed in round $t+1$ is lowered bounded by
$$
(1-o(1))\cdot\frac{|N(u)\cap I_t|}{\Delta}\left(1-\frac{(1+o(1))\cdot|N(u)\cap I_t|}{2\Delta}\right)
$$
by the same argument as in \eq{boundpro}. This is then lower bounded by
$$
(1-o(1))\cdot\frac{|N(u)\cap I_t|}{2\Delta}
$$
since we have $|N(u)\cap I_t|\leq\Delta$.

Again, as $|U_t| \leq
n/2$, by \eq{EMPexp}   we have
$$e(I_t, U_t)\geq  (1-o(1))\cdot \frac{\delta}{2}|U_t|.$$
Let us now compute the expected number of uninformed nodes after one
additional round:
\begin{align*}
 \Ex{|U_{t+1}|} &= \sum_{u \in U_t} (1-p_{t+1}(u))\leq |U_t| - (1-o(1)) \sum_{u \in U_t}\left( \frac{|N(u) \cap
I_t| }{2 \Delta} \right)\\
& = |U_t| - (1-o(1))\frac{e(I_t, U_t)}{2 \Delta} \leq \left(1 -
\frac{\delta}{8 \Delta} \right) |U_t|.
\end{align*}
A simple inductive argument yields for any integer $\tau$ that,
\begin{align*}
 \Ex{|U_{t+\tau}|} &\leq \left(1 -
\frac{\delta}{8 \Delta} \right)^{\tau} |U_t|,
\end{align*}
so for $\tau\triangleq \log \log n / \log (1/(1 -
\frac{\delta}{8 \Delta}) ) + \omega(1)$, where $\omega(1)$
is an arbitrarily slow growing function, we have $\Ex{|U_{t+\tau}|} = o(n/\log
n)$. Hence by Markov's inequality, $\Pro{ |U_{t+\tau}| \geq n/\log n} = o(1)$.

\textbf{Phase \uppercase\expandafter{\romannumeral3}} $|I_t|\in[n-n/\log n,
n]$.
Again, we analyze the shrinking of the set $U_t$.
By \lemref{condenserp}, for at least $(1-\epsilon/2)$-fraction of the choices of $x_t$, it holds that
the size of $\{\Gamma(v,x_t): v\in N(u)\cap I_t\}$ is at least $(1-\epsilon/2)|N(u)\cap I_t|$.
From now on fix $x_t$ such that this event occurs.

For any $u\in U_t$, we have
\[
\Pro{ u \notin I_{t+1}}
=\Pro{ \bigwedge_{v\in N(u)\cap I_t } (X_{v\rightarrow u}=0)}.
\]
Let $F$ be a subset of $N(u)\cap I_t$ of size $(1-\epsilon/2)|N(u)\cap I_t|$ such that
the map $\Gamma(\cdot, x_t)$ is injective when restricted to $F$.
By \lemref{rectanglemodulo}, the function $y\mapsto \left(\gen_{\Gamma(v,x_t)}(y)\bmod \deg(v)\right)_{v\in F}$ is an $(\epsilon'+|F|\Delta/m)$-\textsf{PRG} for $\mathsf{CR}_S$ where $S=\prod_{v\in F}[\deg(v)]$.

Then we have
\begin{align*}
\Pro{ u \notin I_{t+1}}
&\leq\Pro{ \bigwedge_{v\in F} (X_{v\rightarrow u}=0)}\leq \prod_{v\in F}\Pro{ X_{v\rightarrow u}=0 } +\epsilon'+|F|\Delta/m\\
&\leq \prod_{v\in F}\left(1-\frac{1}{\deg(v)}+\epsilon\right)+\epsilon'+\Delta^2/m\\
&\leq \left(1 - \frac{1}{\Delta} +\epsilon\right)^{ (1-\epsilon/2)|N(u)\cap I_t|}+\epsilon'+\Delta^2/m,
\end{align*}
where the second inequality follows from the properties of \PRG\  for combinatorial rectangles, and the third inequality follows from using pairwise independent generators.
Since $\epsilon\leq \frac{1}{\Delta}$, a simple induction shows that
$$
\left(1 - \frac{1}{\Delta} +\epsilon\right)^k
\leq \left(1 - \frac{1}{\Delta}\right)^k + k\epsilon
$$
for any $k\geq 0$.
So we have
\begin{align*}
\Pro{ u \notin I_{t+1}}
&\leq \left(1 - \frac{1}{\Delta}\right)^{ (1-\epsilon/2)|N(u)\cap I_t|} + (1-\epsilon/2)\cdot|N(u)\cap I_t|\cdot\epsilon+\epsilon'+\Delta^2/m\\
&\leq \left(1 - \frac{1}{\Delta}\right)^{ (1-\epsilon/2)|N(u)\cap I_t|} +(1-\epsilon/2)\cdot\Delta\cdot\epsilon+\epsilon'+\Delta^2/m.
\end{align*}
The bound above applies for any choice of $x_t$ such that the size of $\{\Gamma(v,x_t): v\in N(u)\cap I_t\}$ is at least $(1-\epsilon/2)|N(u)\cap I_t|$. And the probability of choosing such $x_t$ is at least $1-\epsilon/2$. So for random $x_t$, we have
\begin{align*}
\Pro{ u \notin I_{t+1}}&\leq \left(1 - \frac{1}{\Delta}\right)^{ (1-\epsilon/2)\cdot|N(u)\cap I_t|} +(1-\epsilon/2)\cdot\Delta\cdot\epsilon+\epsilon'+\Delta^2/m+\epsilon/2\\
&\leq \left(1 - \frac{1}{\Delta}\right)^{ (1-\epsilon/2)\cdot|N(u)\cap I_t|} + o(1),
\end{align*}
where we use the fact that $\epsilon=(1/\Delta)^{\Theta}$ is sufficiently small, and $m=\Theta((\log n)/\epsilon)$.

By \eq{lb_edgeset} it holds that  $e(I_t, U_t) \geq
(1-\frac{1}{\log n}-o(1)) \cdot \Delta  |U_t|$.
Let $A \subseteq U_t$ be  the set of nodes $v$ for which $ |N(v)\cap I_t| \leq
(1-\sqrt{\gamma}/2) \cdot \Delta$, where $\gamma\triangleq\frac{1}{\log n}+ o(1)$. We
assume for a contradiction
that $|A| > 2\sqrt{\gamma} \cdot |U_t|$. Hence,
\begin{align*}
 e(I_t, U_t) &= \sum_{v \in A} |N(v)\cap I_t| + \sum_{v \in U_t \setminus A} |N(v)\cap I_t| \leq |A| \cdot (1-\sqrt{\gamma}/2)  \Delta + |U_t \setminus A| \Delta \\
		&= |U_t|  \Delta - |A| \sqrt{\gamma}\Delta/2 < \left(1-\frac{1}{\log n}-o(1)\right)\cdot\Delta |U_t|,
\end{align*}
which yields the desired contradiction.
Hence $|A|\leq 2\sqrt{\gamma}|U_t|$. Now define $B\triangleq U_t \setminus A$ so that for each $u \in B$, $ |N(v)\cap I_t| > (1-\sqrt{\gamma}/2) \Delta$ and $|B| \geq (1 - 2\sqrt{\gamma}) |U_t|$.
Using linearity of expectation,
\begin{align*}
 \Ex{|U_{t+1}|} &\leq \sum_{u \in B} \Pro{ u \notin I_{t+1}} + \sum_{u \in A} \Pro{ u \notin I_{t+1}} \\
 &\leq \sum_{u \in B} \left(\left(1 - \frac{1}{\Delta}\right)^{(1-\epsilon/2) |N(u)\cap I_t|} + o(1)\right) + \sum_{u \in A} 1\\
&\leq  \sum_{u \in B} \left(1 - \frac{1}{\Delta}\right)^{(1-\epsilon/2) |N(u)\cap I_t|}
+ o(|U_t|) + |A|\\
&= \sum_{u \in B} \left(1 - \frac{1}{\Delta}\right)^{(1-\epsilon/2) |N(u)\cap I_t|}
+ o(|U_t|).
\end{align*}
Using the inequalities that $(1-1/k)\leq\mathrm{e}^{-1/k}$ for $k\geq 1$, $\mathrm{e}^x\leq 1+2x$ for sufficiently small constant $x>0$, and the condition that $|N(u)\cap I_t|\geq (1-\sqrt{\gamma}/2) \cdot \Delta$ for $u\in B$, we get
\begin{align*}
 \Ex{|U_{t+1}|}&\leq \sum_{u \in B} \mathrm{e}^{- (1-\epsilon/2)|N(u)\cap I_t|/\Delta} + o(|U_t|) \leq \sum_{u \in B} \mathrm{e}^{-(1-\sqrt{\gamma}/2-o(1))} + o(|U_t|)\\
 &=\sum_{u \in B} \mathrm{e}^{-1}\cdot\mathrm{e}^{\sqrt{\gamma}/2+o(1)} + o(|U_t|)\leq \sum_{u \in B} \mathrm{e}^{-1} \cdot (1+\sqrt{\gamma}+o(1)) + o(|U_t|) \\
&= (1+o(1))\cdot\mathrm{e}^{-1}\cdot |U_t|.
\end{align*}

By induction, it  follows that for any step $\tau >
0$,
$\Ex{|U_{t+\tau}|} \leq ( (1+o(1))\cdot\mathrm{e}^{-1})^{\tau} \cdot |U_t|$.
We choose $\tau \triangleq -\log_{ (1+o(1))\cdot\mathrm{e}^{-1}}(n)=\ln n+o(\log n)$ and obtain
that $\Ex{|U_{t+\tau}|} \leq (1/\log n)$. So
$\Pro{ |U_{t+\tau}| \geq 1} \leq \Ex{ |U_{t+\tau}| } \leq  1/\log n$.
\end{proof}

\end{document}